\documentclass[11pt]{article}
\usepackage[letterpaper,margin=1in]{geometry}
\usepackage[T1]{fontenc}
\usepackage{lmodern,microtype}
\usepackage{amsmath,amssymb,amsthm,mathtools}
\usepackage{enumitem,booktabs,array}
\usepackage[authoryear,round]{natbib}
\usepackage[colorlinks=true,urlcolor=blue,linkcolor=black,citecolor=black]{hyperref}
\usepackage{tcolorbox}
\usepackage{xcolor}
\usepackage{tikz}
\usetikzlibrary{decorations.pathreplacing}
\usetikzlibrary{arrows.meta,calc,fit,backgrounds}
\usepackage{pgfplots}

\tikzset{
  flowbox/.style={draw=black!70, rounded corners=2pt, align=flush center,
    inner sep=3pt, line width=0.5pt,
    execute at begin node={\hyphenpenalty=10000\exhyphenpenalty=10000}},
  flowarrow/.style={-{Stealth[length=2.2mm,width=1.6mm]}, line width=0.6pt,
    draw=black!75},
  flowlabel/.style={font=\scriptsize\itshape, text=black!65, align=flush center},
}

\hypersetup{pdftitle={Order-selection Prophet Inequality (Almost) Goes I.I.D.}}
\newtheorem{theorem}{Theorem}[section]
\newtheorem{proposition}[theorem]{Proposition}
\newtheorem{lemma}[theorem]{Lemma}
\theoremstyle{definition}
\newtheorem{definition}[theorem]{Definition}

\theoremstyle{remark}
\newtheorem*{remark}{Remark}
\newcommand{\Pp}{\mathbb P}
\newcommand{\E}{\mathbb E}
\newcommand{\R}{\mathbb R}
\newcommand{\dd}{\,\mathrm d}
\newcommand{\ind}{\mathbf 1}
\newcommand{\ALG}{\mathrm{ALG}}
\newcommand{\OPT}{\mathrm{OPT}}
\newcommand{\pos}[1]{\left(#1\right)_+}
\setlist[enumerate]{leftmargin=*,itemsep=3pt,topsep=5pt}
\allowdisplaybreaks[1]
\title{Order-selection Prophet Inequality (Almost) Goes I.I.D.}
\author{
  Bo Peng\thanks{Email: \texttt{ahqspbo@gmail.com}.}\\
  Tel Aviv University
  \and
  Zhihao Gavin Tang\thanks{Email:
    \texttt{tang.zhihao@mail.shufe.edu.cn}.}\\
  Shanghai University of Finance and Economics
}
\date{}
\begin{document}
\maketitle

\begin{abstract}
We study the order-selection prophet inequality. A gambler knows the distributions of
independent nonnegative random values, chooses the order in which to inspect them, and may
irrevocably accept at most one value, at the time it is revealed. Since the order is irrelevant when the
values are i.i.d., no algorithm can guarantee more than the optimal i.i.d.\ ratio
$\Gamma_\star\approx0.74544$ of the expected maximum. We give an algorithm with competitive
ratio $0.74500001$, less than $4.5\times10^{-4}$ below this bound. The algorithm samples
independent arrival times and uses item-dependent thresholds. It is designed so that,
conditional on its value, every item is selected with probability equal to a fixed fraction
of its probability of being the maximum. We show that the constructions of this form that
apply one local rule to every item are exactly those given by an allocation kernel, that is,
by a coupling of two uniform random variables. Requiring every rarely winning item to follow
the optimal i.i.d.\ threshold path singles out a baseline kernel, for which we prove
analytically a ratio $\Gamma_0\approx0.7389$, defined through a one-dimensional control
problem. Our final kernel perturbs this baseline, and its guarantee combines analytic
estimates with a rigorous computer-assisted certificate.
\end{abstract}

\section{Introduction}\label{sec:intro}

Prophet inequalities compare a gambler, who observes independent random values one at a
time and must accept or reject each value as soon as it is revealed, with a prophet, who
sees all values in advance and takes the maximum. For values revealed in an adversarial
order, \citet{KS77,KS78} showed that the gambler can obtain half of the expected maximum,
and \citet{SC84} showed that a single threshold suffices; the factor $1/2$ is best possible.
Through the correspondence between stopping rules and posted prices
\citep{HKS07,CHMS10,CFPV19}, prophet inequalities have become a standard tool for analyzing
posted-price mechanisms and online allocation; see the surveys of \citet{HK92},
\citet{Lucier17}, and \citet{CFHOV18}.

In the \emph{order-selection prophet inequality}, a gambler faces independent nonnegative
random values $V_1,\ldots,V_n$ with known distributions. Before observing any value, the
gambler chooses the order in which the values will be inspected. Each value must be accepted
or rejected irrevocably when it is revealed, and at most one value can be accepted. An
algorithm is \emph{$\Gamma$-competitive} if its expected reward is at least
$\Gamma\,\E[\max_iV_i]$ on every input. The freedom to choose the order can only help, but a
useful ordering rule must weigh the chance that an item is the maximum against the effect of
inspecting it on the opportunities that remain. The model goes back to \citet{Hill83}, who
showed that an order fixed in advance is as good as an adaptively chosen one and that the
freedom to choose the order can at most double the optimal expected reward.

When the values are i.i.d., the order is irrelevant. The optimal competitive ratio for
i.i.d.\ values, $\Gamma_\star\approx0.74544$ \citep{HK82,Kertz86,CFHOV21}, is therefore also
an upper bound for order selection, and no better upper bound is known. Lower bounds have
improved steadily. Taking each item's threshold from the ex-ante relaxation of the expected
maximum and inspecting the items in decreasing order of their thresholds gives the ratio
$1-1/e$ \citep{CHMS10,Yan11}; \citet{BGLPS21} improved it to about $0.654$, and the
$0.669$-competitive algorithm of \citet{CSZ21} for a uniformly random order applies as well. \citet{PT} proposed to sample an
independent arrival time for every item and to inspect the items in order of arrival; they
obtained a competitive ratio of about $0.7251$. \citet{BC} let the acceptance thresholds
depend on the item and improved the ratio to $0.7258$. This separated order selection from
random arrival, for which no algorithm exceeds $0.7254$ \citep{BC}. We give an algorithm with
competitive ratio $\Gamma_1=0.74500001$, less than $4.5\times10^{-4}$ below the upper bound
$\Gamma_\star$. Table~\ref{tab:known-bounds} summarizes the best known bounds.

\begin{table}[t]
\centering
\small
\renewcommand{\arraystretch}{1.15}
\begin{tabular}{@{}l >{\raggedright\arraybackslash}p{5.6cm} >{\raggedright\arraybackslash}p{5.0cm}@{}}
\toprule
Arrival order & Best competitive ratio & Best upper bound\\
\midrule
Adversarial & $1/2$ \citep{KS77,SC84} & $1/2$\\
Uniformly random & $0.688$ \citep{CHLT25} & $0.7235$ \citep{GMS23}\\
Order selection & $0.7258$ \citep{BC};\newline $0.74500001$ (Theorem~\ref{thm:certified-kernel}) &
  $\Gamma_\star\approx0.74544$ \citep{HK82,Kertz86}\\
I.i.d.\ values & $\Gamma_\star\approx0.74544$ \citep{CFHOV21} & $\Gamma_\star$ \citep{HK82,Kertz86}\\
\bottomrule
\end{tabular}
\caption{Best known bounds on the competitive ratio of single-item prophet inequalities for
independent values, by arrival order. The four settings are ordered by the power of the
gambler, and every upper bound for i.i.d.\ values applies to the other three settings.}
\label{tab:known-bounds}
\end{table}

\subsection{Our Approach and Results}\label{sec:approach}

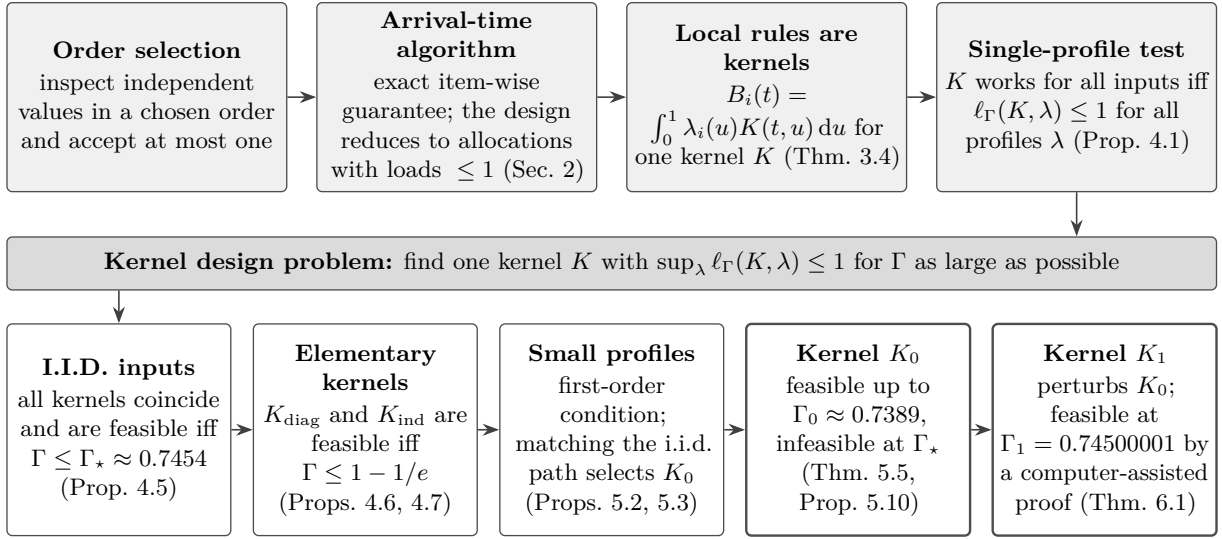
\begin{figure}[t]
\centering
\begin{tikzpicture}[font=\footnotesize]
  \def\wa{34.9mm}\def\wb{27.5mm}
  \node[flowbox, fill=black!6, text width=\wa, minimum height=25mm] (a1) at (18.5mm,0)
    {\textbf{Order selection}\\[1pt] inspect independent values in a chosen order and accept
     at most one};
  \node[flowbox, fill=black!6, text width=\wa, minimum height=25mm] (a2) at (59.5mm,0)
    {\textbf{Arrival-time algorithm}\\[1pt] exact item-wise guarantee; the design reduces
     to allocations with loads ${}\le1$ (Sec.~\ref{sec:model})};
  \node[flowbox, fill=black!6, text width=\wa, minimum height=25mm] (a3) at (100.5mm,0)
    {\textbf{Local rules are kernels}\\[1pt] $B_i(t)=\int_0^1\lambda_i(u)K(t,u)\dd u$
     for one kernel $K$ (Thm.~\ref{thm:local})};
  \node[flowbox, fill=black!6, text width=\wa, minimum height=25mm] (a4) at (141.5mm,0)
    {\textbf{Single-profile test}\\[1pt] $K$ works for all inputs iff
     $\ell_\Gamma(K,\lambda)\le1$ for all profiles~$\lambda$
     (Prop.~\ref{prop:scalar-feasibility})};
  \foreach \i/\j in {1/2,2/3,3/4} { \draw[flowarrow] (a\i) -- (a\j); }
  \node[flowbox, fill=black!14, text width=157.9mm, minimum height=7mm] (bar) at (80mm,-22mm)
    {\textbf{Kernel design problem:} find one kernel $K$ with
     $\sup_{\lambda}\ell_\Gamma(K,\lambda)\le1$ for $\Gamma$ as large as possible};
  \draw[flowarrow] (a4.south) -- (a4.south |- bar.north);
  \foreach \i/\x in {1/14.8, 2/47.4, 3/80, 4/112.6, 5/145.2} { \coordinate (B\i) at (\x mm, -44mm); }
  \node[flowbox, text width=\wb, minimum height=28mm] (b1) at (B1)
    {\textbf{I.I.D.\ inputs}\\[1pt] all kernels coincide and are feasible iff
     $\Gamma\le\Gamma_\star\approx0.7454$ (Prop.~\ref{prop:iid-guarantee})};
  \node[flowbox, text width=\wb, minimum height=28mm] (b2) at (B2)
    {\textbf{Elementary kernels}\\[1pt] $K_{\mathrm{diag}}$ and $K_{\mathrm{ind}}$ are
     feasible iff $\Gamma\le1-1/e$ (Props.~\ref{prop:diagonal},~\ref{prop:independent})};
  \node[flowbox, text width=\wb, minimum height=28mm] (b3) at (B3)
    {\textbf{Small profiles}\\[1pt] first-order condition; matching the i.i.d.\ path
     selects $K_0$ (Props.~\ref{prop:iid-calibration},~\ref{prop:iid-path-kernel})};
  \node[flowbox, line width=0.9pt, text width=\wb, minimum height=28mm] (b4) at (B4)
    {\textbf{Kernel $K_0$}\\[1pt] feasible up to $\Gamma_0\approx0.7389$, infeasible at
     $\Gamma_\star$ (Thm.~\ref{thm:kzero-analytic},
     Prop.~\ref{prop:k0-counterexample})};
  \node[flowbox, line width=0.9pt, text width=\wb, minimum height=28mm] (b5) at (B5)
    {\textbf{Kernel $K_1$}\\[1pt] perturbs $K_0$; feasible at $\Gamma_1=0.74500001$ by a
     computer-assisted proof (Thm.~\ref{thm:certified-kernel})};
  \draw[flowarrow] (bar.south -| b1.north) -- (b1.north);
  \foreach \i/\j in {1/2,2/3,3/4,4/5} { \draw[flowarrow] (b\i) -- (b\j); }
\end{tikzpicture}
\caption{Overview of the paper. The top row reduces the design of a local arrival-time
algorithm to the choice of a single allocation kernel (Sections~\ref{sec:model}--\ref{sec:kernel-design}). The
bottom row shows how the kernel is designed and the ratio guaranteed at each stage
(Sections~\ref{sec:kernel-design}--\ref{sec:final-kernel}).}
\label{fig:roadmap}
\end{figure}

\begin{table}[t]
\centering
\small
\begin{tabular}{@{}l l l l@{}}
\toprule
Kernel & Coupling & Universally feasible at $\Gamma$ & Reference\\
\midrule
any kernel & --- & only if $\Gamma\le\Gamma_\star\approx0.74544$ &
  Prop.~\ref{prop:iid-guarantee}\\
$K_{\mathrm{diag}}$ & $\widetilde T=\widetilde U$ & if and only if $\Gamma\le1-1/e\approx0.632$ &
  Prop.~\ref{prop:diagonal}\\
$K_{\mathrm{ind}}$ & $\widetilde T$, $\widetilde U$ independent & if and only if
  $\Gamma\le1-1/e$ & Prop.~\ref{prop:independent}\\
$K_0$ & $\widetilde T=\phi(X)$, $X\sim\mathrm U[\widetilde U,1]$ & if $\Gamma\le\Gamma_0\approx0.7389$; not if
  $\Gamma\ge0.745$ & Thm.~\ref{thm:kzero-analytic}, Prop.~\ref{prop:k0-counterexample}\\
$K_1$ & perturbed $K_0$ coupling & if $\Gamma\le\Gamma_1=0.74500001$ &
  Thm.~\ref{thm:certified-kernel}\\
\bottomrule
\end{tabular}
\caption{The allocation kernels studied in this paper. A kernel is universally feasible
at $\Gamma$ if the arrival construction based on it is well defined, and hence
$\Gamma$-competitive, on every input (Proposition~\ref{prop:scalar-feasibility}). Here
$(\widetilde U,\widetilde T)$ denotes the coupling of two uniform random variables that
defines the kernel (Section~\ref{sec:characterization}); for $K_0$, the release level $X$ is
uniform on $[\widetilde U,1]$ given $\widetilde U$, and $\phi$ is defined in~\eqref{eq:iid-clock}.}
\label{tab:results}
\end{table}

\paragraph{The design problem.}
We work in the independent-arrival framework of \citet{PT}, in the item-dependent form of
\citet{BC}. Each item receives an independently sampled arrival time and a threshold that
depends on that time. We design the arrival laws and thresholds so that, conditional on its
value, each item is accepted with probability exactly $\Gamma$ times its probability of being
the maximum. This identity alone gives $\E[\ALG]=\Gamma\E[\OPT]$. The design problem can then
be expressed through \emph{area allocations}. The allocation of an item determines its
threshold path, which in turn determines the total arrival probability, or \emph{load}, that
the item needs to sustain the guarantee. The allocations of all items must add up to
normalized time, and every load must be at most one. The question is how to choose
allocations that satisfy both constraints for all independent inputs simultaneously.

\paragraph{Local constructions and allocation kernels.}
We study \emph{local constructions}: a single rule assigns to every item an allocation that
depends only on the item's own area curve. To describe such rules, we represent each item by
its \emph{winning-probability profile}, the probability that the item is the maximum given the
rank of the maximum. The profiles of all items sum to one at almost every rank.
Theorem~\ref{thm:local} shows that local constructions correspond exactly to \emph{allocation
kernels}, the conditional distribution functions of couplings of two uniform random variables:
the allocation of an item is the integral of its profile against the kernel. The
total-allocation constraint then holds automatically, and the same kernel serves every item
and every input. The remaining constraint, that every load is at most one, becomes a test on
single profiles (Proposition~\ref{prop:scalar-feasibility}).

\paragraph{From the i.i.d.\ benchmark to a baseline kernel.}
The i.i.d.\ case is a natural benchmark, but it does not help to choose a kernel: every item
then has the constant profile $1/n$, and all kernels induce the same allocation.
Proposition~\ref{prop:iid-guarantee} shows that this common construction attains
$\Gamma_\star$ on i.i.d.\ inputs. In contrast, the two elementary kernels of
Section~\ref{sec:kernel-design}, given by the perfectly correlated and the independent couplings, are
universally feasible only up to $1-1/e\approx0.632$, and their worst cases are nonconstant
profiles. To obtain a more informative test, we examine profiles with small winning
probability and arbitrary shape. Their limiting threshold paths yield a first-order necessary
condition for feasibility at $\Gamma_\star$, and requiring all of these paths to coincide with
the i.i.d.\ path selects a unique kernel $K_0$. Theorem~\ref{thm:kzero-analytic} proves that
$K_0$ is universally feasible at the exact root $\Gamma_0\approx0.7389$ of an equation given by
a one-dimensional control problem, and Proposition~\ref{prop:k0-counterexample} exhibits a
profile on which $K_0$ fails at $\Gamma_\star$. Matching the i.i.d.\ path for every small
profile therefore does not control profiles with substantial winning probability.

\paragraph{A certified kernel above $0.745$.}
Our final kernel $K_1$ perturbs the coupling behind $K_0$ on a coarse grid while preserving
both marginals; the perturbation is specified by $69$ integers.
Theorem~\ref{thm:certified-kernel} proves that $K_1$ is universally feasible at
$\Gamma_1=0.74500001$, which yields an order-selection algorithm with competitive ratio
$\Gamma_1$ for arbitrary independent nonnegative values. The proof follows the analysis of
$K_0$ step by step (Figure~\ref{fig:proof-architecture}). The perturbation changes the
allocation of a profile only through four block masses. Analytic estimates bound the threshold path and the
load of every profile whose block masses lie in a given region, without restricting the
profile within the blocks, and a finite interval-arithmetic certificate checks these bounds
over a cover of the four-dimensional parameter space. The certificate accounts for
quadrature, interpolation, tail, and rounding errors.

\paragraph{Contributions.}
In summary, we (i)~cast the arrival-time approach as the design of allocations under an exact
item-wise guarantee; (ii)~characterize all local constructions by allocation kernels;
(iii)~derive necessary conditions from small profiles, which lead to the baseline kernel
$K_0$ and its analytic guarantee; and (iv)~construct an explicit kernel $K_1$ with a certified
competitive ratio above $0.745$. Table~\ref{tab:results} lists the kernels studied in the
paper and their guarantees, and Figure~\ref{fig:roadmap} shows how the results fit together.

\subsection{Related Work}\label{sec:related}

\paragraph{Adversarial order and i.i.d.\ values.}
For independent values in a fixed order, the factor $1/2$ of \citet{KS77,KS78} and
\citet{SC84} is tight. For $n$ i.i.d.\ values, \citet{HK82} identified the smallest constant
$a_n$ such that the expected maximum is at most $a_n$ times the optimal stopping value, and
\citet{Kertz86} showed that $a_n$ converges to $1/\Gamma_\star\approx1.3415$. \citet{AEEHKL17}
gave an algorithm with ratio $0.738$ for i.i.d.\ values, and \citet{CFHOV21} obtained the
optimal ratio $\Gamma_\star$ with a sequence of quantile thresholds; \citet{LPPSS21} gave a
primal-dual proof of the bound of \citet{Kertz86}, and \citet{PT} recovered $\Gamma_\star$
within their arrival-time framework. \citet{Allaart07} showed that the same constant is the
sharp asymptotic bound when i.i.d.\ values arrive at the times of a Poisson process. Other
work assumes that the distributions are known only through samples
\citep{AKW14,CDFS22,RWW20}, or refines the worst-case guarantees under additional structure,
such as a bound on the relative variance of the maximum \citep{LMPPV26}.

\paragraph{Random arrival order.}
In the \emph{prophet secretary} problem, introduced by \citet{EHLM17}, the values arrive in a
uniformly random order. \citet{EHLM17} obtained the ratio $1-1/e$, which is also attained by a
single threshold with randomized tie-breaking \citep{EHKS18} and by time-invariant
personalized thresholds \citep{CFHOV21}. \citet{ACK18} gave the first improvement, to
$1-1/e+1/400$. \citet{CSZ21} obtained $0.669$ with blind strategies and showed that blind
strategies cannot exceed $0.675$; \citet{Harb25} improved the ratio to $0.6724$ using
Poissonization and sharding, and \citet{CHLT25} obtained the current best ratio, $0.688$, by
treating the largest item separately. When every distribution occurs sufficiently often,
\citet{AEEHKL17} obtained $0.738$ and \citet{LPPSS21} obtained ratios approaching
$\Gamma_\star$. Every guarantee for random order also
applies to order selection, because the gambler may choose a uniformly random order. On the
hardness side, \citet{CSZ21} showed that no algorithm exceeds $\sqrt3-1\approx0.732$, which
separated random order from the i.i.d.\ setting. \citet{BC} lowered this bound to $0.7254$,
even for algorithms that know the order in advance, and \citet{GMS23} proved the bound $0.7235$
without computer assistance for algorithms that learn the order only as it unfolds.

\paragraph{Order selection and optimal ordering.}
After \citet{Hill83}, \citet{HH85} and \citet{Gilat87} studied which orders are optimal for a
given set of distributions. In the competitive setting, the threshold rule of \citet{CHMS10}
and \citet{Yan11} attains $1-1/e$, which is the best possible against the ex-ante relaxation
used in its analysis, so later improvements compare directly with the expected maximum.
\citet{BGLPS21} obtained $1-1/e+0.022\approx0.654$, and \citet{AEEHKL17} obtained $0.738$
when every distribution occurs sufficiently often. \citet{LPPSS21} showed that, after a
constant number of distributions is removed from the instance, a uniformly random order
attains a ratio arbitrarily close to $\Gamma_\star$ with respect to the expected maximum of the
remaining values. \citet{PT} introduced arrival-time design and obtained $0.7251$, and
\citet{BC} extended it with item-dependent thresholds to obtain $0.7258$. Their constructions
are not local in the sense of Section~\ref{sec:local}; we instead restrict attention to local
constructions and optimize over the allocation kernel that describes them. \citet{ADK21}
study an intermediate model in which the gambler chooses the order from a prescribed set of
permutations.

\paragraph{Pricing and combinatorial extensions.}
Order selection corresponds to sequential posted pricing, in which the seller chooses the
order in which buyers are approached \citep{CHMS10,Yan11,BGLPS21}, and \citet{CFPV19} showed
that guarantees for threshold stopping rules and for posted prices are equivalent under each
of these arrival models. Beyond a single item, prophet inequalities are known for $k$ units
\citep{HKS07,Alaei14}, matroids \citep{KW19}, general downward-closed constraints
\citep{Rubinstein16,RS17}, and combinatorial auctions \citep{DFKL20}, and they are closely
related to online contention resolution schemes \citep{FSZ21,LS18}. Prophet inequalities
are also known for matchings \citep{EFGT22}.

\paragraph{Online benchmarks.}
Another line of work measures algorithms against the best online algorithm instead of the
prophet. For order selection, this benchmark is the optimal policy, which chooses both the
inspection order and the stopping rule. Computing it is NP-hard, even for distributions with
three-point supports \citep{ASZ20}, and approximation schemes for it and for closely related
stochastic optimization problems are given by \citet{CEGMM10}, \citet{FLX18}, \citet{LPPSS21},
and \citet{SS21}; these results concern computational efficiency rather than the competitive
ratio against the prophet. For single-item selection, the \emph{order-competitive ratio} of
\citet{EFGT23} compares the best algorithm that does not know the arrival order with the best
algorithm that does. Its optimal value for deterministic algorithms is
$(\sqrt5-1)/2\approx0.618$ \citep{EFGT23,CEFGST26}; randomized algorithms achieve $0.732$, and
no algorithm exceeds $0.829$ \citep{CST24}. \citet{EG23} study this ratio under combinatorial
constraints. For a uniformly random order, the optimal online algorithm can be approximated
within any factor $1-\varepsilon$ \citep{DGRTT23}, and for bipartite matching it can be
approximated better than the best possible guarantee against the prophet \citep{PPSW24}.

\subsection{Organization}\label{sec:organization}
Section~\ref{sec:model} introduces the model and the arrival-time framework.
Section~\ref{sec:local} characterizes local constructions by allocation kernels.
Section~\ref{sec:kernel-design} formulates the kernel design problem and analyzes the i.i.d.\
benchmark and the two elementary kernels. Section~\ref{sec:k0} derives and analyzes the
baseline kernel $K_0$, Section~\ref{sec:final-kernel} specifies the final kernel $K_1$
and states its guarantee, and Section~\ref{sec:conclusion} concludes.
Appendices~\ref{app:algorithm}--\ref{app:kzero-proofs} contain the
deferred proofs of Sections~\ref{sec:model}--\ref{sec:k0}.
Appendices~\ref{app:cert-proof} and~\ref{app:cert-details} prove
Theorem~\ref{thm:certified-kernel}, and Appendix~\ref{app:cert-atoms} extends the guarantee
to distributions with atoms.

\section{Model and Algorithmic Framework}\label{sec:model}
We describe the model and the independent-arrival framework, then formulate the allocation problem that underlies our algorithm design.

Let $V_1,\ldots,V_n$ be independent nonnegative random values with known cumulative distribution functions $F_1,\ldots,F_n$.
The gambler chooses an inspection order before observing any values and may accept at most one item.
Let $I$ denote the index of the accepted item, $I^*$ the index of the maximum-value item, and $\OPT=\max_i V_i$.
If no item is accepted, set $I=\dagger$ and $\ALG=0$; otherwise, set $\ALG=V_I$.
For ease of exposition, we initially assume that the value distributions are atomless, so the maximum is unique almost surely; Appendix~\ref{app:cert-atoms} extends the final guarantee to distributions with atoms.

Let $\Gamma\in(0,1)$ denote the target competitive ratio.
The case $n=1$ is immediate, so we assume $n\ge2$.

\subsection{Independent Arrival Times with Item-Dependent Thresholds}

We use the algorithmic framework of \citet{BC}, which extends the independent-arrival framework of \citet{PT} by allowing item-dependent threshold functions.

\begin{tcolorbox}
\phantomsection\label{alg:independent-arrivals}
\begin{center}
\small\textbf{Algorithm: Independent Arrivals with Item-Dependent Thresholds}
\end{center}
For each item $i$, specify an arrival-time distribution $G_i$ on $[0,1]\cup\{\dagger\}$ and a threshold function $\tau_i:[0,1]\to\R$.
\begin{enumerate}
\item Sample the arrival times $T_i$ from their respective distributions $G_i$, independently of one another and of the values.
\item Discard all items with $T_i=\dagger$ and inspect the remaining items in increasing order of arrival time, using a fixed rule to break ties.
\item Accept the first item satisfying $V_i\ge\tau_i(T_i)$.
\end{enumerate}
\end{tcolorbox}

\paragraph{Design Principle: An Exact Item-Wise Guarantee.}
We design the arrival-time distributions and threshold functions to satisfy the following exact item-wise guarantee:
\begin{equation}\label{eq:itemwise}
 \Pp(I=i\mid V_i=v)
 =\Gamma\Pp(I^*=i\mid V_i=v)
 =\Gamma\prod_{j\ne i}F_j(v),
\end{equation}
for every item $i$ and $F_i$-almost every $v$.
This identity implies
\[
 \Pp(\ALG>v)=\Gamma\Pp(\OPT>v)
 \qquad\text{for every $v\ge0$},
\]
because integrating~\eqref{eq:itemwise} over $\{V_i>v\}$ and summing over $i$ gives
$\Pp(\ALG>v)=\sum_i\Pp(I=i,\,V_i>v)=\Gamma\sum_i\Pp(I^*=i,\,V_i>v)=\Gamma\Pp(\OPT>v)$.
Integrating over $v\ge0$ yields $\E[\ALG]=\Gamma\E[\OPT]$ whenever $\E[\OPT]<\infty$.

This property is implicit in the analyses of \citet{PT} and \citet{BC}.
We make it explicit as a design principle.
Although the framework involves two families of functions, $\{G_i\}$ and $\{\tau_i\}$, the construction below specifies the arrival-time distributions explicitly once the threshold functions are fixed.
The remaining design freedom therefore lies in the choice of threshold functions, which we encode using \emph{area-allocation functions}.

For completeness, Appendix~\ref{app:algorithm} derives the construction of Bubna and Chiplunkar from the item-wise guarantee.
We now express the construction in a probability-scale parameterization that will be central to our local allocation rules.

\subsection{A Probability-Scale Parameterization}

\paragraph{Item-Dependent Probability Thresholds.}
Rather than specifying the threshold functions $\tau_i$ on the value scale, we work on an item-dependent probability scale.
This parameterization is motivated by the item-wise guarantee, which compares $V_i$ with $\max_{j\ne i}V_j$.
For each item $i$, define
\[
 U_i:=\Pp(I^*\ne i\mid V_i)
     =1-\prod_{j\ne i}F_j(V_i).
\]
Thus, $U_i$ is the conditional probability that item $i$ is not the maximum-value item, given its realized value.
We call $U_i$ the \emph{loss score} of item $i$.
Since $U_i$ is a nonincreasing function of $V_i$, an acceptance condition of the form $U_i\le y$ is a threshold condition on $V_i$; the modified algorithm below is therefore an instance of the original one.

\begin{tcolorbox}
For each item $i$, specify a probability-threshold function $y_i:[0,1]\to[0,1]$ in place of $\tau_i$.
Replace Step~3 of the algorithm with the following:
\begin{enumerate}
\item[3.] Accept the first item satisfying $U_i\le y_i(T_i)$.
\end{enumerate}
\end{tcolorbox}

\paragraph{Area Curves and Allocations.}
For each item $i$, define its area curve by
\begin{equation}\label{eq:area}
 A_i(x):=\E[(x-U_i)_+]
        =\int_0^x\Pp(U_i<z)\dd z,
 \qquad x\in[0,1].
\end{equation}
The function $A_i$ is continuous, convex, nondecreasing, and $1$-Lipschitz.
Its value at $x=1$ satisfies
\begin{equation}\label{eq:mu}
 \mu_i:=A_i(1)=\E[1-U_i]=\Pp(I^*=i),
 \qquad \sum_i\mu_i=1.
\end{equation}
Define the upper generalized inverse of $A_i$ by
\[
 A_i^{\leftarrow}(b)
 :=\sup\{x\in[0,1]:A_i(x)\le b\},
 \qquad b\in[0,\mu_i].
\]
To specify the probability thresholds, choose \emph{area-allocation functions} $B_i:[0,1]\to[0,\mu_i]$ satisfying
\begin{equation}\label{eq:allocation}
\begin{gathered}
 B_i(0)=0,\qquad B_i(1)=\mu_i,\qquad
 B_i\text{ is nondecreasing},\\
 \sum_i B_i(t)=t\qquad\text{for all $t\in[0,1]$}.
\end{gathered}
\end{equation}
The quantity $B_i(t)$ is the area allocated to item $i$ by time $t$.
The normalization $\sum_i B_i(t)=t$ means that the total allocated area grows at unit rate.
These conditions also imply that each $B_i$ is $1$-Lipschitz.

\paragraph{Probability Thresholds and Arrival-Time Distributions.} 
Given these allocation functions, define
\begin{equation}\label{eq:arrivals}
\begin{aligned}
 y_i(t)&:=A_i^{\leftarrow}(B_i(t)),\\
 S(t)&:=1-\Gamma t,\\
 S_i(t)&:=\exp\!\left(-\Gamma\int_0^t
                     \frac{\dd B_i(s)}{S(s)}\right),\\
 \dd G_i(t)&:=\Gamma\frac{S_i(t)}{S(t)}\dd y_i(t), \qquad t\in(0,1].
\end{aligned}
\end{equation}
By Lemma~\ref{lem:threshold-paths} in Appendix~\ref{app:algorithm}, each threshold path $y_i$ is continuous and nondecreasing, with $A_i(y_i(t))=B_i(t)$ and $y_i(1)=1$.
Its starting point $y_i(0)$ is the right endpoint of the initial interval on which $A_i$ vanishes, so it may be positive.
The normalization $\sum_iB_i(t)=t$ gives $\prod_iS_i(t)=S(t)$.
Integrals with respect to $B_i$ and $y_i$ are understood in the Lebesgue--Stieltjes sense.
We assign no arrival mass to $t=0$.

Define the arrival load of item $i$ by
\begin{equation}\label{eq:load}
 L_i:=\mathcal L_\Gamma(A_i,B_i) = \int_{(0,1]} \dd G_i(t)
     =\Gamma\int_{(0,1]}\frac{S_i(t)}{S(t)}\dd y_i(t).
\end{equation}
Thus, $L_i$ is the total mass assigned to arrival times in $(0,1]$.
The construction is feasible if and only if $L_i\le1$ for every item $i$.
In that case, complete each arrival-time distribution by setting
\[
 \Pp(T_i=\dagger)=1-L_i.
\]

\begin{theorem}\label{thm:bc}
Suppose the allocation functions satisfy~\eqref{eq:allocation} and $L_i\le1$ for every item $i$.
Then, for every item $i$ and every $t\in[0,1]$, the resulting algorithm has the following properties:
\begin{itemize}
\item $\Pp\bigl(T_i\in(0,t],\ U_i\le y_i(T_i)\bigr)=1-S_i(t)$.
\item $\Pp(\text{no item is accepted by time }t)=S(t)=1-\Gamma t$.
\item $\Pp(\text{item }i\text{ is accepted by time }t)=\Gamma B_i(t)$.
\item The exact item-wise guarantee~\eqref{eq:itemwise} holds.
\end{itemize}
\end{theorem}

The probability of accepting an item by time $t$ is therefore $\Gamma t$, and the eventual acceptance probability is $\Gamma$.
Thus $t$ is a \emph{normalized algorithmic time}: it measures the fraction of the eventual acceptance probability accumulated by time $t$, not the fraction of items inspected.
The proof is given in Appendix~\ref{app:algorithm}.

\subsection{Design of Area-Allocation Functions}

Theorem~\ref{thm:bc} reduces the problem to finding area-allocation functions whose arrival loads are all at most one.
Our goal is to make $\Gamma$ as large as possible while ensuring that every input admits allocation functions $B_1,\ldots,B_n$ satisfying~\eqref{eq:allocation} and $\mathcal L_\Gamma(A_i,B_i)\le1$ for every item $i$.
The notation $\mathcal L_\Gamma(A_i,B_i)$ emphasizes that, for fixed $\Gamma$, the load depends only on item $i$'s area curve and allocation function.

Even after this reduction, considerable freedom remains in the choice of allocations.
\citet{PT} use a common value-threshold function and achieve a competitive ratio of approximately $0.7251$.
\citet{BC} improve this ratio to $0.7258$ using item-dependent thresholds and a careful case analysis.
We restrict attention to a structured family of allocations, the \emph{local constructions} introduced in the next section.

\section{Local Constructions of Area Allocations}\label{sec:local}

For fixed $\Gamma$, the arrival load $L_i=\mathcal L_\Gamma(A_i,B_i)$ depends only on item $i$'s area curve $A_i$ and its allocation $B_i$.
This suggests defining $B_i$ as a functional of $A_i$, with the same rule applied to every item in every input.
The challenge is to ensure that the resulting allocations automatically satisfy~\eqref{eq:allocation}.

Let $\mathcal A$ denote the class of item area curves arising from inputs in the model above.

\begin{definition}[Local construction]\label{def:local}
Fix $\Gamma\in(0,1)$.
A \emph{local construction} is a functional $\mathcal B$ that assigns to each $A\in\mathcal A$ a nondecreasing allocation
\[
 B_A:=\mathcal B[A]:[0,1]\to[0,A(1)]
\]
satisfying
\begin{equation}\label{eq:local-endpoints}
 B_A(0)=0,
 \qquad B_A(1)=A(1).
\end{equation}
For every input with area curves $A_1,\ldots,A_n$, the functional must also satisfy
\begin{equation}\label{eq:total-allocation}
 \sum_{i=1}^n\mathcal B[A_i](t)=t
 \qquad\text{for every }t\in[0,1].
\end{equation}
We call~\eqref{eq:total-allocation} the \emph{total-allocation identity}: at normalized time $t$, the total area allocated across all items is exactly $t$.
The functional $\mathcal B$ may depend on $\Gamma$, but not separately on the item label, the number of items, or the other area curves.
\end{definition}

A local construction satisfies the allocation constraints by definition.
It is \emph{feasible} on an input if, in addition, $\mathcal L_\Gamma(A_i,\mathcal B[A_i])\le1$ for every item $i$.
Locality is defined relative to the area curve, which already records how the item compares with its competitors, not to the item's own value distribution alone.

We show that local constructions are in one-to-one correspondence with couplings of two uniform distributions on $[0,1]$.
To express this characterization, we represent area curves by profiles that sum to one at almost every maximum rank. This pointwise identity will ensure the total-allocation identity~\eqref{eq:total-allocation}.

\subsection{Winning-Probability Profiles}\label{sec:profile}

Let $F_{\max}(v):=\prod_jF_j(v)$ and define the \emph{maximum rank}
\begin{equation}\label{eq:maximum-rank}
 U:=1-F_{\max}(\OPT)
   =1-\prod_jF_j(\OPT).
\end{equation}
Since $F_{\max}$ is continuous, $U$ is uniformly distributed on $[0,1]$.
In contrast, the loss scores $U_i$ need not be uniformly distributed.
Smaller ranks correspond to larger values of the maximum.

For each item $i$, define its \emph{winning-probability profile} $\lambda_i:[0,1]\to[0,1]$ by
\begin{equation}\label{eq:profile}
 \Pp(I^*=i,\ U\le u)
 =\int_0^u\lambda_i(a)\dd a,
 \qquad u\in[0,1].
\end{equation}
Such a function exists because the left side of~\eqref{eq:profile} grows by at most $\Pp(u<U\le u')=u'-u$ as $u$ increases to $u'$.
For almost every $u$, $\lambda_i(u)$ is the conditional probability that item $i$ attains the maximum, given that the maximum rank is $u$.
We refer to $\lambda_i$ simply as the item's \emph{profile}.

Since exactly one item wins, the profiles satisfy
\begin{equation}\label{eq:partition}
 0\le\lambda_i\le1,
 \qquad \sum_i\lambda_i=1\quad\text{a.e.},
 \qquad \int_0^1\lambda_i(u)\dd u=\mu_i.
\end{equation}

Let $\Lambda$ denote the set of measurable profiles $\lambda:[0,1]\to[0,1]$, identifying profiles that agree almost everywhere.
The following lemma formalizes the equivalence between profiles and area curves.

\begin{lemma}[Profile--area correspondence]\label{lem:profile-bijection}
There is a bijection
\[
 \Lambda\longrightarrow\mathcal A,
 \qquad \lambda\longmapsto A_\lambda,
\]
such that, for every input and every item $i$, the profile $\lambda_i$ corresponds to the area curve $A_i$.
In particular,
\[
 A_\lambda=A_{\widetilde\lambda}
 \quad\Longleftrightarrow\quad
 \lambda=\widetilde\lambda\ \text{a.e.},
 \qquad
 A_\lambda(1)=\int_0^1\lambda(u)\dd u.
\]

More explicitly, for $\lambda\in\Lambda$ and $0\le u<1$, define
\begin{equation}\label{eq:profile-coordinates-main}
\begin{aligned}
 h_\lambda(u)&:=\int_0^u\frac{\lambda(a)}{1-a}\dd a,
 & C_\lambda(u)&:=\int_0^u\lambda(a)\dd a,\\
 x_\lambda(u)&:=1-(1-u)e^{h_\lambda(u)},
 & \mu_\lambda&:=\int_0^1\lambda(a)\dd a.
\end{aligned}
\end{equation}
Then $x_\lambda$ is continuous and nondecreasing, with $0\le x_\lambda(u)\le u$, and
\begin{equation}\label{eq:profile-area-main}
\begin{aligned}
 A_\lambda(x_\lambda(u))
 &=C_\lambda(u)-(1-u)\bigl(e^{h_\lambda(u)}-1\bigr)\\
 &=C_\lambda(u)+x_\lambda(u)-u.
\end{aligned}
\end{equation}
Let $x_*:=\lim_{u\uparrow1}x_\lambda(u)$.
If $x_*<1$, the curve continues with slope one:
\begin{equation}\label{eq:profile-extension-main}
 A_\lambda(x)=\mu_\lambda+x-1,
 \qquad x\in[x_*,1].
\end{equation}
Together with continuity and $A_\lambda(1)=\mu_\lambda$, these formulas determine $A_\lambda$ on $[0,1]$.
\end{lemma}

Appendix~\ref{app:profile} proves the lemma and gives an explicit formula for recovering $\lambda$ from $A_\lambda$.
Consequently, we may equivalently express a local construction as a rule $\lambda\mapsto B_\lambda$, where $B_\lambda:=\mathcal B[A_\lambda]$.
This reparameterization preserves the information available to the rule.
The advantage of this representation is that an allocation kernel can transform the pointwise identity $\sum_i\lambda_i(u)=1$ almost everywhere into the allocation constraint $\sum_iB_i(t)=t$.
Throughout, $t\in[0,1]$ denotes normalized algorithmic time, and $u\in[0,1]$ denotes the maximum-rank coordinate associated with $U=1-F_{\max}(\OPT)$; the variable $x\in[0,1]$ usually denotes an argument of the area curves, that is, a level of the probability threshold.

\subsection{Characterization of Local Constructions}\label{sec:characterization}

\begin{definition}[Allocation kernel]
A jointly measurable function $K:[0,1]^2\to[0,1]$ is an \emph{allocation kernel} if the following conditions hold:
\begin{itemize}
\item for almost every $u$, $K(\cdot,u)$ is nondecreasing and right-continuous, with
\[
 K(0,u)=0,
 \qquad K(1,u)=1;
\]
\item for every $t\in[0,1]$,
\begin{equation}\label{eq:kernel-marginal}
 \int_0^1K(t,u)\dd u=t.
\end{equation}
\end{itemize}
\end{definition}
Equivalently, an allocation kernel is a conditional distribution function for an auxiliary pair $(\widetilde U,\widetilde T)$ with uniform marginals on $[0,1]$:
\[
 K(t,u)=\Pp(\widetilde T\le t\mid\widetilde U=u).
\]
The correspondence with couplings is one-to-one when kernels that agree for all $t$ outside a common null set of $u$ are identified.
For a fixed maximum rank $u$, $K(t,u)$ specifies the fraction of winning probability mass allocated by normalized algorithmic time $t$.
The same allocation kernel is used for every item and every input.
Given an allocation kernel $K$ and a profile $\lambda$, define
\begin{equation}\label{eq:kernel-B}
 B_{K,\lambda}(t):=\int_0^1\lambda(u)K(t,u)\dd u.
\end{equation}
The following theorem shows that this representation captures precisely the local constructions of Definition~\ref{def:local}.

\begin{theorem}[Characterization of local constructions]\label{thm:local}
A functional $\mathcal B$ is a local construction if and only if there exists an allocation kernel $K$ such that
\begin{equation}\label{eq:local-representation}
 \mathcal B[A_\lambda](t)
 =B_{K,\lambda}(t)
 =\int_0^1\lambda(u)K(t,u)\dd u
 \qquad(\lambda\in\Lambda,\ t\in[0,1]).
\end{equation}
The kernel is unique up to a null set of maximum ranks, simultaneously for all $t$.
\end{theorem}

\begin{proof}[Proof of sufficiency]
By Lemma~\ref{lem:profile-bijection}, equation~\eqref{eq:local-representation} defines a functional of the area curve alone.
Monotonicity and the endpoint conditions follow directly from the corresponding properties of $K$ and the identity $A_\lambda(1)=\int_0^1\lambda(u)\dd u$.
For an input with profiles $\lambda_1,\ldots,\lambda_n$,
\[
 \sum_iB_{K,\lambda_i}(t)
 =\int_0^1\left(\sum_i\lambda_i(u)\right)K(t,u)\dd u
 =t.
\]
Thus the rule is a local construction.
\end{proof}

To prove necessity, we use the fact that every finite family of profiles satisfying $\sum_i\lambda_i=1$ almost everywhere is realizable by an input; see Lemma~\ref{lem:profile-realization} in Appendix~\ref{app:realization}.
The total-allocation identity in Definition~\ref{def:local} therefore holds for every such family.
Comparing families obtained by splitting and merging profiles forces the allocation rule to be linear in $\lambda$.
Monotonicity in $t$ then yields the conditional distribution functions that define $K$.
Appendix~\ref{app:converse} gives the full argument and proves uniqueness.

The characterization concerns a single universal rule in the fixed normalized time coordinate $t$, not arbitrary allocations for an individual input.
An input-dependent change of time need not preserve locality.
The auxiliary coupling $(\widetilde U,\widetilde T)$ is used only to define the allocations; the actual arrival times $T_i$ are still sampled independently according to~\eqref{eq:arrivals}.
Figure~\ref{fig:pipeline} summarizes how a kernel turns an input into an algorithm.

\begin{figure}[t]
\centering
\begin{tikzpicture}[font=\footnotesize]
  \def\w{25.5mm}
  \foreach \i/\x in {0/13.8, 1/46.4, 2/79, 3/111.6, 4/144.2} { \coordinate (P\i) at (\x mm, 0); }
  \node[flowbox, fill=black!6, text width=\w, minimum height=15mm] (p0) at (P0)
    {\textbf{Input}\\[1pt] distributions\\ $F_1,\dots,F_n$};
  \node[flowbox, text width=\w, minimum height=15mm] (p1) at (P1)
    {\textbf{Profile} $\lambda_i$~\eqref{eq:profile}\\[1pt] share of the maximum at each rank};
  \node[flowbox, text width=\w, minimum height=15mm] (p2) at (P2)
    {\textbf{Allocation}~\eqref{eq:kernel-B}\\[1pt] $B_i(t)=$\\ $\int_0^1\!\lambda_i(u)K(t,u)\dd u$};
  \node[flowbox, text width=\w, minimum height=15mm] (p3) at (P3)
    {\textbf{Threshold path} $y_i$,\\ arrival law $G_i$~\eqref{eq:arrivals}};
  \node[flowbox, text width=\w, minimum height=15mm] (p4) at (P4)
    {\textbf{Load}~\eqref{eq:scalar-load}\\[1pt] $L_i=\ell_\Gamma(K,\lambda_i)$};
  \foreach \i/\j in {0/1,1/2,2/3,3/4} { \draw[flowarrow] (p\i) -- (p\j); }
  \begin{scope}[on background layer]
    \node[draw=black!45, dashed, rounded corners=3pt, fit=(p1) (p4), inner sep=2mm,
      label={[flowlabel, anchor=south, inner sep=1pt]north:computed separately for each item $i$, by the same rule}] {};
  \end{scope}
  \node[flowbox, fill=black!6, text width=84mm, minimum height=10mm] (alg) at (111.6mm,-24mm)
    {\textbf{Algorithm} (if every $L_i\le1$): sample $T_i\sim G_i$ independently,
     inspect the items in order of arrival, and accept the first $i$ with $U_i\le y_i(T_i)$};
  \node[flowbox, fill=black!6, text width=55mm, minimum height=10mm] (guar) at (29.5mm,-24mm)
    {\textbf{Guarantee} (Thm.~\ref{thm:bc}): $\Pp(I=i\mid V_i)=\Gamma\,\Pp(I^*=i\mid V_i)$, so
     $\E[\ALG]=\Gamma\,\E[\OPT]$};
  \draw[flowarrow] (p4.south) -- (p4.south |- alg.north);
  \draw[flowarrow] (alg) -- (guar);
\end{tikzpicture}
\caption{From an input to the algorithm, for a fixed allocation kernel $K$. Every item is
processed by the same rule. Because the profiles add up to one at almost every rank, the allocations add
up to normalized time, as required by~\eqref{eq:allocation}.}
\label{fig:pipeline}
\end{figure}
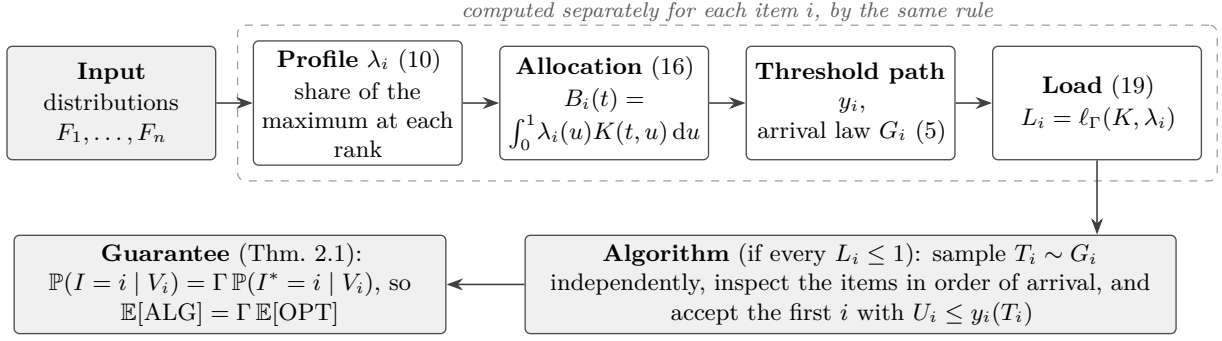

\begin{remark}
The constructions of \citet{PT} and \citet{BC} are not local in the sense of Definition~\ref{def:local} in general.
In Peng and Tang's common-threshold construction, each item's accumulated-area path depends only on $A_i$ when parameterized by the maximum-rank coordinate $u$.
However, expressing these paths in our normalized time $t$ requires a common, instance-dependent change of time.
The resulting allocation $B_i$ therefore need not be determined by $A_i$ alone.
Bubna and Chiplunkar's two-scheme algorithm additionally uses the Peng--Tang arrival loads of the entire instance to decide whether, and for which item, to modify the threshold path.
Thus our definition requires locality \emph{after} time normalization, not merely before it.
This restriction concerns the particular constructions above, rather than the general Bubna--Chiplunkar framework.
\end{remark}

\section{Kernel Design}\label{sec:kernel-design}\label{sec:elementary}

Theorem~\ref{thm:local} reduces the design of local constructions to the choice of an allocation kernel.
The allocation constraints are now automatic; the remaining task is to control the largest arrival load over all profiles.
We formulate this optimization for a fixed target ratio $\Gamma$.

\subsection{The Kernel Optimization Problem}\label{sec:kernel-optimization}

Fix a target ratio $\Gamma\in(0,1)$, and let $\mathcal K$ denote the class of allocation kernels.
For $K\in\mathcal K$ and $\lambda\in\Lambda$, let $A_\lambda$ be the corresponding area curve and $B_{K,\lambda}$ the allocation induced by $K$.
Define
\begin{equation}\label{eq:scalar-paths}
\begin{aligned}
 y_{K,\lambda}(t)&:=A_\lambda^{\leftarrow}\bigl(B_{K,\lambda}(t)\bigr),\\
 S_{K,\lambda}(t)&:=\exp\!\left(-\Gamma\int_0^t
                       \frac{\dd B_{K,\lambda}(s)}{S(s)}\right),
 \qquad S(t)=1-\Gamma t.
\end{aligned}
\end{equation}
We suppress the dependence of $S_{K,\lambda}$ on the fixed parameter $\Gamma$.
For $0\le s\le t\le1$, the bound $\lambda\le1$ and the marginal identity~\eqref{eq:kernel-marginal} give
\[
 0\le B_{K,\lambda}(t)-B_{K,\lambda}(s)
 \le\int_0^1\bigl(K(t,u)-K(s,u)\bigr)\dd u
 =t-s,
\]
so $B_{K,\lambda}$ is nondecreasing and $1$-Lipschitz for every profile $\lambda$, not only for the items of an input.
The load analysis below relies on this bound: it makes the threshold path $y_{K,\lambda}$ continuous (Lemma~\ref{lem:threshold-paths}) and the survival factor $S_{K,\lambda}$ absolutely continuous (Lemma~\ref{lem:load-parts}), and it gives $0\le1-B_{K,\lambda}'\le1$ almost everywhere, which Lemma~\ref{lem:gamma-monotone} and Proposition~\ref{prop:diagonal} use.
The corresponding single-profile load is
\begin{equation}\label{eq:scalar-load}
 \ell_\Gamma(K,\lambda)
 :=\mathcal L_\Gamma\bigl(A_\lambda,B_{K,\lambda}\bigr)
 =\Gamma\int_{(0,1]}\frac{S_{K,\lambda}(t)}{S(t)}\dd y_{K,\lambda}(t).
\end{equation}
These formulas define the load for every profile, whether or not the resulting arrival measure is feasible.
For an input item $i$, this load is exactly $L_i=\ell_\Gamma(K,\lambda_i)$; a load exceeding one means that the prescribed arrival measure cannot be completed to a probability distribution.

At this fixed target ratio, the objective is to minimize the worst-case load:
\begin{equation}\label{eq:kernel-optimization}
 \inf_{K\in\mathcal K}\;\sup_{\lambda\in\Lambda}\ell_\Gamma(K,\lambda).
\end{equation}
To attain the target ratio, we seek a kernel for which the inner supremum is at most one.
The kernel may depend on $\Gamma$, but the same kernel must work for every profile and every input.
The following proposition makes this feasibility criterion precise.

\begin{proposition}[Universal feasibility]\label{prop:scalar-feasibility}
For the fixed pair $(\Gamma,K)$, the prescribed arrival construction
is feasible for every atomless input if and only if
\begin{equation}\label{eq:robust}
 \sup_{\lambda\in\Lambda}\ell_\Gamma(K,\lambda)\le1.
\end{equation}
Whenever~\eqref{eq:robust} holds, the resulting algorithm satisfies
the exact item-wise guarantee~\eqref{eq:itemwise} for every atomless
input. Moreover, the same competitive guarantee extends to arbitrary
independent nonnegative distributions, including those with atoms
(Appendix~\ref{app:cert-atoms}).
\end{proposition}
\begin{proof}
By Theorem~\ref{thm:local}, $K$ defines a local construction, so the allocation constraints hold for every input.
If~\eqref{eq:robust} holds, all item loads are at most one, and Theorem~\ref{thm:bc} gives feasibility and the item-wise guarantee.
Conversely, every profile $\lambda\in\Lambda$ is realizable in an atomless, nonnegative two-item input with profiles $(\lambda,1-\lambda)$; Appendix~\ref{app:realization} gives the construction.
Universal feasibility therefore requires $\ell_\Gamma(K,\lambda)\le1$ for every $\lambda$. The extension to arbitrary independent distributions
is established in Appendix~\ref{app:cert-atoms}.
\end{proof}

We call an allocation kernel $K$ \emph{universally feasible at $\Gamma$}
if it satisfies~\eqref{eq:robust}.
This equivalence concerns feasibility of the prescribed arrival laws; it does not characterize all algorithms with competitive ratio $\Gamma$.
The following equivalent expression for the load will be useful in the analysis below.

\begin{lemma}[Integration by parts]\label{lem:load-parts}
Abbreviate $B=B_{K,\lambda}$, $y=y_{K,\lambda}$, and $S_\lambda=S_{K,\lambda}$.
Then
\begin{equation}\label{eq:load-parts}
 \ell_\Gamma(K,\lambda)
 =\Gamma\bigl(1-y(0)\bigr)
  +\Gamma^2\int_0^1
   \frac{S_\lambda(t)(1-y(t))(1-B'(t))}{S(t)^2}\dd t.
\end{equation}
Here $B^\prime$ is defined almost everywhere.
\end{lemma}
\begin{proof}
Since $B$ is $1$-Lipschitz, the function $W=S_\lambda/S$ is absolutely continuous and
\[
 W'(t)=\Gamma\frac{S_\lambda(t)(1-B'(t))}{S(t)^2}
 \quad\text{a.e.}
\]
Using $W(0)=1$ and $y(1)=1$, integration by parts gives
\[
 \int_{(0,1]}W(t)\dd y(t)
 =1-y(0)+\int_0^1(1-y(t))W'(t)\dd t.
\]
Multiplying by $\Gamma$ proves the identity.
\end{proof}

Here $y(0)=A_\lambda^{\leftarrow}(0)$ is the right endpoint of the initial interval on which $A_\lambda$ vanishes; it need not be zero.
Since the loss score of an item with profile $\lambda$ is almost surely at least $y(0)$ (Lemma~\ref{lem:threshold-paths}(b)), threshold levels below $y(0)$ require no arrival mass.
We will also use the following monotonicity property.

\begin{lemma}[Monotonicity in the target ratio]\label{lem:gamma-monotone}
Fix an allocation kernel $K$ and a profile $\lambda\in\Lambda$.
Then $\Gamma\mapsto\ell_\Gamma(K,\lambda)$ is nondecreasing on $(0,1)$.
In particular, if $K$ is universally feasible at some ratio $\Gamma$, then it is universally feasible at every ratio in $(0,\Gamma]$.
\end{lemma}

Its proof is given in Appendix~\ref{app:gamma-monotone}.

We first establish the common guarantee of all allocation kernels on
i.i.d.\ inputs.  We then analyze two elementary kernels on arbitrary
independent inputs.

\subsection{The I.I.D.\ Benchmark}\label{sec:iid-blind}

On i.i.d.\ inputs, all allocation kernels induce the same allocation.
We show that their common arrival construction attains the optimal
i.i.d.\ competitive ratio $\Gamma_\star$.

On an atomless i.i.d.\ input with $n\ge2$ items, exchangeability gives
$\Pp(I^*=i,\ U\le u)=u/n$ for every item $i$ and $u\in[0,1]$.
Hence $\lambda_i\equiv1/n$, and~\eqref{eq:kernel-marginal} yields
\begin{equation}\label{eq:iid-common-allocation}
 B_{K,1/n}(t)=\frac1n\int_0^1K(t,u)\dd u=\frac tn
\end{equation}
for every allocation kernel $K$.
Thus all kernels give the same probability thresholds $y_{K,1/n}$,
survival factors $S_{K,1/n}$, and arrival loads $\ell_\Gamma(K,1/n)$.

To establish the guarantee, we analyze a general constant profile
$\lambda\equiv c$, where $c\in(0,1)$, and then specialize to $c=1/n$.
The same marginal identity gives $B_{K,c}(t)=ct$.
The limiting threshold path is described by
\begin{equation}\label{eq:iid-clock}
 \phi(x):=x+(1-x)\log(1-x)
         =\int_0^x\frac{x-u}{1-u}\dd u,
 \qquad q:=\phi^{-1},
 \qquad x\in[0,1],
\end{equation}
where $\phi(1)=1$ by continuity.
Since $\phi(0)=0$ and $\phi'(x)=-\log(1-x)>0$ for $x\in(0,1)$,
$\phi$ is a continuous, strictly increasing bijection of $[0,1]$.
Moreover, $\phi(x)\le x$, so $q(t)\ge t$.
Define
\begin{equation}\label{eq:iid-load}
 J(\Gamma):=\Gamma\int_0^1\frac{\dd x}{1-\Gamma\phi(x)}.
\end{equation}

\begin{lemma}[Constant profiles]\label{lem:iid-kernel-blind}
Let $K$ be an allocation kernel, $\Gamma\in(0,1)$, and $c\in(0,1)$.
For the constant profile $\lambda\equiv c$, the area curve is
\begin{equation}\label{eq:constant-profile-area}
 A_c(x)=x-(1-c)+(1-c)(1-x)^{1/(1-c)},
 \qquad x\in[0,1].
\end{equation}
Writing $\bar A_c:=A_c/c$, we have
\[
 y_{K,c}=\bar A_c^{-1},
 \qquad S_{K,c}(t)=(1-\Gamma t)^c,
\]
and
\begin{equation}\label{eq:finite-iid-load}
 \ell_\Gamma(K,c)
 =\Gamma\int_0^1
   \frac{\dd x}{\bigl(1-\Gamma\bar A_c(x)\bigr)^{1-c}}
 \le J(\Gamma).
\end{equation}
As $c\downarrow0$, $\bar A_c\to\phi$ uniformly on $[0,1]$,
$y_{K,c}(t)\to q(t)$ for every $t\in[0,1]$, and
$\ell_\Gamma(K,c)\to J(\Gamma)$.
\end{lemma}

The proof is given in Appendix~\ref{app:iid-benchmark}.
In particular, $q$ is the limiting probability-threshold path as the
number of i.i.d.\ items grows; we call it the \emph{i.i.d.\ threshold path}.

The function $J$ is continuous and strictly increasing from zero to
infinity as $\Gamma$ ranges over $(0,1)$.
Let $\Gamma_\star$ be the unique solution of
\begin{equation}\label{eq:iid-constant}
 J(\Gamma_\star)=1,
 \qquad \Gamma_\star\approx0.7454403321.
\end{equation}
This is the classical optimal i.i.d.\ competitive ratio
\citep{HK82,Kertz86,CFHOV21}; Appendix~\ref{app:iid-benchmark}
records the equivalent defining equation.

\begin{proposition}[The i.i.d.\ guarantee]\label{prop:iid-guarantee}
For every allocation kernel $K$ and target ratio $\Gamma\in(0,1)$,
the prescribed arrival construction is feasible on all atomless
i.i.d.\ inputs, for every number of items,
if and only if $\Gamma\le\Gamma_\star$.
At $\Gamma_\star$, it satisfies the exact item-wise
guarantee~\eqref{eq:itemwise} on every such input and therefore
attains competitive ratio $\Gamma_\star$.
\end{proposition}
\begin{proof}
For an i.i.d.\ input with $n\ge2$ items, Lemma~\ref{lem:iid-kernel-blind}
and the monotonicity of $J$ give
\[
 \ell_\Gamma(K,1/n)\le J(\Gamma)\le J(\Gamma_\star)=1
 \qquad(\Gamma\le\Gamma_\star).
\]
Theorem~\ref{thm:bc} gives feasibility and the item-wise guarantee.
If $\Gamma>\Gamma_\star$, then
$\ell_\Gamma(K,1/n)\to J(\Gamma)>1$, so the arrival construction is
infeasible for all sufficiently large $n$.
The case $n=1$ is immediate.
\end{proof}

Since the inspection order is irrelevant for i.i.d.\ values, $\Gamma_\star$ is also an
upper bound on the competitive ratio of every order-selection algorithm: for every
$\varepsilon>0$, some i.i.d.\ input admits no stopping rule whose expected reward exceeds
$(\Gamma_\star+\varepsilon)$ times the expected maximum \citep{HK82,Kertz86}.
Thus i.i.d.\ inputs fix the target ratio and rule out universal
feasibility above $\Gamma_\star$, but they do not distinguish kernels.
We next test the perfectly correlated and independent kernels on
general inputs.  For both, the largest universally feasible ratio
is only $1-1/e$.

\subsection{The Perfectly Correlated Kernel}\label{sec:diagonal}

The perfectly correlated coupling sets $\widetilde T=\widetilde U$, giving the diagonal kernel
\begin{equation}\label{eq:diagonal-kernel}
 K_{\mathrm{diag}}(t,u):=\ind\{u\le t\}.
\end{equation}
Its marginal identity is immediate, and the induced allocation is
\begin{equation}\label{eq:diagonal-allocation}
 B_{K_{\mathrm{diag}},\lambda}(t)=\int_0^t\lambda(u)\dd u.
\end{equation}
Thus winning probability mass at maximum rank $u$ is allocated at time $t=u$.

\begin{proposition}[Perfectly correlated kernel]\label{prop:diagonal}
For every $\Gamma\in(0,1)$,
\begin{equation}\label{eq:diagonal-worst-load}
 \sup_{\lambda\in\Lambda}\ell_\Gamma(K_{\mathrm{diag}},\lambda)
 =-\log(1-\Gamma).
\end{equation}
Consequently, the prescribed arrival construction is feasible for every input if and only if $\Gamma\le1-1/e$.
\end{proposition}
\begin{proof}
Fix $\lambda\in\Lambda$ and abbreviate $A=A_\lambda$, $B=B_{K_{\mathrm{diag}},\lambda}$, $y=y_{K_{\mathrm{diag}},\lambda}$, and $S_\lambda=S_{K_{\mathrm{diag}},\lambda}$.
We first show that $y(t)\ge t$ for all $t\in[0,1]$.
By Lemma~\ref{lem:profile-bijection}, for $0\le u<1$,
\[
 0\le x_\lambda(u)\le u,
 \qquad A(x_\lambda(u))=C_\lambda(u)+x_\lambda(u)-u.
\]
Since $A$ is $1$-Lipschitz and $B=C_\lambda$,
\[
 A(u)\le A(x_\lambda(u))+u-x_\lambda(u)=C_\lambda(u)=B(u).
\]
Hence $A(t)\le B(t)$ and $y(t)=A^{\leftarrow}(B(t))\ge t$, including $t=1$ by the endpoint conditions.
Using this bound, $S_\lambda\le1$, $y(0)\ge0$, and $0\le1-B'(t)\le1$ in~\eqref{eq:load-parts}, we obtain
\begin{align}\label{eq:diagonal-upper-bound}
 \ell_\Gamma(K_{\mathrm{diag}},\lambda)
 &\le\Gamma+\Gamma^2\int_0^1
           \frac{1-t}{(1-\Gamma t)^2}\dd t\notag\\
 &=-\log(1-\Gamma).
\end{align}

For sharpness, take the profiles $\lambda_k(u)=u^k$, where $k\ge1$ is an integer.
Write $A_k,y_k,S_k,h_k,x_k$ for the corresponding functions.
Then
\[
 B_k(t)=\frac{t^{k+1}}{k+1},
 \qquad B_k'(t)=t^k.
\]
We claim that $y_k(t)\to t$ for every $0<t<1$.
Differentiating~\eqref{eq:profile-area-main} and using
$x_k'(u)=e^{h_k(u)}(1-u^k)$ gives, for $0<u<1$,
\begin{align}\label{eq:diagonal-sharp-area}
 A_k(x_k(u))
 &=\int_0^u(1-a^k)\bigl(e^{h_k(a)}-1\bigr)\dd a\notag\\
 &\ge\frac{(1-u^k)u^{k+2}}{(k+1)(k+2)}.
\end{align}
The inequality follows from $e^{h_k(a)}-1\ge h_k(a)\ge a^{k+1}/(k+1)$.
It also shows that $A_k$ has no initial flat interval, so $y_k(0)=0$.
For fixed $0<t<u<1$,
\[
 \frac{A_k(x_k(u))}{B_k(t)}
 \ge\frac{u(1-u^k)}{k+2}
       \left(\frac{u}{t}\right)^{k+1}
 \longrightarrow\infty.
\]
Thus $y_k(t)\le x_k(u)\le u$ for all sufficiently large $k$.
Together with $y_k(t)\ge t$, and then letting $u\downarrow t$, this proves the claim.
Moreover,
\[
 0\le-\log S_k(t)
 =\Gamma\int_0^t\frac{s^k}{1-\Gamma s}\dd s
 \le\frac{\Gamma}{(1-\Gamma)(k+1)},
\]
so $S_k\to1$ uniformly on $[0,1]$.
Since $B_k'(t)\to0$ for $t<1$, dominated convergence in~\eqref{eq:load-parts} yields
\[
 \lim_{k\to\infty}\ell_\Gamma(K_{\mathrm{diag}},\lambda_k)
 =\Gamma+\Gamma^2\int_0^1\frac{1-t}{(1-\Gamma t)^2}\dd t
 =-\log(1-\Gamma).
\]
This proves~\eqref{eq:diagonal-worst-load}.
The feasibility threshold follows from Proposition~\ref{prop:scalar-feasibility} and the equivalence
$-\log(1-\Gamma)\le1$ if and only if $\Gamma\le1-1/e$.
\end{proof}

\subsection{The Independent Kernel}\label{sec:independent}

When $\widetilde U$ and $\widetilde T$ are independent, the allocation kernel is
\begin{equation}\label{eq:independent-kernel}
 K_{\mathrm{ind}}(t,u):=t.
\end{equation}
Writing $\mu:=\int_0^1\lambda(u)\dd u$, we obtain
\begin{equation}\label{eq:independent-allocation}
 B_{K_{\mathrm{ind}},\lambda}(t)=\mu t.
\end{equation}
Thus each item's total winning probability is allocated at a constant rate, irrespective of its distribution across maximum ranks.

\begin{proposition}[Independent kernel]\label{prop:independent}
For every $\Gamma\in(0,1)$ and every profile $\lambda\in\Lambda$ of mass $\mu>0$,
\begin{equation}\label{eq:independent-load-bound}
 \ell_\Gamma(K_{\mathrm{ind}},\lambda)
 \le\frac{1-(1-\Gamma)^\mu}{\mu}
 \le-\log(1-\Gamma).
\end{equation}
If $\mu=0$, the load is zero.
Moreover,
\begin{equation}\label{eq:independent-worst-load}
 \sup_{\lambda\in\Lambda}\ell_\Gamma(K_{\mathrm{ind}},\lambda)
 =-\log(1-\Gamma).
\end{equation}
Consequently, the prescribed arrival construction is feasible for every input if and only if $\Gamma\le1-1/e$.
\end{proposition}
\begin{proof}
Suppose $\mu>0$ and abbreviate $A=A_\lambda$ and $y=y_{K_{\mathrm{ind}},\lambda}$.
Convexity, together with $A(0)=0$ and $A(1)=\mu$, gives $A(x)\le\mu x$.
Since $B(t)=\mu t$, it follows that $y(t)\ge t$.
The single-item survival factor is
\[
 S_{K_{\mathrm{ind}},\lambda}(t)=(1-\Gamma t)^\mu.
\]
Substituting into~\eqref{eq:load-parts} gives
\begin{align*}
 \ell_\Gamma(K_{\mathrm{ind}},\lambda)
 &\le\Gamma+\Gamma^2(1-\mu)
       \int_0^1(1-t)(1-\Gamma t)^{\mu-2}\dd t\\
 &=\Gamma\int_0^1(1-\Gamma t)^{\mu-1}\dd t
 =\frac{1-(1-\Gamma)^\mu}{\mu}.
\end{align*}
The second line follows by integrating the derivative of
$(1-t)(1-\Gamma t)^{\mu-1}$.
The remaining inequality in~\eqref{eq:independent-load-bound} follows from $1-e^{-z}\le z$ with $z=-\mu\log(1-\Gamma)$.
If $\mu=0$, then $A\equiv0$, $y\equiv1$, and the arrival load is zero.

For sharpness, take $\lambda_\mu(u)=\ind\{u\le\mu\}$ for $0<\mu<1$.
Lemma~\ref{lem:profile-bijection} gives $A_{\lambda_\mu}(x)=\mu x$, so $y(t)=t$ and
\[
 \ell_\Gamma(K_{\mathrm{ind}},\lambda_\mu)
 =\Gamma\int_0^1(1-\Gamma t)^{\mu-1}\dd t
 =\frac{1-(1-\Gamma)^\mu}{\mu}.
\]
Letting $\mu\downarrow0$ proves~\eqref{eq:independent-worst-load}.
The feasibility threshold again follows from Proposition~\ref{prop:scalar-feasibility}.
\end{proof}

Both elementary kernels attain $\Gamma_\star$ on i.i.d.\ inputs,
yet their universal feasibility threshold is only $1-1/e$.
Their worst-case loads arise from nonconstant profiles.
The next section uses the common i.i.d.\ threshold path to design a
kernel that gives a stronger guarantee on arbitrary independent inputs.

\section{From the I.I.D.\ Benchmark to the Kernel \texorpdfstring{$K_0$}{K0}}\label{sec:k0}

Section~\ref{sec:iid-blind} shows that every allocation kernel attains
$\Gamma_\star$ on i.i.d.\ inputs.  Our goal is to extend this guarantee
to arbitrary independent inputs.  Since i.i.d.\ inputs do not distinguish
kernels, we examine small nonconstant profiles and use their limiting
threshold paths to guide the design of a baseline kernel $K_0$.
The resulting kernel is universally feasible up to
$\Gamma_0\approx0.7389$, but not at $\Gamma_\star$.

\paragraph{Overview.}
The construction proceeds in two steps.
\begin{enumerate}
\item \emph{Linearizing small profiles}
(Section~\ref{sec:small-profiles}).
When a profile is scaled toward zero, its area curve is approximated
to first order by a linear function of the profile.
This approximation determines the limiting threshold path and arrival
load for each profile shape.
For the constant shape, the limiting path is $q$ and the load is
$J(\Gamma_\star)=1$.
\item \emph{Matching the i.i.d.\ threshold path}
(Section~\ref{sec:design-principle}).
Feasibility at $\Gamma_\star$ imposes a first-order condition on the
kernel.  A natural strengthening requires the limiting threshold path
to equal $q$ for every small-profile shape.
This requirement determines a unique kernel,
\[
 K_0(t,u)=\frac{(q(t)-u)_+}{1-u}.
\]
\end{enumerate}
The kernel $K_0$ allocates to each item the linearized area of its
profile evaluated at the i.i.d.\ threshold $q(t)$.
Section~\ref{sec:kzero-definition} gives the formal definition and a
probabilistic description.
For profiles of positive size, the true area curve differs from this
linear approximation, while the item's survival factor also affects
the load.  Section~\ref{sec:kzero-analysis} bounds these effects
uniformly, and Section~\ref{sec:kzero-limit} gives a profile whose load
exceeds one at $\Gamma_\star$.

\subsection{Small Profiles and the Linearized Area Curve}\label{sec:small-profiles}

To distinguish kernels near the i.i.d.\ limit, consider
\emph{small profiles} $\lambda=\varepsilon f$, obtained by scaling a
fixed \emph{shape} $f$ by $\varepsilon\downarrow0$.
The constant shape $f\equiv1$ recovers the i.i.d.\ limit.
Nonconstant shapes describe rarely winning items whose winning
probability is spread nonuniformly over maximum ranks.
Let $\mathcal F$ denote the set of measurable functions $f:[0,1]\to(0,\infty)$ that are bounded and bounded away from zero.
Throughout, $f\in\mathcal F$ and $0<\varepsilon\le1/(2\|f\|_\infty)$, so that $\varepsilon f\le1/2$.

\paragraph{The rank representation.}
The analysis of small profiles is based on a representation of area curves proved in Appendix~\ref{app:profile}.
For an item $i$, let
\[
 R_i:=1-F_{\max}(V_i)
\]
be the \emph{rank} of its value, that is, the maximum rank that would result if $V_i$ were the maximum; in particular, $U=\min_jR_j$.
Recall that $U_i=1-\prod_{j\ne i}F_j(V_i)$ is the loss score of item $i$, and recall the functions $h_\lambda$ and $x_\lambda$ from
\eqref{eq:profile-coordinates-main}.
Lemma~\ref{lem:rank-representation} shows that $\Pp(R_i>s)=e^{-h_{\lambda_i}(s)}$ for $s\in[0,1)$ and that $U_i=x_{\lambda_i}(R_i)$.
Combining these identities with the definition
$A_i(x)=\E[(x-U_i)_+]$, we obtain, for every profile
$\lambda\in\Lambda$,
\begin{equation}\label{eq:rank-area-main}
 A_\lambda(x)
 =\int_0^1\bigl(x-x_\lambda(s)\bigr)_+\,\frac{\lambda(s)}{1-s}\,e^{-h_\lambda(s)}\dd s,
 \qquad x\in[0,1].
\end{equation}
Here $\lambda(s)(1-s)^{-1}e^{-h_\lambda(s)}$ is the density of the rank $R_i$ on $[0,1)$, and $x_\lambda(s)=1-(1-s)e^{h_\lambda(s)}$ is the loss score of a value of rank $s$.
The loss score satisfies $x_\lambda(s)\le s$ because it compares
item $i$ only with the other items, whereas rank $s$ is defined
using the maximum of all items.

\paragraph{Linearization.}
For a small profile, $h_\lambda\approx0$, so $e^{-h_\lambda}\approx1$ and $x_\lambda(s)\approx s$.
The area curve then becomes linear in the profile:
\[
 A_\lambda(x)\approx a_\lambda(x),
 \qquad\text{where}\qquad
 a_g(x):=\int_0^1g(u)\,\kappa(x,u)\dd u,
 \qquad
 \kappa(x,u):=\frac{(x-u)_+}{1-u}\quad(u<1),
\]
for every bounded measurable function $g$ on $[0,1]$.
The elementary area curve $\kappa(\cdot,u)$ has a transparent meaning.
If item $i$ is negligible, a value of rank $u$ beats all other values with probability about $1-u$, and its loss score is about $u$.
One unit of winning probability at maximum rank $u$ therefore corresponds to $1/(1-u)$ units of probability that $V_i$ has rank $u$, each of which contributes $(x-u)_+$ to $A_i(x)=\E[(x-U_i)_+]$.
Thus $\kappa(x,u)$ is the area at threshold level $x$ generated by one unit of winning probability at maximum rank $u$.
As a function of $x$, it is the distribution function of the uniform distribution on $[u,1]$.
Integrating over $u$ recovers the i.i.d.\ curve: $a_1=\phi$ by~\eqref{eq:iid-clock}.
Evaluating the integral, for $x<1$,
\begin{equation}\label{eq:linearized-area}
 a_g(x)=\int_0^xg(u)\dd u-(1-x)\int_0^x\frac{g(u)}{1-u}\dd u,
 \qquad
 a_g'(x)=\int_0^x\frac{g(u)}{1-u}\dd u.
\end{equation}
In particular, $a_\lambda(x)=C_\lambda(x)-(1-x)h_\lambda(x)$ for every profile $\lambda$.

For $f\in\mathcal F$, the function $a_f$ is continuous and strictly increasing from $a_f(0)=0$ to $a_f(1)=\int_0^1 f(u)\dd u$.
Given an allocation kernel $K$, define
\begin{equation}\label{eq:rare-area-allocation}
 m_{K,f}(t):=\int_0^1f(u)K(t,u)\dd u,
 \qquad
 z_{K,f}(t):=a_f^{-1}\bigl(m_{K,f}(t)\bigr),
 \qquad t\in[0,1].
\end{equation}
For $s<t$, the monotonicity of $K(\cdot,u)$ and
\eqref{eq:kernel-marginal} give
$0\le m_{K,f}(t)-m_{K,f}(s)\le\|f\|_\infty(t-s)$; moreover, $m_{K,f}(0)=0$ and $m_{K,f}(1)=\int_0^1 f(u)\dd u$.
Hence $z_{K,f}$ is continuous and nondecreasing, with
$z_{K,f}(0)=0$ and $z_{K,f}(1)=1$.
The next lemma shows that it is the limit of $y_{K,\varepsilon f}$
as $\varepsilon\downarrow0$ and computes the corresponding limiting
arrival load.

\begin{lemma}[Small-profile limit]\label{lem:rare-limit}
Let $K$ be an allocation kernel, $\Gamma\in(0,1)$, and $f\in\mathcal F$.
As $\varepsilon\downarrow0$, the following hold.
\begin{enumerate}
\item[(a)] $A_{\varepsilon f}/\varepsilon\to a_f$ uniformly on $[0,1]$, and $y_{K,\varepsilon f}(t)\to z_{K,f}(t)$ for every $t\in[0,1]$.
\item[(b)] The loads converge to a limit, which we call the \emph{zero-scale load} of the shape $f$ and which is determined by
$z_{K,f}$:
\begin{equation}\label{eq:rare-load}
 \lim_{\varepsilon\downarrow0}\ell_\Gamma(K,\varepsilon f)
 =\ell^{0}_\Gamma(K,f)
 :=\Gamma\int_{(0,1]}
       \frac{\dd z_{K,f}(t)}{1-\Gamma t}
 =\Gamma+\Gamma^2\int_0^1
       \frac{1-z_{K,f}(t)}{(1-\Gamma t)^2}\dd t.
\end{equation}
\end{enumerate}
Consequently, if $K$ is universally feasible at $\Gamma$, then $\ell^{0}_\Gamma(K,f)\le1$ for every $f\in\mathcal F$.
\end{lemma}

\begin{proof}
(a)
Write $h_f(s):=\int_0^sf(u)(1-u)^{-1}\dd u$, so that $h_{\varepsilon f}=\varepsilon h_f$.
By~\eqref{eq:rank-area-main},
\[
 \frac{A_{\varepsilon f}(x)}{\varepsilon}
 =\int_0^1\bigl(x-x_{\varepsilon f}(s)\bigr)_+\,\frac{f(s)}{1-s}\,e^{-\varepsilon h_f(s)}\dd s.
\]
For every $s<1$, $x_{\varepsilon f}(s)=1-(1-s)e^{\varepsilon h_f(s)}\to s$ and $e^{-\varepsilon h_f(s)}\to1$, so the integrand converges to $f(s)\kappa(x,s)$.
Fix $x<1$.
Since $h_f(s)\le\|f\|_\infty\log\frac1{1-s}$ and $\varepsilon\|f\|_\infty\le\frac12$, we have $1-x_{\varepsilon f}(s)\le(1-s)^{1/2}$.
Hence the integrand vanishes unless $1-s>(1-x)^2$, and it is bounded by $\|f\|_\infty(1-x)^{-2}$.
Dominated convergence gives $A_{\varepsilon f}(x)/\varepsilon\to a_f(x)$.
At $x=1$, both sides equal $\int_0^1 f(u)\dd u$ by Lemma~\ref{lem:profile-bijection}.
Pointwise convergence of nondecreasing functions to a continuous limit on a compact interval is uniform, which proves the first claim.

For the second claim, note that $B_{K,\varepsilon f}=\varepsilon m_{K,f}$, so $y_{K,\varepsilon f}(t)$ is the upper generalized inverse of $A_{\varepsilon f}/\varepsilon$ evaluated at $m_{K,f}(t)$.
Let $z:=z_{K,f}(t)$ and $\eta>0$.
If $z+\eta\le1$, then $a_f(z+\eta)>m_{K,f}(t)$, and uniform convergence gives $A_{\varepsilon f}(z+\eta)/\varepsilon>m_{K,f}(t)$ for small $\varepsilon$; hence $y_{K,\varepsilon f}(t)\le z+\eta$.
Symmetrically, if $z-\eta\ge0$, then $y_{K,\varepsilon f}(t)\ge z-\eta$ for small $\varepsilon$.
Thus $y_{K,\varepsilon f}(t)\to z$.

(b)
Since $\varepsilon f>0$, the integrand in~\eqref{eq:rank-area-main} is positive near $s=0$ for every $x>0$; hence $A_{\varepsilon f}>0$ on $(0,1]$ and $y_{K,\varepsilon f}(0)=0$.
Moreover, $0\le B'_{K,\varepsilon f}\le\varepsilon\|f\|_\infty$ almost everywhere, so
\[
 1\ge S_{K,\varepsilon f}(t)
 \ge\exp\!\left(-\Gamma\varepsilon\|f\|_\infty\int_0^1\frac{\dd s}{1-\Gamma s}\right)
 =(1-\Gamma)^{\varepsilon\|f\|_\infty}
 \longrightarrow1.
\]
In the representation~\eqref{eq:load-parts} of Lemma~\ref{lem:load-parts}, the integrand is bounded by $(1-\Gamma)^{-2}$ and converges to $(1-z_{K,f}(t))(1-\Gamma t)^{-2}$ for almost every $t$.
Dominated convergence proves~\eqref{eq:rare-load}.
The final assertion follows from Proposition~\ref{prop:scalar-feasibility}.
\end{proof}

In the small-profile limit, the item's own survival factor tends to
one and its allocation rate tends to zero.
The second expression in~\eqref{eq:rare-load} shows that a higher
limiting threshold path gives a lower load.
Thus a kernel serves small profiles well if it keeps their threshold
paths high, uniformly over shapes.
Two benchmarks frame what is possible.
\begin{itemize}
\item For the constant shape $f\equiv1$, we have $a_1=\phi$ and $m_{K,1}(t)=t$, so $z_{K,1}=q$ and $\ell^{0}_\Gamma(K,1)=J(\Gamma)$ for \emph{every} kernel.
Hence $\sup_{f\in\mathcal F}\ell^{0}_\Gamma(K,f)\ge J(\Gamma)$: in the worst case, no kernel does better than the i.i.d.\ path.
\item The two elementary kernels considered in
Section~\ref{sec:kernel-design} behave very differently.
The sharpness examples in Propositions~\ref{prop:diagonal} and~\ref{prop:independent} are items whose profiles have vanishing mass and whose threshold paths converge to the identity. Their loads therefore converge to
$\Gamma\int_0^1(1-\Gamma t)^{-1}\dd t=-\log(1-\Gamma)$.
\end{itemize}
The question is therefore whether a single kernel can make the i.i.d.\ path $q$
the limiting threshold path for every shape $f\in\mathcal F$.

\subsection{A First-Order Condition and the Design Principle}\label{sec:design-principle}
At $\Gamma_\star$, Lemma~\ref{lem:rare-limit} implies that any
universally feasible kernel $K$ satisfies
$\ell^{0}_{\Gamma_\star}(K,f)\le1$ for every $f\in\mathcal F$.
For the constant shape $f\equiv1$, Lemma~\ref{lem:iid-kernel-blind}
gives $\ell^{0}_{\Gamma_\star}(K,1)=J(\Gamma_\star)=1$.
Hence $f\equiv1$ maximizes $f\mapsto\ell^{0}_{\Gamma_\star}(K,f)$,
so the derivative of this map at $f\equiv1$ must vanish in every
direction. This yields a linear condition on the kernel.

\begin{proposition}[First-order necessary condition]\label{prop:iid-calibration}
If an allocation kernel $K$ is universally feasible at $\Gamma_\star$, then, for almost every $u\in[0,1]$,
\begin{equation}\label{eq:iid-calibration}
 \int_0^1
 \frac{K(t,u)-\kappa\bigl(q(t),u\bigr)}
 {[-\log(1-q(t))]\,(1-\Gamma_\star t)^2}\dd t=0.
\end{equation}
\end{proposition}

\begin{proof}
Fix any bounded measurable function $g$ with $\|g\|_\infty>0$, let $\delta_0:=1/(2\|g\|_\infty)$, and for $|\delta|<\delta_0$ put $f_\delta:=1+\delta g$.
Then $\frac12\le f_\delta\le\frac32$, so $f_\delta\in\mathcal F$.
Write $z_\delta(t):=z_{K,f_\delta}(t)$.
Since $a_{f_\delta}=\phi+\delta a_g$ and $m_{K,f_\delta}(t)=t+\delta m_{K,g}(t)$, the number $z_\delta(t)$ is the unique solution $z\in[0,1]$ of
\[
 \phi(z)+\delta a_g(z)=t+\delta m_{K,g}(t).
\]
Fix $t\in(0,1)$.
Then $0<m_{K,f_\delta}(t)<\int_0^1f_\delta$, so $z_\delta(t)\in(0,1)$.
The left side is continuously differentiable in $(\delta,z)\in(-\delta_0,\delta_0)\times(0,1)$, and by~\eqref{eq:linearized-area} its derivative in $z$ is $a_{f_\delta}'(z)\ge\frac12\log\frac1{1-z}>0$.
By the implicit function theorem, $\delta\mapsto z_\delta(t)$ is continuously differentiable, with
\begin{equation}\label{eq:first-order-derivative}
 \partial_\delta z_\delta(t)
 =\frac{m_{K,g}(t)-a_g(z_\delta(t))}{a'_{f_\delta}(z_\delta(t))},
 \qquad
 \partial_\delta z_\delta(t)\big|_{\delta=0}
 =\frac{\int_0^1g(u)\bigl(K(t,u)-\kappa(q(t),u)\bigr)\dd u}{-\log(1-q(t))}.
\end{equation}
This derivative is bounded uniformly in $t$ and $\delta$.
Indeed, $\frac12t\le m_{K,f_\delta}(t)=a_{f_\delta}(z_\delta)\le\frac32\phi(z_\delta)$ gives $t\le3\phi(z_\delta)$.
Hence the numerator in~\eqref{eq:first-order-derivative} is at most $\|g\|_\infty(t+\phi(z_\delta))\le4\|g\|_\infty\phi(z_\delta)$.
The denominator is at least $\frac12\log\frac1{1-z_\delta}\ge\frac12z_\delta\ge\frac12\phi(z_\delta)$.
Therefore $|\partial_\delta z_\delta(t)|\le8\|g\|_\infty$.

By Lemma~\ref{lem:rare-limit}, $\ell^{0}_{\Gamma_\star}(K,f_\delta)\le1$ for every $|\delta|<\delta_0$, with equality at $\delta=0$.
The uniform bound on $\partial_\delta z_\delta$ allows differentiation under the integral in~\eqref{eq:rare-load}, by the mean value theorem and dominated convergence.
Since $\delta=0$ is a maximizer,
\[
 0=\frac{\dd}{\dd\delta}\ell^{0}_{\Gamma_\star}(K,f_\delta)\Big|_{\delta=0}
 =-\Gamma_\star^2\int_0^1
 \frac{\int_0^1g(u)\bigl(K(t,u)-\kappa(q(t),u)\bigr)\dd u}
 {[-\log(1-q(t))]\,(1-\Gamma_\star t)^2}\dd t.
\]
Let $\Theta(u)$ denote the left side of~\eqref{eq:iid-calibration}.
Since $|K-\kappa|\le1$ and the substitution $t=\phi(x)$ gives $\int_0^1[-\log(1-q(t))]^{-1}\dd t=1$, Fubini's theorem applies, $\Theta$ is bounded, and the identity above gives $\int_0^1g(u)\Theta(u)\dd u=0$.
Since $g$ is arbitrary, $\Theta=0$ almost everywhere.
\end{proof}

Condition~\eqref{eq:iid-calibration} constrains only a weighted time average of the deviation $K(t,u)-\kappa(q(t),u)$, so it does not determine $K$.
It holds trivially, however, when the deviation vanishes identically.
By~\eqref{eq:first-order-derivative}, the first-order perturbation of every small-profile threshold path then vanishes at every time, not merely on average.
Because $a_f$ and $m_{K,f}$ are linear in $f$, even more is true: the next proposition shows that a vanishing deviation is equivalent to $z_{K,f}=q$ exactly, for every shape, and that it singles out one kernel.
Recall that allocation kernels are identified when they agree
for almost every maximum rank at all times $t$.

\begin{proposition}[Matching the i.i.d.\ path]\label{prop:iid-path-kernel}
For an allocation kernel $K$, the following are equivalent:
\begin{enumerate}
\item[(a)] $z_{K,f}=q$ for every shape $f\in\mathcal F$;
\item[(b)] there is a null set $N\subseteq[0,1]$ such that $K(t,u)=\kappa(q(t),u)$ for all $t\in[0,1]$ and $u\notin N$.
\end{enumerate}
In this case, $\ell^{0}_\Gamma(K,f)=J(\Gamma)$ for every $f\in\mathcal F$ and every $\Gamma\in(0,1)$.
\end{proposition}

\begin{proof}
Recall from~\eqref{eq:rare-area-allocation} that
$z_{K,f}=a_f^{-1}\circ m_{K,f}$. Since $a_f$ is injective, $z_{K,f}(t)=q(t)$ holds if and only if $m_{K,f}(t)=a_f(q(t))$, that is,
\[
 \int_0^1f(u)\bigl(K(t,u)-\kappa(q(t),u)\bigr)\dd u=0.
\]
Hence (b) implies (a).
Conversely, assume (a) and fix $t$.
Applying the identity to $f\equiv1$ and to $f=1+\ind_E$ for a measurable set $E$ gives $\int_E(K(t,u)-\kappa(q(t),u))\dd u=0$ for every $E$, so $K(t,\cdot)=\kappa(q(t),\cdot)$ almost everywhere.
Taking the union of the exceptional null sets over rational $t$, and using that both functions are right-continuous in $t$ for almost every $u$, we obtain agreement for all $t$ outside a common null set.
The last assertion follows from~\eqref{eq:rare-load} and~\eqref{eq:iid-load}.
\end{proof}

Proposition~\ref{prop:iid-path-kernel} identifies the only kernel under which \emph{every small item, whatever the shape of its profile, follows the i.i.d.\ threshold path}.
We denote it by $K_0$; Proposition~\ref{prop:kzero-kernel} below verifies that it is indeed an allocation kernel.
In terms of allocations, $K_0$ assigns to a profile $\lambda$ the allocation $B_{K_0,\lambda}(t)=a_\lambda(q(t))$.
That is, \emph{each item receives, at time $t$, the linearized area of its profile at the i.i.d.\ threshold level $q(t)$.}
For a small profile, the linearized area is asymptotically the true area, so the item's probability threshold is asymptotically $q(t)$, exactly as in the i.i.d.\ benchmark.

\subsection{The Kernel \texorpdfstring{$K_0$}{K0}}\label{sec:kzero-definition}

Formally, define
\begin{equation}\label{eq:kzero-uniform}
 K_0(t,u):=\kappa\bigl(q(t),u\bigr)=\frac{(q(t)-u)_+}{1-u},
 \qquad t\in[0,1],\ u\in[0,1),
\end{equation}
and, for definiteness, $K_0(t,1):=t$.
Although an allocation kernel may in general depend on the target ratio, $K_0$ does not.

\begin{proposition}\label{prop:kzero-kernel}
The function $K_0$ is an allocation kernel.
For every profile $\lambda\in\Lambda$ and $t\in[0,1)$, its allocation is
\begin{equation}\label{eq:kzero-allocation}
 B_{K_0,\lambda}(t)
 =a_\lambda\bigl(q(t)\bigr)
 =C_\lambda\bigl(q(t)\bigr)-\bigl(1-q(t)\bigr)h_\lambda\bigl(q(t)\bigr),
\end{equation}
and $B_{K_0,\lambda}(1)=\mu_\lambda$.
Moreover, for every $f\in\mathcal F$ and $\Gamma\in(0,1)$,
\[
 z_{K_0,f}=q
 \qquad\text{and}\qquad
 \lim_{\varepsilon\downarrow0}\ell_\Gamma(K_0,\varepsilon f)=J(\Gamma).
\]
In particular, $K_0$ passes every small-profile test of Lemma~\ref{lem:rare-limit} at $\Gamma_\star$, with equality.
\end{proposition}

\begin{proof}
For each $u<1$, the function $t\mapsto K_0(t,u)$ is continuous and nondecreasing, because $q$ is; it vanishes at $t=0$ and equals one at $t=1$, because $q(0)=0$ and $q(1)=1$.
The function $K_0$ is continuous on $[0,1]\times[0,1)$ and hence jointly measurable.
By~\eqref{eq:iid-clock}, $\int_0^1K_0(t,u)\dd u=\phi(q(t))=t$.
Thus $K_0$ is an allocation kernel.
Formula~\eqref{eq:kzero-allocation} is the definition of $a_\lambda$ combined with~\eqref{eq:linearized-area}, and $B_{K_0,\lambda}(1)=\int_0^1\lambda(u)\dd u=\mu_\lambda$.
The remaining assertions follow from Proposition~\ref{prop:iid-path-kernel} and Lemma~\ref{lem:rare-limit}.
\end{proof}

\paragraph{A probabilistic description.}
Recall from Section~\ref{sec:characterization} that an allocation kernel is the conditional distribution function $K(t,u)=\Pp(\widetilde T\le t\mid\widetilde U=u)$ of a coupling of two uniform random variables.
The coupling behind $K_0$ is simple.
Let $\widetilde U$ be uniform on $[0,1]$; given $\widetilde U=u$, let the \emph{release level} $X$ be uniform on $[u,1]$; and set $\widetilde T:=\phi(X)$.
Then
\[
 \Pp(\widetilde T\le t\mid\widetilde U=u)
 =\Pp\bigl(X\le q(t)\mid\widetilde U=u\bigr)
 =\frac{(q(t)-u)_+}{1-u}
 =K_0(t,u).
\]
The release level $X$ has density $\int_0^x(1-u)^{-1}\dd u=\phi'(x)$, so $\phi$ is its distribution function and $\widetilde T=\phi(X)$ is uniform.
Thus \emph{$K_0$ spreads the winning probability at maximum rank $u$ uniformly over the threshold levels $x\in[u,1]$, and releases each part when the i.i.d.\ threshold reaches its level}, that is, at time $\phi(x)$.
The marginal identity~\eqref{eq:kernel-marginal} for $K_0$ is precisely the statement that $\phi$ is the distribution function of the release level.
Figure~\ref{fig:kzero}(a) illustrates the construction.

The same coupling has a simple description on the exponential scale
$r=-\log(1-u)$.
Let $\xi$ and $\xi'$ be independent standard exponential random variables.
Then $\widetilde U:=1-e^{-\xi}$ is uniform and, given $\widetilde U$, the level $X:=1-e^{-(\xi+\xi')}$ is uniform on $[\widetilde U,1]$, because $e^{-\xi'}$ is uniform on $[0,1]$.
Moreover, $\phi(1-e^{-r})=1-(1+r)e^{-r}$ is the distribution function of the $\mathrm{Gamma}(2,1)$ distribution, so
\[
 \widetilde T=1-(1+\xi+\xi')\,e^{-(\xi+\xi')}.
\]
Thus, \emph{on the exponential scale, $K_0$ releases the winning probability located at level $\xi$ after an independent standard exponential delay $\xi'$}, and normalized time is the $\mathrm{Gamma}(2,1)$ distribution function of the release level $\xi+\xi'$.
For comparison, the perfectly correlated kernel of Section~\ref{sec:diagonal} releases the winning probability at rank $u$ at time $u$, and the independent kernel of Section~\ref{sec:independent} ignores the rank altogether.
The kernel $K_0$ releases it only after the i.i.d.\ threshold has passed the rank, and then at a memoryless pace; see Figure~\ref{fig:kzero}(b).
On the same scale, profiles acquire a direct meaning: by Lemma~\ref{lem:rank-representation}, $\Pp(R_i>1-e^{-r})=\exp(-\int_0^r\lambda_i(1-e^{-s})\dd s)$, so $\lambda_i(1-e^{-r})$ is the hazard rate at level $r$ of item $i$'s rank.
Section~\ref{sec:kzero-analysis} carries out the analysis in these coordinates.

\begin{figure}[t]
\centering
\begin{minipage}[b]{0.47\linewidth}
\centering
\begin{tikzpicture}[scale=4.4,>=stealth]
 \def\uu{0.5}\def\qq{0.76}
 \fill[black!5] (0,0) -- (0,1) -- (1,1) -- cycle;
 \fill[black!17] (0,0) -- (0,\qq) -- (\qq,\qq) -- cycle;
 \draw[black!55] (0,1) -- (1,1) -- (1,0);
 \draw[black!55] (0,0) -- (1,1);
 \draw[dashed] (0,\qq) -- (1,\qq);
 \draw[dotted] (\uu,0) -- (\uu,\uu);
 \draw[line width=1.8pt] (\uu,\uu) -- (\uu,\qq);
 \draw[line width=1.8pt,black!35] (\uu,\qq) -- (\uu,1);
 \draw[->] (0,0) -- (1.1,0);
 \draw[->] (0,0) -- (0,1.1) node[above right] {\small threshold level};
 \node[anchor=north] at (0.5,-0.085) {\small maximum rank};
 \node[left] at (0,\qq) {\small $q(t)$};
 \node[left] at (0,1) {\small $1$};
 \node[below] at (1,0) {\small $1$};
 \node[below] at (\uu,0) {\small $u$};
 \node[below left] at (0,0) {\small $0$};
 \draw[decorate,decoration={brace,amplitude=2.5pt}] (\uu-0.025,\uu+0.005) -- (\uu-0.025,\qq-0.005)
   node[midway,left=3pt,align=right] {\scriptsize released\\[-2pt]\scriptsize by time $t$};
 \draw[decorate,decoration={brace,amplitude=2.5pt}] (\uu-0.025,\qq+0.005) -- (\uu-0.025,0.995)
   node[midway,left=3pt,align=right] {\scriptsize released\\[-2pt]\scriptsize later};
 \node at (0.15,0.40) {\scriptsize mass $t$};
 \node[rotate=45] at (0.72,0.635) {\scriptsize $x=u$};
\end{tikzpicture}

\small (a) The coupling behind $K_0$.
\end{minipage}
\hfill
\begin{minipage}[b]{0.49\linewidth}
\centering
\begin{tikzpicture}
\begin{axis}[
 width=6.4cm,height=5.4cm,
 xmin=0,xmax=1,ymin=-0.04,ymax=1.04,
 xlabel={\small time $t$},ylabel={\small $K(t,\tfrac12)$},
 xtick={0,0.1534,0.5,1},xticklabels={$0$,$\phi(\frac12)$,$\frac12$,$1$},
 ytick={0,0.5,1},
 tick label style={font=\scriptsize},
 legend style={font=\scriptsize,at={(0.97,0.05)},anchor=south east,draw=none,
   fill=white,fill opacity=0.9,text opacity=1},
 legend cell align=left]
 \addplot[line width=1.4pt,black,domain=0.5:0.99999,samples=160,variable=\s]
   ({\s+(1-\s)*ln(1-\s)},{(\s-0.5)/0.5});
 \addlegendentry{$K_0$}
 \addplot[thick,dashed,black!70] coordinates {(0,0) (0.5,0) (0.5,1) (1,1)};
 \addlegendentry{$K_{\mathrm{diag}}$}
 \addplot[thick,dotted,black] coordinates {(0,0) (1,1)};
 \addlegendentry{$K_{\mathrm{ind}}$}
 \addplot[line width=1.4pt,black,forget plot] coordinates {(0,0) (0.1534,0)};
\end{axis}
\end{tikzpicture}

\small (b) Release schedules at maximum rank $u=\frac12$.
\end{minipage}
\caption{The kernel $K_0$.
(a)~The pair (maximum rank, release level) has density $1/(1-u)$ on the triangle above the
diagonal: winning probability at maximum rank $u$ is spread uniformly over the levels $[u,1]$
(vertical segment) and is released when the i.i.d.\ threshold $q(t)$ reaches each level. By
time $t$, the released fraction is the part of the segment below $q(t)$, namely
$K_0(t,u)=(q(t)-u)/(1-u)$, and the darker triangle carries total mass $\phi(q(t))=t$.
(b)~Conditional distribution functions $t\mapsto K(t,u)$ for $u=\frac12$: the perfectly
correlated kernel releases all winning probability at time $u$, the independent kernel
releases it uniformly over time, and $K_0$ releases it gradually from time
$\phi(u)\approx0.153$ on.}
\label{fig:kzero}
\end{figure}

\paragraph{Limitations of small-profile tests.}
By Proposition~\ref{prop:kzero-kernel}, $K_0$ passes all small-profile tests at $\Gamma_\star$.
Profiles of positive size, however, are subject to two effects that the linearization ignores, and these effects push the load in opposite directions.
\begin{itemize}
\item \emph{Self-competition.}
An item does not compete with itself: a value of rank $s$ has loss score $x_\lambda(s)=1-(1-s)e^{h_\lambda(s)}$, which is strictly smaller than $s$ when $h_\lambda(s)>0$.
Hence the true area curve $A_\lambda$ differs from its linearization $a_\lambda$.
Since $K_0$ allocates exactly $a_\lambda(q(t))$, the probability threshold $y_{K_0,\lambda}(t)=A_\lambda^\leftarrow\bigl(a_\lambda(q(t))\bigr)$ departs from the i.i.d.\ path.
Where $A_\lambda>a_\lambda$, it lags behind the i.i.d.\ path, and this raises the load.
\item \emph{Self-survival.}
The survival factor $S_{K_0,\lambda}<1$ accounts for the possibility that the item itself has already arrived and met its threshold, and this lowers the load.
\end{itemize}
For constant profiles, self-survival dominates (Lemma~\ref{lem:iid-kernel-blind}).
Section~\ref{sec:kzero-analysis} bounds the net effect for all profiles, and Section~\ref{sec:kzero-limit} exhibits a profile for which self-competition dominates at $\Gamma_\star$.

\subsection{A Universal Guarantee for \texorpdfstring{$K_0$}{K0}}
\label{sec:kzero-analysis}

The following theorem bounds the net effect of self-competition and self-survival uniformly over all profiles.
The guarantee is described by a scalar differential equation.

\begin{theorem}[Guarantee for $K_0$]\label{thm:kzero-analytic}
For $\Gamma\in(0,1)$ and $r\ge0$, put
\begin{equation}\label{eq:kzero-D}
 c_\Gamma:=\frac{1-\Gamma}{\Gamma},
 \qquad
 D_\Gamma(r):=c_\Gamma e^r+1+r,
\end{equation}
and let $V_\Gamma$ be the unique nonnegative $C^1$ solution on $[0,\infty)$ of
\begin{equation}\label{eq:kzero-bellman}
 -V_\Gamma'(r)
 =\frac{r+\pos{r-1-D_\Gamma(r)V_\Gamma(r)}^2/4}{D_\Gamma(r)^2},
 \qquad
 \lim_{r\to\infty}V_\Gamma(r)=0.
\end{equation}
There is a unique $\Gamma_0\in(0,1)$ such that
\begin{equation}\label{eq:kzero-root}
 \Gamma_0+V_{\Gamma_0}(0)=1,
 \qquad
 \Gamma_0\approx0.7388957625.
\end{equation}
For every $\Gamma\in(0,\Gamma_0]$ and every profile $\lambda\in\Lambda$, $\ell_\Gamma(K_0,\lambda)\le1$.
Consequently, the prescribed arrival construction with $K_0$ is universally feasible at every ratio $\Gamma\le\Gamma_0$, and it satisfies the exact item-wise guarantee~\eqref{eq:itemwise} on every input.
\end{theorem}

Lemma~\ref{lem:kzero-bellman} below shows that~\eqref{eq:kzero-bellman}
is well posed and identifies $V_\Gamma(a)$ as the value of a control problem.
The theorem establishes $\Gamma_0$ as a lower bound on the feasibility
threshold of $K_0$.

We prove the theorem by bounding $\ell_\Gamma(K_0,\lambda)$ for an
arbitrary profile $\lambda$ in three steps.
First, we pass to exponential coordinates, in which the allocation, the survival factor, and the i.i.d.\ path become explicit in terms of a single function $h$, the cumulative hazard of the item's rank (Lemma~\ref{lem:kzero-coordinates}).
Second, we bound how far the item's probability threshold can lag behind the i.i.d.\ path (Lemma~\ref{lem:kzero-inverse}); the resulting bound depends on $\lambda$ only through $h$.
Third, we maximize this bound over all admissible paths $h$, which leads to a deterministic control problem whose value function is $V_\Gamma$ (Lemma~\ref{lem:kzero-bellman}).
Lemma~\ref{lem:kzero-root} then locates $\Gamma_0$.
For the i.i.d.\ path, the three steps reproduce $J(\Gamma)$ exactly; see~\eqref{eq:iid-load-r} below.

\subsubsection{Exponential coordinates}\label{sec:kzero-exp}

\begin{table}[t]
\centering
\small
\begin{tabular}{@{}c l l@{}}
\toprule
Coordinate & Meaning & Exponential form\\
\midrule
$t\in[0,1]$ & normalized time: $\Pp(\text{acceptance by time }t)=\Gamma t$ & $t_r=1-(1+r)e^{-r}$, $r\ge0$\\
$u\in[0,1]$ & maximum rank $U=1-F_{\max}(\OPT)$, uniform on $[0,1]$ & $u_a=1-e^{-a}$, $a\ge0$\\
$x\in[0,1]$ & level of a probability threshold; argument of $A_i$ & $1-x=e^{-r}R$ at time $t_r$\\
\bottomrule
\end{tabular}
\caption{The three coordinates used in the paper and their exponential forms. The i.i.d.\
threshold path satisfies $q(t_r)=u_r$, so the exponential time $r$ and the exponential rank $a$
are aligned along it, and the rescaled threshold complement $R$ equals one along this path
(Section~\ref{sec:kzero-lag}).}
\label{tab:coordinates}
\end{table}

Table~\ref{tab:coordinates} summarizes the coordinates used from now on.

Using the exponential coordinate introduced in
Section~\ref{sec:kzero-definition}, we parametrize the i.i.d.\
threshold path by
\begin{equation}\label{eq:kzero-r-coordinates}
 u_r:=1-e^{-r},
 \qquad
 t_r:=\phi(u_r)=1-(1+r)e^{-r},
 \qquad\text{so that}\quad
 q(t_r)=u_r.
\end{equation}
The map $r\mapsto t_r$ is an increasing $C^1$ bijection from $[0,\infty)$ onto $[0,1)$ with $\frac{\dd}{\dd r}t_r=re^{-r}$, and a direct computation gives
\begin{equation}\label{eq:kzero-S}
 1-\Gamma t_r=\Gamma e^{-r}D_\Gamma(r).
\end{equation}
Fix $\lambda\in\Lambda$.
The substitution $a=u_s$ in~\eqref{eq:profile-coordinates-main} gives
\begin{equation}\label{eq:kzero-h}
 h(r):=h_\lambda(u_r)=\int_0^r\lambda(u_s)\dd s,
 \qquad
 C(r):=C_\lambda(u_r)=\int_0^r\lambda(u_s)e^{-s}\dd s.
\end{equation}
By Lemma~\ref{lem:rank-representation}, $e^{-h(r)}$ is the probability that the item's value has rank greater than $u_r$.
Thus $h$ is the cumulative hazard of the item's rank on the exponential scale, and $h'(r)=\lambda(u_r)$ is its hazard rate.
In particular, $h$ is nondecreasing and $1$-Lipschitz with $h(0)=0$, so $0\le h(r)\le r$.
A rarely winning item has $h\approx0$, whereas $h'=1$ on an interval means that the item is the maximum whenever the maximum rank lies in the corresponding range.

\begin{lemma}[The load in exponential coordinates]\label{lem:kzero-coordinates}
Let $B:=B_{K_0,\lambda}$, $S_\lambda:=S_{K_0,\lambda}$, $y:=y_{K_0,\lambda}$, and $Y(r):=y(t_r)$.
For every $r\ge0$,
\begin{equation}\label{eq:kzero-r-allocation}
 B(t_r)=C(r)-e^{-r}h(r)=\int_0^re^{-s}h(s)\dd s.
\end{equation}
Consequently, $B$ is continuously differentiable on $(0,1)$, with $B'(t_r)=h(r)/r$ for $r>0$.
Moreover,
\begin{equation}\label{eq:kzero-survival}
 S_\lambda(t_r)=\sigma(r):=\exp\!\left(-\int_0^r\frac{h(s)}{D_\Gamma(s)}\dd s\right)
\end{equation}
and
\begin{equation}\label{eq:kzero-r-load}
 \ell_\Gamma(K_0,\lambda)
 =\Gamma\bigl(1-y(0)\bigr)
 +\int_0^\infty\sigma(r)\,\frac{\bigl(r-h(r)\bigr)\,e^r\,\bigl(1-Y(r)\bigr)}{D_\Gamma(r)^2}\dd r.
\end{equation}
\end{lemma}

\begin{proof}
The first identity in~\eqref{eq:kzero-r-allocation} is~\eqref{eq:kzero-allocation} at $t=t_r$.
The second follows by integrating by parts with $h'(s)=\lambda(u_s)$:
$\int_0^re^{-s}h(s)\dd s=-e^{-r}h(r)+\int_0^re^{-s}\lambda(u_s)\dd s$.
Since $s\mapsto e^{-s}h(s)$ is continuous, $r\mapsto B(t_r)$ is continuously differentiable with derivative $e^{-r}h(r)$; dividing by $\frac{\dd}{\dd r}t_r=re^{-r}$ gives the formula for $B'$.
Substituting $s=t_a$ in the definition~\eqref{eq:scalar-paths} of $S_\lambda$ and using~\eqref{eq:kzero-S},
\[
 \Gamma\int_0^{t_r}\frac{B'(s)}{1-\Gamma s}\dd s
 =\Gamma\int_0^r\frac{e^{-a}h(a)}{\Gamma e^{-a}D_\Gamma(a)}\dd a
 =\int_0^r\frac{h(a)}{D_\Gamma(a)}\dd a,
\]
which proves~\eqref{eq:kzero-survival}.
Finally, substitute $t=t_r$ in~\eqref{eq:load-parts}, using $\dd t=re^{-r}\dd r$, $1-B'(t_r)=(r-h(r))/r$, and $(1-\Gamma t_r)^2=\Gamma^2e^{-2r}D_\Gamma(r)^2$.
\end{proof}
For the i.i.d.\ benchmark, substitute $z_{K,1}=q$
into~\eqref{eq:rare-load} and set $t=t_r$.
Since $1-q(t_r)=e^{-r}$, this gives
\begin{equation}\label{eq:iid-load-r}
 J(\Gamma)=\Gamma+\int_0^\infty
                  \frac{r}{D_\Gamma(r)^2}\dd r.
\end{equation}
This is~\eqref{eq:kzero-r-load} with $h\equiv0$, $\sigma\equiv1$, $y(0)=0$, and $Y(r)$ replaced by the i.i.d.\ level $u_r$.

\subsubsection{The threshold lag}\label{sec:kzero-lag}

In exponential coordinates, the i.i.d.\ path satisfies $1-q(t_r)=e^{-r}$.
The next lemma bounds how far the threshold of a profile of positive size can fall behind it.

\begin{lemma}[Threshold lag]\label{lem:kzero-inverse}
For every $r\ge0$,
\begin{equation}\label{eq:kzero-lag-identity}
 a_\lambda(u_r)-A_\lambda\bigl(x_\lambda(u_r)\bigr)
 =e^{-r}\bigl(e^{h(r)}-1-h(r)\bigr)\ge0,
\end{equation}
and consequently
\begin{equation}\label{eq:kzero-inverse}
 1-Y(r)\le\bigl(1+h(r)\bigr)e^{-r}.
\end{equation}
\end{lemma}

Identity~\eqref{eq:kzero-lag-identity} compares the allocation $B(t_r)=a_\lambda(u_r)$ with the true area at the loss score $x_\lambda(u_r)$ of a value of rank $u_r$.
Inequality~\eqref{eq:kzero-inverse} says that the distance of the threshold to the top exceeds that of the i.i.d.\ path by at most the factor $1+h(r)$; for a rarely winning item, this factor is close to one.
Equivalently, the \emph{rescaled threshold complement} $R(r):=e^r(1-Y(r))$, which equals one along the i.i.d.\ path, satisfies $R(r)\le1+h(r)$.

\begin{proof}
Write $x_r:=x_\lambda(u_r)=1-e^{h(r)-r}$.
By~\eqref{eq:profile-area-main}, $A_\lambda(x_r)=C(r)+x_r-u_r=C(r)+e^{-r}-e^{h(r)-r}$.
Subtracting this from~\eqref{eq:kzero-r-allocation} gives~\eqref{eq:kzero-lag-identity}, and the right side is nonnegative because $e^h\ge1+h$.
Since $A_\lambda(x_r)\le B(t_r)$, the definition of the upper generalized inverse gives $Y(r)\ge x_r$.
Moreover, $A_\lambda(Y(r))=B(t_r)$ by Lemma~\ref{lem:threshold-paths}(a), applied to the area curve $A_\lambda$ and the allocation $B$, and $A_\lambda$ is $1$-Lipschitz.
Hence
\[
 Y(r)-x_r\ge A_\lambda\bigl(Y(r)\bigr)-A_\lambda(x_r)=e^{-r}\bigl(e^{h(r)}-1-h(r)\bigr),
\]
and therefore $1-Y(r)\le e^{h(r)-r}-e^{-r}\bigl(e^{h(r)}-1-h(r)\bigr)=(1+h(r))e^{-r}$.
\end{proof}

Inserting~\eqref{eq:kzero-inverse} into~\eqref{eq:kzero-r-load} and using $y(0)\ge0$ yields
\begin{equation}\label{eq:kzero-envelope}
 \ell_\Gamma(K_0,\lambda)
 \le\Gamma+\int_0^\infty\sigma(r)\,p_\Gamma\bigl(r,h(r)\bigr)\dd r,
 \qquad
 p_\Gamma(r,k):=\frac{(r-k)(1+k)}{D_\Gamma(r)^2}.
\end{equation}
Compared with the i.i.d.\ integrand $r/D_\Gamma(r)^2$ in~\eqref{eq:iid-load-r}, the integrand in~\eqref{eq:kzero-envelope} carries three factors that reflect the effects discussed in Section~\ref{sec:kzero-definition}: the factor $1+h(r)$ is the threshold lag caused by self-competition; the factor $\sigma(r)$ is the self-survival discount; and $r$ is replaced by $r-h(r)$ because the item's own allocation rate $B'(t_r)=h(r)/r$ slows the growth of $S_\lambda/S$ in~\eqref{eq:load-parts}.

\subsubsection{The control problem}\label{sec:kzero-control}

The right side of~\eqref{eq:kzero-envelope} depends on $\lambda$ only through $h$.
We bound it by letting an adversary choose $h(r)$ freely in $[0,r]$, ignoring that $h$ must also be nondecreasing and $1$-Lipschitz.
For $a\ge0$ and a measurable control $k:[a,\infty)\to\R$ with $0\le k(r)\le r$, define
\[
 \sigma_{a,k}(r):=\exp\!\left(-\int_a^r\frac{k(s)}{D_\Gamma(s)}\dd s\right),
 \qquad
 \mathcal P_a(k):=\int_a^\infty\sigma_{a,k}(r)\,p_\Gamma\bigl(r,k(r)\bigr)\dd r.
\]
Thus $\sigma=\sigma_{0,h}$, and the integral in~\eqref{eq:kzero-envelope} equals $\mathcal P_0(h)$.
The payoff $\mathcal P_a(k)$ is the contribution from $r\ge a$,
with the survival weight at $a$ normalized to one.
At $r$, choosing $k$ yields the payoff rate $p_\Gamma(r,k)$
and reduces the weight of future payoffs at rate $k/D_\Gamma(r)$.
If the remaining payoff has value $v\ge0$, maximizing the resulting
net payoff rate gives the Hamiltonian
\begin{equation}\label{eq:kzero-hamiltonian}
 H_\Gamma(r,v)
 :=\max_{0\le k\le r}
   \left\{p_\Gamma(r,k)-\frac{kv}{D_\Gamma(r)}\right\}
 =\frac{r+\pos{r-1-D_\Gamma(r)v}^2/4}{D_\Gamma(r)^2}.
\end{equation}
Indeed, $D_\Gamma(r)^2$ times the expression in braces equals $r+k\bigl(r-1-D_\Gamma(r)v\bigr)-k^2$, a concave quadratic in $k$ whose maximizer over $[0,r]$ is
\[
 k_*(r,v):=\tfrac12\pos{r-1-D_\Gamma(r)v}\in[0,r].
\]
Equation~\eqref{eq:kzero-bellman} is the Bellman equation $-V'=H_\Gamma(r,V)$ of this problem.
Since $k_*(r,v)=0$ whenever $r\le1$, the adversary keeps the hazard at zero among the highest values, which is the limit of a rarely winning item, and lets it grow only further down.
Let
\[
 M(r):=\max_{0\le k\le r}(r-k)(1+k)
 =\begin{cases}
 r,&0\le r\le1,\\
 (r+1)^2/4,&r\ge1,
 \end{cases}
\]
so that $0\le p_\Gamma(r,k)\le M(r)/D_\Gamma(r)^2$ for all admissible $k$.

\begin{lemma}[Value of the control problem]\label{lem:kzero-bellman}
For every $\Gamma\in(0,1)$, equation~\eqref{eq:kzero-bellman} has a unique nonnegative $C^1$ solution $V_\Gamma$ on $[0,\infty)$.
For every $a\ge0$,
\begin{equation}\label{eq:kzero-control-value}
 V_\Gamma(a)=\sup_k\mathcal P_a(k),
\end{equation}
where the supremum is over measurable controls $k$ with $0\le k(r)\le r$, and
\begin{equation}\label{eq:kzero-tail-majorant}
 0\le V_\Gamma(a)
 \le\int_a^\infty\frac{M(r)}{D_\Gamma(r)^2}\dd r
 \le c_\Gamma^{-2}\int_a^\infty M(r)e^{-2r}\dd r<\infty.
\end{equation}
\end{lemma}

The proof is a standard verification argument and is given in Appendix~\ref{app:kzero-proofs}.
Since $h$ is an admissible control, \eqref{eq:kzero-envelope} and Lemma~\ref{lem:kzero-bellman} give
\begin{equation}\label{eq:kzero-universal-bound}
 \ell_\Gamma(K_0,\lambda)\le\Gamma+V_\Gamma(0)
 \qquad\text{for every }\lambda\in\Lambda\text{ and }\Gamma\in(0,1).
\end{equation}
On the other hand, the zero control gives $V_\Gamma(0)\ge\mathcal P_0(0)=J(\Gamma)-\Gamma$ by~\eqref{eq:iid-load-r}.

\subsubsection{The constant \texorpdfstring{$\Gamma_0$}{Gamma0}}\label{sec:kzero-constant}

\begin{lemma}\label{lem:kzero-root}
The function $\Psi(\Gamma):=\Gamma+V_\Gamma(0)$ is continuous and strictly increasing on $(0,1)$, with $\Psi(\Gamma)\to0$ as $\Gamma\downarrow0$ and $\Psi(\Gamma)\to\infty$ as $\Gamma\uparrow1$.
\end{lemma}

The proof, which differentiates the payoff of each fixed control with respect to $\Gamma$, is given in Appendix~\ref{app:kzero-proofs}.

\begin{proof}[Proof of Theorem~\ref{thm:kzero-analytic}]
Lemma~\ref{lem:kzero-bellman} shows that $V_\Gamma$ is well defined.
By Lemma~\ref{lem:kzero-root} and the intermediate value theorem, there is a unique $\Gamma_0\in(0,1)$ with $\Psi(\Gamma_0)=1$.
For $\Gamma\le\Gamma_0$ and $\lambda\in\Lambda$, \eqref{eq:kzero-universal-bound} gives $\ell_\Gamma(K_0,\lambda)\le\Psi(\Gamma)\le\Psi(\Gamma_0)=1$.
Proposition~\ref{prop:scalar-feasibility} completes the proof.
\end{proof}

Since $\Psi(\Gamma)\ge J(\Gamma)$ for every $\Gamma$, the root satisfies $\Gamma_0\le\Gamma_\star$, in agreement with Proposition~\ref{prop:iid-guarantee}. To estimate $\Gamma_0$ numerically, we truncate the control problem
at a horizon $\bar r\ge1$. By~\eqref{eq:kzero-tail-majorant},
this changes its value by at most
\begin{equation}\label{eq:kzero-truncation}
 c_\Gamma^{-2}\int_{\bar r}^\infty M(r)e^{-2r}\dd r
 =c_\Gamma^{-2}e^{-2\bar r}\,\frac{2\bar r^2+6\bar r+5}{16}.
\end{equation}
The truncated value is the solution at $r=0$ of~\eqref{eq:kzero-bellman} on $[0,\bar r]$ with terminal condition $V(\bar r)=0$ (see Appendix~\ref{app:kzero-proofs}).
Integrating this terminal value problem numerically with $\bar r=40$ and solving $\Psi(\Gamma)=1$ by a bracketing root finder gives the decimal in~\eqref{eq:kzero-root}; the digits are unchanged for $\bar r\in\{20,30,60\}$.
This decimal approximates the exact constant defined by~\eqref{eq:kzero-root}.

\subsection{Infeasibility of \texorpdfstring{$K_0$}{K0} at the I.I.D.\ Constant \texorpdfstring{$\Gamma_\star$}{Gamma-star}}
\label{sec:kzero-limit}

The small-profile tests of Section~\ref{sec:small-profiles} are only necessary conditions.
The following explicit profile shows that $K_0$ is not universally feasible at $\Gamma_\star$:
\begin{equation}\label{eq:k0-explicit-counterexample}
 \lambda(u):=
 \begin{cases}
  1/100,&0\le u\le9/10,\\
  1,&9/10<u\le1.
 \end{cases}
\end{equation}

\begin{proposition}\label{prop:k0-counterexample}
For the profile~\eqref{eq:k0-explicit-counterexample}, $\ell_{\Gamma_\star}(K_0,\lambda)>1$.
Consequently, $K_0$ is universally feasible at every $\Gamma\le\Gamma_0$ but at no $\Gamma\ge149/200$; in particular, it is not universally feasible at $\Gamma_\star$.
\end{proposition}

\begin{proof}
For this profile, every ingredient of~\eqref{eq:kzero-r-load} is explicit.
Put $r_0:=\log10$ and $\mu:=\mu_\lambda=109/1000$.
By~\eqref{eq:kzero-h} and~\eqref{eq:kzero-r-allocation},
\[
 h(r)=\frac{r}{100},\quad B(t_r)=\frac{t_r}{100}\quad(r\le r_0);
 \qquad
 h(r)=r-\frac{99}{100}\,r_0,\quad B(t_r)=\mu-\bigl(1+h(r)\bigr)e^{-r}\quad(r\ge r_0).
\]
By Lemma~\ref{lem:profile-bijection}, $A_\lambda$ coincides on $[0,x_*]$, where $x_*:=1-10^{-99/100}$, with the constant-profile area curve $A_{1/100}$ given by~\eqref{eq:constant-profile-area}, and it continues with slope one on $[x_*,1]$.
Since $A_{1/100}$ is $1$-Lipschitz with slope less than one, this means
\[
 A_\lambda(x)=\max\Bigl\{x-\tfrac{99}{100}+\tfrac{99}{100}(1-x)^{100/99},\ \mu+x-1\Bigr\},
 \qquad x\in[0,1).
\]
Moreover, $y(0)=0$ because $\lambda>0$.
Thus $\sigma$ is explicit up to a one-dimensional integral, and $Y(r)=A_\lambda^{-1}(B(t_r))$ is explicit up to the inversion of $A_\lambda$.
Set $\widehat\Gamma:=149/200$.
Direct evaluation of~\eqref{eq:kzero-r-load} and~\eqref{eq:iid-load} gives
\begin{equation}\label{eq:k0-explicit-certificate}
 \ell_{\widehat\Gamma}(K_0,\lambda)>1.0028,
 \qquad
 J(\widehat\Gamma)<0.999643<1.
\end{equation}
Numerically, the two quantities equal $1.00787\ldots$ and $0.99909\ldots$, so both inequalities hold with margins larger than $5\times10^{-4}$.
Since $J$ is increasing and $J(\Gamma_\star)=1$ by~\eqref{eq:iid-constant},
the second inequality implies $\widehat\Gamma<\Gamma_\star$.
The first inequality and Lemma~\ref{lem:gamma-monotone} give
$\ell_\Gamma(K_0,\lambda)>1$ for every $\Gamma\ge\widehat\Gamma$,
in particular for $\Gamma=\Gamma_\star$.
Proposition~\ref{prop:scalar-feasibility} therefore rules out universal
feasibility at every such $\Gamma$, while
Theorem~\ref{thm:kzero-analytic} gives universal feasibility for
every $\Gamma\le\Gamma_0$.
\end{proof}

\paragraph{How self-competition defeats $K_0$.}
On the exponential scale, the profile~\eqref{eq:k0-explicit-counterexample} has hazard rate $1/100$ up to level $r_0:=\log10$ and hazard rate one afterwards.
The item is thus rarely the maximum when the maximum is large (maximum rank at most $9/10$), and it is always the maximum when the maximum is small.
By Lemma~\ref{lem:rank-representation}, its value has rank above $9/10$ with probability $e^{-h(r_0)}=10^{-1/100}\approx0.977$.
All these values have the same loss score
\[
 x_*:=1-10^{-99/100}\approx0.8977,
\]
because $x_\lambda$ is constant on $(9/10,1)$.
Indeed, the other items place no probability on this range of values, so such a value beats them exactly when all of them lie below the range, regardless of its own rank.
This is self-competition in its most extreme form.
Consequently, the area curve has slope one on $[x_*,1]$ by~\eqref{eq:profile-extension-main}.
The threshold reaches $x_*$ at a level $r_1\approx2.291$, and Lemma~\ref{lem:kzero-inverse} holds with equality for $r\ge r_0$:
\[
 1-Y(r)=\bigl(1+h(r)\bigr)e^{-r},
 \qquad
 h(r)=r-\tfrac{99}{100}\,r_0
 \qquad(r\ge r_0).
\]
While the i.i.d.\ path satisfies $1-q(t_r)=e^{-r}$, the threshold of this item thus stays behind by a factor that grows linearly in $r$.
The reason is a mismatch of schedules.
The kernel $K_0$ releases the winning probability at ranks $u\in(9/10,1]$ as the i.i.d.\ threshold sweeps the levels $[u,1]$, which is the right schedule when values of rank $u$ have loss score close to $u$, as they do for rarely winning items.
Here, these values already qualify at the lower level $x_*$, so the same allocation buys a lower threshold.
The late increments of the threshold are the expensive ones in the
load formula~\eqref{eq:scalar-load}, and the self-survival discount
does not compensate for them.

\begin{remark}[Numerical observations]
The mechanism above appears to determine the feasibility threshold of $K_0$, and Theorem~\ref{thm:kzero-analytic} appears to be close to tight.
The observations in this remark are numerical and are not used in the proofs.
Consider the two-phase profiles $\lambda=\delta\ind_{[0,u_0]}+\ind_{(u_0,1]}$ with small $\delta>0$.
For them, Lemma~\ref{lem:kzero-inverse} holds with equality in the second phase, and as $\delta\downarrow0$ their loads converge to $\Gamma+\mathcal P_0(k)$ for the control $k(r)=\pos{r+\log(1-u_0)}$.
Numerically, $\delta=10^{-4}$ and $u_0=0.86$ give $\ell_{0.74}(K_0,\lambda)\approx1.00018>1$, so the exact feasibility threshold of $K_0$ appears to lie in $[\Gamma_0,0.74)$, within about $10^{-3}$ of $\Gamma_0$.
The remaining slack in Theorem~\ref{thm:kzero-analytic} comes from the threshold-lag step rather than from the relaxation of Section~\ref{sec:kzero-control}.
Indeed, at $\Gamma_0$ the maximizing feedback control $k_*(r,V_{\Gamma_0}(r))$ is numerically nondecreasing with slope at most $0.57$, and it vanishes for $r\le1.46$.
With a small regularization $\varepsilon r$ added, it is therefore the cumulative hazard $h$ of a genuine profile whose bound~\eqref{eq:kzero-envelope} is within $10^{-4}$ of $\Psi(\Gamma_0)=1$.
For that profile, however, Lemma~\ref{lem:kzero-inverse} is not tight, and its actual load at $\Gamma_0$ is about $0.989$.
Conversely, the two-phase profiles make Lemma~\ref{lem:kzero-inverse} tight, but their loads at $\Gamma_0$ reach only about $0.998$.
In any case, Proposition~\ref{prop:k0-counterexample} shows that the shortfall relative to $\Gamma_\star$ is a property of $K_0$ itself, not of its analysis: reaching $\Gamma_\star$ requires changing how the kernel treats profiles of positive size.
\end{remark}

\section{A Kernel with Competitive Ratio Above \texorpdfstring{$0.745$}{0.745}}
\label{sec:final-kernel}

We now perturb the baseline kernel $K_0$ to obtain a kernel $K_1$ with
competitive ratio strictly greater than $0.745$.
We first describe the numerical observation that motivates the construction,
then specify the kernel exactly, and finally state its guarantee
and outline its proof.

\subsection{Motivation}\label{sec:final-motivation}

Section~\ref{sec:kzero-limit} shows that $K_0$ fails at $\Gamma_\star$ and traces the failure to a
mismatch of schedules: $K_0$ releases the winning probability at rank $u$ only after the i.i.d.\
threshold has passed $u$, which is too late for an item whose values already qualify at a
lower level. A simple remedy is to mix $K_0$ with the independent kernel
$K_{\mathrm{ind}}(t,u)=t$ of Section~\ref{sec:independent}. In the mixture, a small fraction
of the winning probability at each maximum rank is allocated uniformly over time, so an item
receives some allocation before the i.i.d.\ threshold reaches its ranks, which can reduce the
threshold lag. Numerically, the mixture $0.972K_0+0.028K_{\mathrm{ind}}$ appears to be
universally feasible at ratio $0.74$, above the numerically observed threshold of $K_0$
reported at the end of Section~\ref{sec:kzero-limit}.

The mixture uses the same weight at every rank and time. A finer numerical search transfers
mass between selected cells of a grid in the rank and time coordinates while preserving both
marginals of the coupling. It leads to the kernel $K_1$ below, for which
Theorem~\ref{thm:certified-kernel} gives a rigorous guarantee above $0.745$. Every kernel
gives a constant profile the same allocation (Section~\ref{sec:iid-blind}), so the transfers
act only on nonconstant profiles; unlike $K_0$, the kernel $K_1$ does not make every
small-profile threshold path coincide with $q$.

\subsection{\texorpdfstring{The Kernel $K_1$}{The Kernel K1}}\label{sec:final-definition}

We specify $K_1$ as a perturbation of the coupling induced by $K_0$.
Use the coordinates
\[
 u_a=1-e^{-a},\qquad t_r=1-(1+r)e^{-r},\qquad a,r\ge0,
\]
from~\eqref{eq:kzero-r-coordinates}.
Here $a$ parametrizes maximum rank and $r$ parametrizes normalized
algorithmic time.
The probabilistic description in Section~\ref{sec:kzero-definition}
gives the baseline joint density $p_0(a,r)=e^{-r}\ind\{a<r\}$: the pair $(a,r)=(\xi,\xi+\xi')$ has density $e^{-a}e^{-(r-a)}$ on $\{a<r\}$.
Its marginals are $e^{-a}$ and $re^{-r}$, which become uniform under
$a\mapsto u_a$ and $r\mapsto t_r$.

Let $\Delta=2/5$ and set
\[
 I_i=[i\Delta,(i+1)\Delta),\quad 0\le i<16,
 \qquad J_j=[j\Delta,(j+1)\Delta),\quad 0\le j<25.
\]
Write $E_{ij}$ for the $16\times25$ matrix unit and put
\[
 \vartheta=\frac{999999999}{10^9},\qquad
 \mathcal I_1=\{1,2,3,4\},\quad
 \mathcal I_2=\{5,6,7\},\quad
 \mathcal I_3=\{8,\ldots,15\}.
\]
The exact perturbation of the cell masses is
\begin{equation}\label{eq:cert-delta}
 \delta=\frac{\vartheta}{10^{21}}
 \sum_{g=1}^3\sum_{i\in\mathcal I_g}\sum_{j=1}^{23}N_{g,j}
 (E_{ij}+E_{0,24}-E_{i,24}-E_{0,j}),
\end{equation}
where the $69$ integers $N_{g,j}$ are listed in
Table~\ref{tab:cert-kernel} in Appendix~\ref{app:cert-coefficients}.
Each summand transfers mass between the four cells in rows $(0,i)$
and columns $(j,24)$ according to
\[
 \begin{pmatrix}-1&1\\1&-1\end{pmatrix},
\]
preserving every row and column sum.

Let $\delta_{ij}$ be the entries of $\delta$, and define
\[
 \omega_j=e^{-j\Delta}-e^{-(j+1)\Delta},\qquad
 w_{ij}=\frac{\delta_{ij}}{\Delta\omega_j}.
\]
Set
\begin{equation}\label{eq:cert-p}
 p_1(a,r)=e^{-r}\left[
 \ind\{a<r\}+\sum_{i=0}^{15}\sum_{j=0}^{24}
 w_{ij}\ind\{a\in I_i\}\ind\{r\in J_j\}\right].
\end{equation}
The kernel is
\begin{equation}\label{eq:cert-K}
 K_1(t_r,u_a):=e^a\int_0^r p_1(a,v)\dd v,
 \qquad a,r\ge0.
\end{equation}
Set $K_1(1,u)=1$ for $u<1$ and $K_1(t,1)=t$.
These formulas and the integer table specify a single kernel for every
item and every input.
The perturbation is supported on the rectangle $[0,32/5)\times[0,10)$ of the $(a,r)$-plane, and for each fixed $a$ it integrates to zero over $r$.
Hence $K_1(t,u)=K_0(t,u)$ whenever $u\ge1-e^{-32/5}$ or $t\ge t_{10}$ (Lemma~\ref{lem:cert-kernel-valid}).

\subsection{\texorpdfstring{The Guarantee}{The Guarantee}}\label{sec:final-guarantee}

\begin{figure}[tb]
\centering
\begin{tikzpicture}[font=\footnotesize]
  \def\w{68mm}
  \node[align=flush center] (hl) at (37mm,0) {\small\textbf{Baseline kernel $K_0$}\\[-1pt]
    analytic proof (Thm.~\ref{thm:kzero-analytic})};
  \node[align=flush center] (hr) at (123mm,0) {\small\textbf{Final kernel $K_1$}\\[-1pt]
    computer-assisted proof (Thm.~\ref{thm:certified-kernel})};
  \foreach \row/\y in {1/-13.5, 2/-30.5, 3/-47.5, 4/-64.5} {
    \coordinate (L\row) at (37mm,\y mm); \coordinate (R\row) at (123mm,\y mm); }
  \node[flowbox, text width=\w, minimum height=12.5mm] (l1) at (L1)
    {\textbf{Load along the hazard path.} In exponential coordinates, the load is an integral
     driven by the cumulative hazard $h$ (Lemma~\ref{lem:kzero-coordinates}).};
  \node[flowbox, text width=\w, minimum height=12.5mm] (r1) at (R1)
    {\textbf{Load along the hazard path.} The same representation
     (Lemma~\ref{lem:cert-load-identity}); the perturbation changes the allocation only through
     four block masses (App.~\ref{sec:cert-finite}).};
  \node[flowbox, text width=\w, minimum height=12.5mm] (l2) at (L2)
    {\textbf{Threshold lag.} The rescaled threshold complement satisfies $R\le1+h$
     (Lemma~\ref{lem:kzero-inverse}).};
  \node[flowbox, text width=\w, minimum height=12.5mm] (r2) at (R2)
    {\textbf{Regional threshold bounds.} Several estimates for $R$, plus a correction term
     that recovers slack (App.~\ref{app:cert-inverse-bounds}--\ref{app:cert-credit}).};
  \node[flowbox, text width=\w, minimum height=12.5mm] (l3) at (L3)
    {\textbf{Worst path.} An adversary chooses $h(r)\in[0,r]$; the Bellman equation gives the
     value $V_\Gamma$ (Lemma~\ref{lem:kzero-bellman}).};
  \node[flowbox, text width=\w, minimum height=12.5mm] (r3) at (R3)
    {\textbf{Worst path in a region.} Backward comparison over all admissible paths, with
     state bands or a potential (App.~\ref{app:cert-regional-comparison},
     \ref{app:cert-direct-updates}--\ref{app:cert-potential-construction}).};
  \node[flowbox, line width=0.9pt, text width=\w, minimum height=12.5mm] (l4) at (L4)
    {\textbf{Constant.} $\Gamma_0\approx0.7389$ is the root of $\Gamma+V_\Gamma(0)=1$
     (Lemma~\ref{lem:kzero-root}).};
  \node[flowbox, line width=0.9pt, text width=\w, minimum height=12.5mm] (r4) at (R4)
    {\textbf{Finite cover.} Every nonzero mass vector lies in a region with a certified bound
     below one (Lemma~\ref{lem:cert-finite-data}); the zero vector is treated
     by~\eqref{eq:cert-zero-scale}.};
  \foreach \i/\j in {1/2,2/3,3/4} { \draw[flowarrow] (l\i) -- (l\j); \draw[flowarrow] (r\i) -- (r\j); }
  \foreach \i in {1,2,3,4} { \draw[-{Stealth[length=1.8mm]}, dashed, draw=black!50] (l\i) -- (r\i); }
\end{tikzpicture}
\caption{Proof architecture. Each step of the proof of Theorem~\ref{thm:certified-kernel}
(right) refines the corresponding step of the proof of Theorem~\ref{thm:kzero-analytic}
(left).}
\label{fig:proof-architecture}
\end{figure}
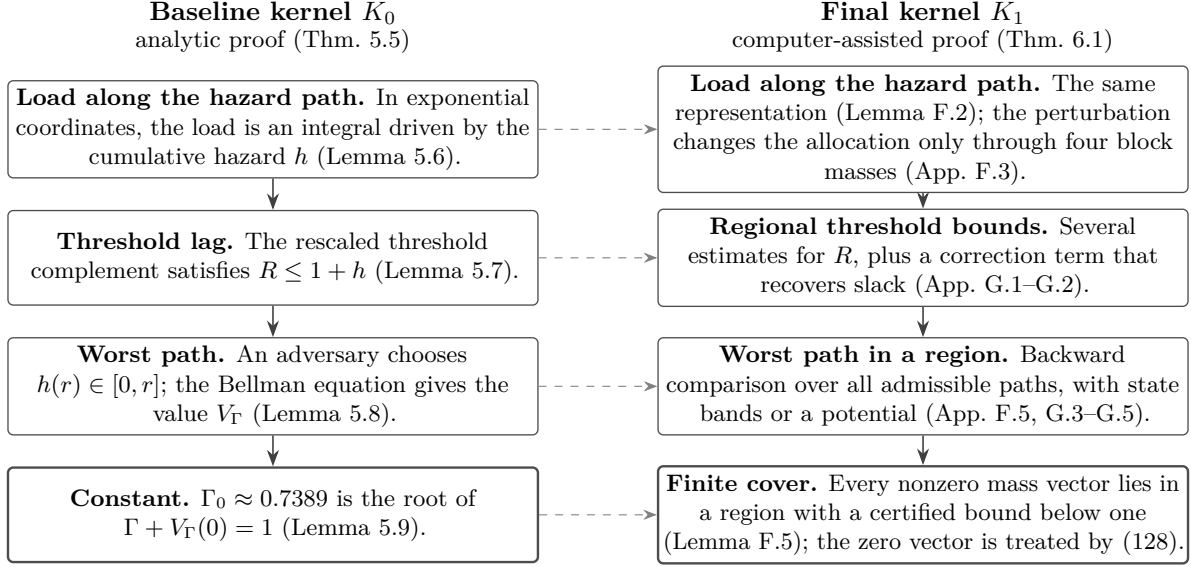

Set
\begin{equation}\label{eq:cert-gamma}
 \Gamma_1:=\frac{74500001}{100000000}=0.74500001>0.745.
\end{equation}
\begin{theorem}[Universal guarantee for $K_1$]\label{thm:certified-kernel}
The function $K_1$ defined in~\eqref{eq:cert-K} is an allocation kernel
and satisfies
\begin{equation}\label{eq:cert-load-target}
 \sup_{\lambda\in\Lambda}\ell_{\Gamma_1}(K_1,\lambda)<1.
\end{equation}
Consequently, by Proposition~\ref{prop:scalar-feasibility},
the prescribed arrival construction is universally feasible at
$\Gamma_1$ and satisfies the exact item-wise
guarantee~\eqref{eq:itemwise} on every atomless input.
Moreover, it yields an order-selection algorithm for arbitrary
independent nonnegative values with
\[
 \E[\ALG]\ge\Gamma_1\E[\OPT].
\]
\end{theorem}

The proof, including the finite numerical certificate, is given in
Appendix~\ref{app:cert-proof}.
Appendix~\ref{app:cert-details} provides the supporting estimates and
exact arithmetic.
A complete verification package is
\href{https://github.com/ahqspxy/order-selection-certificate}
{publicly available}.

\paragraph{Proof architecture.}
Figure~\ref{fig:proof-architecture} compares the proof with that of
Theorem~\ref{thm:kzero-analytic}. In both proofs, the load is an integral along the
cumulative hazard $h$ of the profile in exponential coordinates. For $K_1$, the perturbation
changes the allocation by a term that depends on the profile only through four block masses,
the integrals of the hazard rate over the blocks of ranks on which the perturbation does not
depend on the rank (Appendix~\ref{sec:cert-finite}). Where the analysis of $K_0$ uses
the single inequality $R\le1+h$ for the rescaled threshold complement, the certificate uses
several regional estimates, together with a correction term that recovers part of the slack
in this step (Appendices~\ref{app:cert-inverse-bounds} and~\ref{app:cert-credit}). Where the
analysis of $K_0$ lets an adversary choose $h(r)\in[0,r]$ freely, the certificate bounds the
load along all admissible paths by a backward comparison, with state bands or a state-dependent
potential, which keeps the information that $h$ is nondecreasing and $1$-Lipschitz (Appendix~\ref{app:cert-regional-comparison}). Finally,
instead of a scalar equation for the constant, the certificate covers the set $[0,1]^4$ of
mass vectors (block masses divided by block lengths) by finitely many regions, each with a
certified bound below one
(Appendix~\ref{sec:cert-results}).

\section{Concluding Remarks}\label{sec:conclusion}

Local constructions in the arrival-time framework are exactly the allocation kernels, and this
description reduces the design of a local arrival-time algorithm to the choice of a single
coupling of two uniform random variables. Using it, we obtained a kernel whose competitive
ratio exceeds $0.745$, within $4.5\times10^{-4}$ of the upper bound $\Gamma_\star$.

Two questions remain open. First, is some allocation kernel universally feasible at
$\Gamma_\star$? Proposition~\ref{prop:iid-calibration} gives a necessary condition, and
Proposition~\ref{prop:k0-counterexample} shows that $K_0$, the kernel singled out by the
strongest form of this condition, falls short; the kernel $K_1$ shows that giving up the
exact matching of Proposition~\ref{prop:iid-path-kernel} improves on $K_0$. Second, the guarantee for $K_1$ rests on a
computer-assisted certificate. An analytic proof of a ratio above $\Gamma_0$, for instance
for a structured family of kernels whose analysis reduces to a tractable control problem as
in Section~\ref{sec:kzero-analysis}, would clarify which features of $K_1$ are essential.

\section*{AI Disclosure}

AI systems were used in both the development of the results presented in this paper and the preparation of the manuscript.

The characterization of local constructions via allocation kernels (Section~\ref{sec:local}) and the construction of the baseline kernel $K_0$ (Section~\ref{sec:k0}) were developed by the authors. The improved construction achieving the ratio $\Gamma_1>0.745$, namely the kernel $K_1$ in Section~\ref{sec:final-kernel}, was developed by AI systems, along with its analysis and the computer-assisted certificate in Appendices~\ref{app:cert-proof} and~\ref{app:cert-details}. Beyond the motivation provided in Section~\ref{sec:final-motivation}, we do not yet have a conceptual understanding of this construction. Although the performance guarantee for $K_1$ is rigorous, we view the construction primarily as compelling numerical evidence that the order-selection prophet inequality admits the optimal i.i.d.\ ratio $\Gamma_\star$, rather than as a definitive explanation of why this ratio is attainable. We regard it as an intermediate step toward a proof of this optimal ratio.

\newpage

\appendix

\section{The Arrival Construction}\label{app:algorithm}

This appendix explains where the arrival laws~\eqref{eq:arrivals} come from and proves Theorem~\ref{thm:bc}.
Throughout, the allocation functions satisfy~\eqref{eq:allocation}; in particular, each $B_i$ is $1$-Lipschitz.
Recall from Section~\ref{sec:model} that $U_i=1-\prod_{j\ne i}F_j(V_i)$ is the \emph{loss score} of item $i$: it is the conditional probability that item $i$ is not the maximum, given its value.
Integrals with respect to $B_i$, $y_i$, and $G_i$ are Lebesgue--Stieltjes integrals, and no arrival mass is assigned to $t=0$.

\subsection{Threshold Paths}

The probability threshold $y_i=A_i^\leftarrow\circ B_i$ converts allocated area into threshold levels.
The following lemma collects the properties of this conversion that we use.
We state it for a general pair of an area curve and an allocation, because Section~\ref{sec:kzero-analysis} and Appendix~\ref{app:cert-proof} also apply it to single profiles.

\begin{lemma}[Threshold paths]\label{lem:threshold-paths}
Let $U$ be a random variable with values in $[0,1]$, let $A(x):=\E[(x-U)_+]=\int_0^x\Pp(U<z)\dd z$ for $x\in[0,1]$, and let $\mu:=A(1)$.
Let $B:[0,1]\to[0,\mu]$ be continuous and nondecreasing with $B(0)=0$ and $B(1)=\mu$, and put $y:=A^\leftarrow\circ B$.
\begin{enumerate}
\item[(a)] The path $y$ is continuous and nondecreasing, $y(0)=\max\{x\in[0,1]:A(x)=0\}$, $y(1)=1$, and $A(y(t))=B(t)$ for every $t\in[0,1]$.
\item[(b)] $U\ge y(0)$ almost surely.
\item[(c)] For every Borel function $g\ge0$ and all $0\le s\le t\le1$,
\begin{equation}\label{eq:stieltjes-substitution}
 \int_{(s,t]}g\bigl(y(\tau)\bigr)\dd y(\tau)=\int_{y(s)}^{y(t)}g(x)\dd x.
\end{equation}
\item[(d)] As measures on $(0,1]$,
\begin{equation}\label{eq:chain}
 \dd B(\tau)=\Pp\bigl(U\le y(\tau)\bigr)\dd y(\tau).
\end{equation}
\end{enumerate}
\end{lemma}

For an item $i$ of an input, the lemma applies with $U=U_i$, $A=A_i$, $B=B_i$, and $y=y_i$.

\begin{proof}
If $\mu=0$, then $U=1$ almost surely, because $\mu=\E[1-U]$; hence $A\equiv0$, $B\equiv0$, and $y\equiv A^\leftarrow(0)=1$, and all four claims are trivial.
Assume $\mu>0$, and let $x_0:=\max\{x\in[0,1]:A(x)=0\}<1$.
The function $A$ is strictly increasing on $[x_0,1]$.
Indeed, if $x_0\le a<b$ and $A(a)=A(b)$, then the slope of the convex function $A$ on $[a,b]$ is zero, so its slopes on $[0,a]$ are nonpositive; since $A$ is nondecreasing, it is constant on $[0,b]$, and $A(b)=A(0)=0$ contradicts $b>x_0$.
Hence $A$ maps $[x_0,1]$ homeomorphically onto $[0,\mu]$, and for $b\in[0,\mu]$ the set $\{x:A(x)\le b\}$ is the interval $[0,(A|_{[x_0,1]})^{-1}(b)]$.
Thus $A^\leftarrow=(A|_{[x_0,1]})^{-1}$ on $[0,\mu]$, and (a) follows because $B$ is continuous and nondecreasing with $B(0)=0$ and $B(1)=\mu$; in particular, $y(0)=x_0$.

For (b), $0=A(x_0)=\int_0^{x_0}\Pp(U<z)\dd z$ forces $\Pp(U<z)=0$ for every $z<x_0$, and hence $\Pp(U<x_0)=0$.

For (c), both sides of~\eqref{eq:stieltjes-substitution} define finite measures in $g$ that are concentrated on $(y(s),y(t)]$; for the left side, note that $\{\tau\in(s,t]:y(\tau)=y(s)\}$ is an interval on which $y$ is constant and therefore has $\dd y$-measure zero.
By the $\pi$--$\lambda$ theorem, it suffices to take $g=\ind_{(y(s),c]}$ with $y(s)\le c\le y(t)$.
By continuity and monotonicity of $y$, the set $\{\tau\in(s,t]:y(\tau)\le c\}$ is an interval $(s,\tau_c]$ with $y(\tau_c)=c$, so both sides of~\eqref{eq:stieltjes-substitution} equal $c-y(s)$.

For (d), recall that $A(x)=\int_0^x\Pp(U<z)\dd z$.
By (a), (c), and $A(y(0))=0$,
\[
 B(t)=A\bigl(y(t)\bigr)
 =\int_{y(0)}^{y(t)}\Pp(U<x)\dd x
 =\int_{(0,t]}\Pp\bigl(U<y(\tau)\bigr)\dd y(\tau).
\]
Replacing $<$ by $\le$ changes the integrand only at times $\tau$ with $y(\tau)$ in the countable set of atoms of $U$.
By~\eqref{eq:stieltjes-substitution}, the set of such times has $\dd y$-measure zero.
\end{proof}

\subsection{Deriving the Arrival Laws from the Item-Wise Guarantee}

We now explain how the item-wise guarantee~\eqref{eq:itemwise} leads to the formulas in~\eqref{eq:arrivals}.
This derivation is heuristic; the next subsection verifies the resulting construction rigorously.

Call an arrival of item $i$ at time $\tau$ \emph{qualifying} if $U_i\le y_i(\tau)$, whether or not the algorithm has already stopped.
Let $S_j(\tau)$ denote the probability that item $j$ has no qualifying arrival in $(0,\tau]$; Section~\ref{app:bc-proof} verifies that these functions coincide with those defined in~\eqref{eq:arrivals}.
Since items arrive independently, an arrival of item $i$ at time $\tau$ is inspected with probability $\prod_{j\ne i}S_j(\tau)$, and the measure
\[
 \prod_{j\ne i}S_j(\tau)\,\dd G_i(\tau)
\]
describes when item $i$ is inspected; we call it the \emph{effective inspection measure} of item $i$.
Given $U_i=x$, item $i$ is accepted if and only if it is inspected at a time $\tau$ with $y_i(\tau)\ge x$.
Since $y_i$ is continuous and nondecreasing with $y_i(1)=1$, these times form an interval ending at $1$, and the item-wise guarantee requires
\[
 \int_{\{\tau:\,y_i(\tau)\ge x\}}\prod_{j\ne i}S_j(\tau)\,\dd G_i(\tau)
 =\Gamma(1-x)
 =\Gamma\int_{\{\tau:\,y_i(\tau)\ge x\}}\dd y_i(\tau)
\]
for all $x\in[y_i(0),1]$, which by Lemma~\ref{lem:threshold-paths}(b) are the only values that $U_i$ takes.
This holds for all such $x$ if the effective inspection measure equals $\Gamma\dd y_i$:
\begin{equation}\label{eq:effective-inspection}
 \prod_{j\ne i}S_j(\tau)\,\dd G_i(\tau)=\Gamma\dd y_i(\tau).
\end{equation}
In words, item $i$ is inspected at rate $\Gamma$ per unit increase of its threshold level.

It remains to determine the survival functions.
A qualifying arrival of item $i$ occurs in $\dd\tau$ with probability $\Pp(U_i\le y_i(\tau))\dd G_i(\tau)$.
By~\eqref{eq:effective-inspection} and~\eqref{eq:chain}, this equals $\Gamma\dd B_i(\tau)/\prod_{j\ne i}S_j(\tau)$.
Writing $S^{\mathrm{tot}}:=\prod_jS_j$, we obtain
\[
 -\frac{\dd S_i(\tau)}{S_i(\tau)}=\Gamma\frac{\dd B_i(\tau)}{S^{\mathrm{tot}}(\tau)},
 \qquad\text{and, summing over $i$,}\qquad
 -\dd S^{\mathrm{tot}}(\tau)=\Gamma\,\dd\Bigl(\sum_iB_i\Bigr)(\tau).
\]
Under the normalization $\sum_iB_i(\tau)=\tau$, the probability that no item has qualified by time $\tau$ is therefore $S^{\mathrm{tot}}(\tau)=1-\Gamma\tau=S(\tau)$.
This is the sense in which $t$ is a normalized time: it measures accepted probability, divided by $\Gamma$.
Substituting $S^{\mathrm{tot}}=S$ back gives the formulas for $S_i$ and $G_i$ in~\eqref{eq:arrivals}.

\subsection{Proof of Theorem~\ref{thm:bc}}\label{app:bc-proof}

\begin{proof}[Proof of Theorem~\ref{thm:bc}]
Assume~\eqref{eq:allocation} and $L_i\le1$ for every item $i$, and define $S$, $S_i$, and $G_i$ by~\eqref{eq:arrivals}.
Each $S_i$ is absolutely continuous, nonincreasing, and positive, with $S_i(0)=1$, and
\begin{equation}\label{eq:survival-product}
 \prod_iS_i(t)
 =\exp\!\left(-\Gamma\int_0^t\frac{\dd\bigl(\sum_iB_i\bigr)(s)}{1-\Gamma s}\right)
 =\exp\bigl(\log(1-\Gamma t)\bigr)
 =S(t).
\end{equation}
Since each $S_j\le1$, this gives $S_i\ge S\ge1-\Gamma>0$, so $S/S_i=\prod_{j\ne i}S_j$ is well defined.
The measure $\dd G_i=\Gamma(S_i/S)\dd y_i$ on $(0,1]$ is nonnegative with total mass $L_i\le1$; assigning mass $1-L_i$ to $\dagger$ makes $G_i$ a probability distribution.
Because $y_i$ is continuous, the restriction of $G_i$ to $(0,1]$ has no atoms.
Since arrival times are independent, two items arrive at the same time in $(0,1]$ with probability zero, so tie-breaking is irrelevant.

For each item $i$ and $t\in[0,1]$, let $Q_i(t)$ be the event that item $i$ has a qualifying arrival in $(0,t]$, that is, $T_i\in(0,t]$ and $U_i\le y_i(T_i)$.
These events depend on the pairs $(T_i,V_i)$ and are therefore independent across items.

\emph{First claim.}
Since $T_i$ is independent of $V_i$, Lemma~\ref{lem:threshold-paths}(d) gives
\[
 \Pp\bigl(Q_i(t)\bigr)
 =\int_{(0,t]}\Pp\bigl(U_i\le y_i(s)\bigr)\dd G_i(s)
 =\Gamma\int_{(0,t]}\frac{S_i(s)}{S(s)}\dd B_i(s)
 =1-S_i(t),
\]
where the last equality holds because $S_i$ is absolutely continuous with $S_i'=-\Gamma(S_i/S)B_i'$ almost everywhere.

\emph{Second claim.}
The algorithm accepts the first qualifying arrival.
Hence no item is accepted by time $t$ if and only if none of the events $Q_j(t)$ occurs, and by independence and~\eqref{eq:survival-product} this has probability $\prod_jS_j(t)=S(t)$.

\emph{Third claim.}
Outside the null event of ties, item $i$ is accepted at a time in $(0,t]$ if and only if $Q_i(t)$ occurs and no other item has a qualifying arrival before $T_i$.
Conditioning on $(T_i,V_i)$ and using the continuity of the $S_j$,
\[
 \Pp(I=i,\ T_i\le t)
 =\int_{(0,t]}\Pp\bigl(U_i\le y_i(s)\bigr)\frac{S(s)}{S_i(s)}\dd G_i(s)
 =\Gamma\int_{(0,t]}\Pp\bigl(U_i\le y_i(s)\bigr)\dd y_i(s)
 =\Gamma B_i(t),
\]
by~\eqref{eq:arrivals} and Lemma~\ref{lem:threshold-paths}(d).

\emph{Fourth claim.}
For a value $v$, let $x(v):=1-\prod_{j\ne i}F_j(v)$ be the loss score of item $i$ when $V_i=v$.
Since $V_i$ is independent of $T_i$ and of $(T_j,V_j)_{j\ne i}$, Fubini's theorem shows that $\Pp(I=i,\ V_i\in E)=\int_E\psi(v)\dd F_i(v)$ for every Borel set $E$, where $\psi(v)$ is the acceptance probability computed as in the third claim with $U_i$ replaced by the constant $x(v)$.
Thus $\psi$ is a version of $\Pp(I=i\mid V_i=v)$, and
\[
\begin{aligned}
 \psi(v)
 &=\int_{(0,1]}\ind\bigl\{x(v)\le y_i(s)\bigr\}\frac{S(s)}{S_i(s)}\dd G_i(s)\\
 &=\Gamma\int_{(0,1]}\ind\bigl\{x(v)\le y_i(s)\bigr\}\dd y_i(s)
 =\Gamma\bigl(1-\max\{x(v),y_i(0)\}\bigr),
\end{aligned}
\]
where the last equality uses~\eqref{eq:stieltjes-substitution} with $g=\ind_{[x(v),1]}$ and $y_i(1)=1$.
By Lemma~\ref{lem:threshold-paths}(b), $x(V_i)=U_i\ge y_i(0)$ almost surely, so $x(v)\ge y_i(0)$ for $F_i$-almost every $v$.
For such $v$, $\psi(v)=\Gamma\prod_{j\ne i}F_j(v)=\Gamma\Pp(I^*=i\mid V_i=v)$, where the last identity holds because the values are independent and atomless.
This is~\eqref{eq:itemwise}.
\end{proof}

\section{Profiles and Area Curves}\label{app:profile}

This appendix proves Lemma~\ref{lem:profile-bijection}.
The main tool is a representation of each item's loss score through the rank of its own value (Lemma~\ref{lem:rank-representation}).
From it we derive the area formulas, recover the profile from the area curve, and show that every finite family of profiles summing to one is realized by some input.

Throughout, values are independent and atomless, $F_{\max}=\prod_jF_j$, and $F_{-i}:=\prod_{j\ne i}F_j$.
For $\lambda\in\Lambda$ we use $h_\lambda$, $C_\lambda$, $x_\lambda$, and $\mu_\lambda$ from~\eqref{eq:profile-coordinates-main}, and we set $x_\lambda(1):=1$.
Since $0\le\lambda\le1$, we have $0\le h_\lambda(u)\le\log\frac1{1-u}$ and hence $0\le x_\lambda(u)\le u$ for $u<1$.
Moreover, $x_\lambda$ is absolutely continuous on each interval $[0,u_1]$ with $u_1<1$, because $h_\lambda$ is, and
\[
 x_\lambda'(u)=e^{h_\lambda(u)}\bigl(1-\lambda(u)\bigr)\ge0
 \qquad\text{for almost every }u\in[0,1).
\]
Hence $x_\lambda$ is continuous on $[0,1)$ and nondecreasing on $[0,1]$.
For every $\lambda\in\Lambda$, define $A_\lambda:[0,1]\to\R$ by formula~\eqref{eq:rank-area-main} of Section~\ref{sec:small-profiles}:
\begin{equation*}
 A_\lambda(x):=\int_0^1\bigl(x-x_\lambda(s)\bigr)_+\,\frac{\lambda(s)}{1-s}\,e^{-h_\lambda(s)}\dd s.
 \tag{\ref*{eq:rank-area-main}}
\end{equation*}
We show below that $A_\lambda$ is the area curve of every item with profile $\lambda$ and that it satisfies all assertions of Lemma~\ref{lem:profile-bijection}.

It is convenient to read~\eqref{eq:rank-area-main} probabilistically.
Let $R$ be a random variable on $[0,1]$ with $\Pp(R>s)=e^{-h_\lambda(s)}$ for $s\in[0,1)$.
Since $h_\lambda$ is absolutely continuous on each interval $[0,s_1]$ with $s_1<1$, the law of $R$ has density $h_\lambda'(s)e^{-h_\lambda(s)}=\lambda(s)(1-s)^{-1}e^{-h_\lambda(s)}$ on $[0,1)$, together with an atom of mass $\lim_{s\uparrow1}e^{-h_\lambda(s)}$ at $s=1$.
Since $(x-x_\lambda(1))_+=0$ for $x\le1$, the atom does not contribute, and
\begin{equation}\label{eq:rank-area-expectation}
 A_\lambda(x)=\E\bigl[(x-X_\lambda)_+\bigr],
 \qquad
 X_\lambda:=x_\lambda(R).
\end{equation}
In particular, $A_\lambda$ is continuous, convex, nondecreasing, and $1$-Lipschitz, with $A_\lambda(0)=0$, and its right derivative is $\Pp(X_\lambda\le x)$ for $x\in[0,1)$.

\subsection{The Rank Representation}

For an item $i$ of an input, define the \emph{rank} of its value by
\[
 R_i:=1-F_{\max}(V_i)\in[0,1].
\]
This is the maximum rank that would result if $V_i$ were the maximum.
Since $F_{\max}$ is nondecreasing, $U=1-F_{\max}(\OPT)=\min_jR_j=R_{I^*}$.

\begin{lemma}[Rank representation]\label{lem:rank-representation}
Fix an input and an item $i$ with profile $\lambda=\lambda_i$.
Then:
\begin{enumerate}
\item[(a)] $\Pp(R_i>u)=e^{-h_\lambda(u)}$ for every $u\in[0,1)$;
\item[(b)] $U_i=x_\lambda(R_i)$ almost surely.
\end{enumerate}
Consequently, the area curve of item $i$ is $A_i=A_\lambda$, where $A_\lambda$ is given by~\eqref{eq:rank-area-main}.
In particular, $A_i$ depends on the input only through $\lambda_i$.
\end{lemma}

\begin{proof}
\emph{Common levels.}
For $u\in(0,1)$, the level set $\{v:F_{\max}(v)=1-u\}$ is a nonempty compact interval, because $F_{\max}$ is continuous and nondecreasing, vanishes on $(-\infty,0)$, and tends to one at infinity.
Let $v_u$ denote its left endpoint, and put $p_j(u):=F_j(v_u)$ for every item $j$.
Each $F_j$ is constant on the level set: the product $F_{\max}$ is constant and positive there, and each factor is nondecreasing.
Hence
\begin{equation}\label{eq:common-levels}
 F_{\max}(v)=1-u
 \quad\Longrightarrow\quad
 F_j(v)=p_j(u)\ \text{for every }j,
 \qquad\text{and}\qquad
 \prod_jp_j(u)=1-u.
\end{equation}
As $u$ increases, $v_u$ decreases, so each $\log p_j$ is nonincreasing.
For $0<a<b<1$, the nonnegative increments $\log p_j(a)-\log p_j(b)$ sum over $j$ to $\log\frac{1-a}{1-b}$.
Hence each $\log p_j$ is Lipschitz on every interval $(0,u_1]$ with $u_1<1$.
Moreover, $p_j(u)\to1$ as $u\downarrow0$, because $\prod_jp_j(u)=1-u\to1$ and each $p_j\le1$.
Set $p_j(0):=1$; then each $\log p_j$ is absolutely continuous on $[0,u_1]$ for every $u_1<1$.

\emph{The law of $R_i$.}
For $u\in(0,1)$, continuity and monotonicity of $F_{\max}$ give $\{R_i>u\}=\{F_{\max}(V_i)<1-u\}=\{V_i<v_u\}$.
Since $V_i$ is atomless, $\Pp(R_i>u)=F_i(v_u)=p_i(u)$.
For $u=0$, $\Pp(R_i=0)\le\Pp(F_i(V_i)=1)=0$, because $F_i(V_i)$ is uniform on $[0,1]$.
Thus $\Pp(R_i>u)=p_i(u)$ for all $u\in[0,1)$.

\emph{Winning probability through ranks.}
Given $V_i=v$, item $i$ is the maximum with probability $F_{-i}(v)$, because the values are independent and atomless.
On $\{0<R_i<1\}$, \eqref{eq:common-levels} gives $F_i(V_i)=p_i(R_i)>0$, so $F_{-i}(V_i)=F_{\max}(V_i)/F_i(V_i)=(1-R_i)/p_i(R_i)$.
On $\{R_i=1\}$, we have $F_{\max}(V_i)=0$ while $F_i(V_i)>0$ almost surely, so $F_{-i}(V_i)=0$ almost surely.
Since $U=R_i$ on $\{I^*=i\}$, and since $R_i$ has density $-p_i'$ on $[0,1)$, we obtain for every $u\in[0,1)$
\[
 \int_u^1\lambda(s)\dd s
 =\Pp(I^*=i,\ U>u)
 =\E\bigl[F_{-i}(V_i)\ind\{R_i>u\}\bigr]
 =\int_u^1\frac{1-s}{p_i(s)}\bigl(-p_i'(s)\bigr)\dd s.
\]
Hence $\lambda(s)=(1-s)\bigl(-\log p_i\bigr)'(s)$ for almost every $s$, and integrating from zero gives $-\log p_i(u)=h_\lambda(u)$.
This proves~(a).

\emph{Loss scores.}
On $\{0<R_i<1\}$, part (a) gives $U_i=1-F_{-i}(V_i)=1-(1-R_i)e^{h_\lambda(R_i)}=x_\lambda(R_i)$.
On $\{R_i=1\}$, $U_i=1=x_\lambda(1)$ almost surely.
This proves~(b).
Finally, (a) shows that $R_i$ has the law of $R$ in~\eqref{eq:rank-area-expectation}, so (b) gives $A_i(x)=\E[(x-U_i)_+]=\E[(x-x_\lambda(R_i))_+]=A_\lambda(x)$.
\end{proof}

\begin{remark}[Competing risks]
On the exponential scale, Lemma~\ref{lem:rank-representation}(a) reads
\[
 \Pp\bigl(-\log(1-R_i)>r\bigr)=\exp\!\left(-\int_0^r\lambda_i\bigl(1-e^{-s}\bigr)\dd s\right),
 \qquad r\ge0.
\]
Thus $\lambda_i(1-e^{-r})$ is the hazard rate at level $r$ of the exponential rank $-\log(1-R_i)=-\log F_{\max}(V_i)$.
These exponential ranks are independent and possibly defective (they equal $+\infty$ on $\{R_i=1\}$), their minimum $-\log(1-U)$ is standard exponential, and the identity $\sum_i\lambda_i=1$ expresses the familiar fact that the hazard rates of independent competing risks add up.
This viewpoint underlies the exponential coordinates of Section~\ref{sec:kzero-analysis} and the construction in Lemma~\ref{lem:profile-realization}.
\end{remark}

\subsection{From a Profile to Its Area Curve}

\begin{lemma}[Area formulas]\label{lem:profile-area}
Let $\lambda\in\Lambda$ and let $A_\lambda$ be given by~\eqref{eq:rank-area-main}.
\begin{enumerate}
\item[(a)] $A_\lambda(1)=\mu_\lambda$.
\item[(b)] For every $u\in[0,1)$,
\begin{equation*}
 A_\lambda\bigl(x_\lambda(u)\bigr)
 =C_\lambda(u)-(1-u)\bigl(e^{h_\lambda(u)}-1\bigr)
 =C_\lambda(u)+x_\lambda(u)-u.
 \tag{\ref*{eq:profile-area-main}}
\end{equation*}
\item[(c)] Let $x_*:=\lim_{u\uparrow1}x_\lambda(u)$.
If $x_*<1$, then
\begin{equation*}
 A_\lambda(x)=\mu_\lambda+x-1,
 \qquad x\in[x_*,1].
 \tag{\ref*{eq:profile-extension-main}}
\end{equation*}
\end{enumerate}
Since $x_\lambda$ is continuous and nondecreasing on $[0,1)$ with $x_\lambda(0)=0$, its range contains $[0,x_*)$; hence these identities and the continuity of $A_\lambda$ determine $A_\lambda$ on $[0,1]$.
\end{lemma}

\begin{proof}
Abbreviate $h=h_\lambda$, and recall that $h'(s)=\lambda(s)/(1-s)$ almost everywhere.

(a)
By~\eqref{eq:rank-area-main}, $A_\lambda(1)=\int_0^1(1-s)e^{h(s)}h'(s)e^{-h(s)}\dd s=\int_0^1\lambda(s)\dd s=\mu_\lambda$.

(b)
Fix $u<1$.
Since $x_\lambda$ is nondecreasing on $[0,1]$, the integrand in~\eqref{eq:rank-area-main} at $x=x_\lambda(u)$ vanishes for $s\ge u$.
For $s<u$, $x_\lambda(u)-x_\lambda(s)=(1-s)e^{h(s)}-(1-u)e^{h(u)}$.
Therefore
\[
 A_\lambda\bigl(x_\lambda(u)\bigr)
 =\int_0^u\lambda(s)\dd s-(1-u)e^{h(u)}\int_0^uh'(s)e^{-h(s)}\dd s
 =C_\lambda(u)-(1-u)e^{h(u)}\bigl(1-e^{-h(u)}\bigr),
\]
which is the first expression in~\eqref{eq:profile-area-main}.
The second follows from $x_\lambda(u)-u=-(1-u)(e^{h(u)}-1)$.

(c)
If $x_*<1$, then $(1-u)e^{h(u)}=1-x_\lambda(u)\to1-x_*>0$ as $u\uparrow1$, so $h(u)\to\infty$.
Hence $\Pp(R=1)=0$ in~\eqref{eq:rank-area-expectation}, and $X_\lambda\le x_*$ almost surely.
For $x\ge x_*$, it follows that $A_\lambda(x)=\E[x-X_\lambda]=x-1+A_\lambda(1)=x-1+\mu_\lambda$, by (a).
\end{proof}

\subsection{Recovering the Profile from the Area Curve}

For an area curve $A=A_\lambda$, let $A'_+$ denote its right derivative, and define
\[
\begin{aligned}
 \rho_A(x)&:=1-(1-x)\bigl(1-A'_+(x)\bigr)\quad(x\in[0,1)),
 &\qquad \rho_A(1)&:=1,\\
 x_A(u)&:=\inf\{x\in[0,1]:\rho_A(x)>u\}\quad(u\in[0,1)).
\end{aligned}
\]
The last set is nonempty because $\rho_A(1)=1$.
The function $\rho_A$ converts the loss-score coordinate into the rank coordinate: as the proof shows, $\rho_A(x)$ is the supremum of the ranks $s<1$ whose loss score $x_\lambda(s)$ does not exceed $x$.

\begin{lemma}[Recovery of the profile]\label{lem:profile-recovery}
For every $\lambda\in\Lambda$ and $A=A_\lambda$, we have $x_A(u)=x_\lambda(u)$ for $u\in[0,1)$, and
\begin{equation}\label{eq:profile-recovery}
 \lambda(u)=\frac{\dd}{\dd u}\Bigl[u-x_A(u)+A\bigl(x_A(u)\bigr)\Bigr]
 \qquad\text{for almost every }u\in[0,1].
\end{equation}
In particular, $A_\lambda$ determines $\lambda$ up to a null set.
\end{lemma}

\begin{proof}
For $x\in[0,1)$, let $\bar s(x):=\sup\{s\in[0,1):x_\lambda(s)\le x\}$; the set contains $s=0$.
By~\eqref{eq:rank-area-expectation}, $A'_+(x)=\Pp(X_\lambda\le x)$.

Suppose first that $\bar s(x)<1$.
Since $x_\lambda$ is continuous and nondecreasing, $\{s\in[0,1):x_\lambda(s)\le x\}=[0,\bar s(x)]$ and $x_\lambda(\bar s(x))=x$.
Moreover, $X_\lambda=x_\lambda(R)\le x$ if and only if $R\le\bar s(x)$, because $x_\lambda(1)=1>x$.
Hence $A'_+(x)=\Pp(R\le\bar s(x))=1-e^{-h_\lambda(\bar s(x))}$, and
\[
 \rho_A(x)=1-\bigl(1-x_\lambda(\bar s(x))\bigr)e^{-h_\lambda(\bar s(x))}=\bar s(x).
\]
If instead $\bar s(x)=1$, then $x_\lambda(s)\le x<1$ for all $s<1$, so $x_*\le x<1$.
As in the proof of Lemma~\ref{lem:profile-area}(c), $\Pp(R=1)=0$, so $A'_+(x)=1$ and $\rho_A(x)=1=\bar s(x)$.
Thus $\rho_A=\bar s$ on $[0,1)$.

Now fix $u\in[0,1)$.
By the definition of $\bar s$, the inequality $\rho_A(x)>u$ holds if and only if $x_\lambda(s)\le x$ for some $s\in(u,1)$, or $x=1$.
Since $x_\lambda$ is continuous and nondecreasing, $\inf_{s\in(u,1)}x_\lambda(s)=x_\lambda(u)$.
Hence $\{x:\rho_A(x)>u\}$ contains $(x_\lambda(u),1]$ and is contained in $[x_\lambda(u),1]$, so $x_A(u)=x_\lambda(u)$.
By Lemma~\ref{lem:profile-area}(b), $u-x_A(u)+A(x_A(u))=C_\lambda(u)$, whose derivative is $\lambda(u)$ for almost every $u$.
\end{proof}

\subsection{Realizing Families of Profiles}\label{app:realization}

\begin{lemma}[Realizability]\label{lem:profile-realization}
Let $\lambda_1,\ldots,\lambda_n:[0,1]\to[0,1]$ be measurable with $\sum_i\lambda_i=1$ almost everywhere.
Then there are independent, atomless, nonnegative random values $V_1,\ldots,V_n$ whose profiles are $\lambda_1,\ldots,\lambda_n$, up to null sets.
In particular, every profile $\lambda\in\Lambda$ is realized in a two-item input with profiles $(\lambda,1-\lambda)$.
\end{lemma}

\begin{proof}
Write $h_i:=h_{\lambda_i}$ and $c_i:=\lim_{u\uparrow1}e^{-h_i(u)}$.
By assumption,
\[
 \sum_ih_i(u)=\int_0^u\frac{\dd a}{1-a}=\log\frac1{1-u},
 \qquad\text{so}\qquad
 \prod_ie^{-h_i(u)}=1-u
 \quad\text{and}\quad
 \prod_ic_i=0.
\]
Let $V_1,\ldots,V_n$ be independent with distribution functions
\begin{equation}\label{eq:profile-realization}
 F_i(v):=
 \begin{cases}
  0,&v<0,\\
  c_iv,&0\le v\le1,\\
  e^{-h_i(2-v)},&1<v\le2,\\
  1,&v>2.
 \end{cases}
\end{equation}
Each $F_i$ is continuous and nondecreasing: $c_iv=c_i$ at $v=1$, $e^{-h_i(2-v)}\to c_i$ as $v\downarrow1$, and $e^{-h_i(2-v)}=1$ at $v=2$.
Hence each $V_i$ is atomless and nonnegative.
On $[1,2]$ we have $F_{\max}(v)=\prod_ie^{-h_i(2-v)}=v-1$, and on $[0,1]$ we have $F_{\max}(v)=v^n\prod_ic_i=0$.
Thus, in the notation of the proof of Lemma~\ref{lem:rank-representation}, $v_u=2-u$ for $u\in(0,1)$, and the paragraph ``The law of $R_i$'' in that proof, together with part~(a) of the lemma, gives
\[
 e^{-h_{\widetilde\lambda_i}(u)}=\Pp(R_i>u)=F_i(2-u)=e^{-h_i(u)},
 \qquad u\in[0,1),
\]
where $\widetilde\lambda_i$ denotes the actual profile of item $i$.
Hence $h_{\widetilde\lambda_i}=h_{\lambda_i}$ on $[0,1)$, and differentiation gives $\widetilde\lambda_i=\lambda_i$ almost everywhere.
\end{proof}

In this construction, $V_i=2-R_i$ whenever $R_i<1$: each value is placed at the location dictated by its rank, and the remaining probability $c_i$, on which item $i$ never wins, is spread over $[0,1]$ to keep the distribution atomless.

\subsection{Proof of Lemma~\ref{lem:profile-bijection}}

\begin{proof}[Proof of Lemma~\ref{lem:profile-bijection}]
For $\lambda\in\Lambda$, define $A_\lambda$ by~\eqref{eq:rank-area-main}.
By Lemma~\ref{lem:profile-realization}, some input contains an item with profile $\lambda$, and by Lemma~\ref{lem:rank-representation} its area curve is $A_\lambda$; hence $A_\lambda\in\mathcal A$.
Conversely, every $A\in\mathcal A$ is the area curve $A_i$ of some item $i$ of some input, and $A_i=A_{\lambda_i}$ by Lemma~\ref{lem:rank-representation}; hence $\lambda\mapsto A_\lambda$ maps $\Lambda$ onto $\mathcal A$.
The integral~\eqref{eq:rank-area-main} depends only on the equivalence class of $\lambda$, and Lemma~\ref{lem:profile-recovery} recovers this class from $A_\lambda$; hence the map is a bijection.
Lemma~\ref{lem:rank-representation} also shows that $\lambda_i$ corresponds to $A_i$ for every input and item.
Finally, Lemma~\ref{lem:profile-area} gives $A_\lambda(1)=\mu_\lambda$, the formulas~\eqref{eq:profile-area-main} and~\eqref{eq:profile-extension-main}, and the fact that they determine $A_\lambda$; the properties of $x_\lambda$ were noted at the beginning of this appendix.
\end{proof}

\section{Proof of the Characterization of Local Constructions}\label{app:converse}

Sufficiency in Theorem~\ref{thm:local} was proved in Section~\ref{sec:characterization}; here we prove necessity and uniqueness.
Let $\mathcal B$ be a local construction.
By Lemma~\ref{lem:profile-bijection}, we may write $B_f:=\mathcal B[A_f]$ for $f\in\Lambda$, and $B_f$ depends only on the equivalence class of $f$.
By Lemma~\ref{lem:profile-realization}, every finite family of profiles summing to one almost everywhere is the family of profiles of some input.
The total-allocation identity~\eqref{eq:total-allocation} therefore yields
\begin{equation}\label{eq:partition-identity}
 \sum_{i=1}^nB_{f_i}(t)=t
 \qquad\text{whenever }f_1,\ldots,f_n\in\Lambda\text{ and }\sum_if_i=1\text{ almost everywhere}.
\end{equation}
The proof has four steps.
For fixed $t$, the map $f\mapsto B_f(t)$ is additive and positive (Step~1), so it is integration against a density $k_t$ (Step~2).
These densities can be chosen consistently in $t$, which produces the kernel (Step~3).
Step~4 proves uniqueness.

\begin{proof}[Proof of Theorem~\ref{thm:local}]
\emph{Step 1: additivity.}
Fix $t\in[0,1]$ and write $\mathcal T(f):=B_f(t)$ for $f\in\Lambda$.
Since $\mathcal B[A]$ takes values in $[0,A(1)]$ and $A_f(1)=\int_0^1f$, we have
\begin{equation}\label{eq:local-bounds}
 0\le \mathcal T(f)\le\int_0^1f(u)\dd u,
 \qquad\text{in particular}\qquad
 \mathcal T(0)=0.
\end{equation}
Let $f,g\in\Lambda$ with $f+g\le1$.
Applying~\eqref{eq:partition-identity} to the families $(f,g,1-f-g)$ and $(f+g,0,1-f-g)$ gives
\[
 \mathcal T(f)+\mathcal T(g)+\mathcal T(1-f-g)=t=\mathcal T(f+g)+\mathcal T(0)+\mathcal T(1-f-g),
\]
so $\mathcal T(f+g)=\mathcal T(f)+\mathcal T(g)$.
Three consequences follow.
First, $\mathcal T$ is monotone: if $f\le g$, then $\mathcal T(g)=\mathcal T(f)+\mathcal T(g-f)\ge\mathcal T(f)$.
Second, $\mathcal T$ is positively homogeneous: additivity gives $\mathcal T(f)=m\,\mathcal T(f/m)$ for every integer $m\ge1$, hence $\mathcal T(cf)=c\,\mathcal T(f)$ for rational $c\in[0,1]$, and monotonicity extends this to real $c\in[0,1]$.
Third, $\mathcal T$ is Lipschitz in $L^1$: writing $f=f\wedge g+(f-g)_+$ and $g=f\wedge g+(g-f)_+$, additivity and~\eqref{eq:local-bounds} give
\begin{equation}\label{eq:local-lipschitz}
 |\mathcal T(f)-\mathcal T(g)|
 =\bigl|\mathcal T\bigl((f-g)_+\bigr)-\mathcal T\bigl((g-f)_+\bigr)\bigr|
 \le\int_0^1|f(u)-g(u)|\dd u.
\end{equation}

\emph{Step 2: a density at each time.}
For a Borel set $E\subseteq[0,1]$, put $\nu_t(E):=\mathcal T(\ind_E)$.
By additivity, $\nu_t$ is finitely additive, and $0\le\nu_t(E)\le|E|$ by~\eqref{eq:local-bounds}, where $|E|$ denotes Lebesgue measure.
If $E_k\downarrow\emptyset$, then $\nu_t(E_k)\le|E_k|\to0$; hence $\nu_t$ is countably additive and absolutely continuous with respect to Lebesgue measure.
By the Radon--Nikodym theorem, $\nu_t(E)=\int_Ek_t(u)\dd u$ for a Borel function $k_t$ with $0\le k_t\le1$ almost everywhere.
By additivity and homogeneity, $\mathcal T(f)=\int_0^1fk_t$ for every simple function $f=\sum_jc_j\ind_{E_j}$ with disjoint $E_j$ and $c_j\in[0,1]$.
Every $f\in\Lambda$ is a uniform limit of such functions, so~\eqref{eq:local-lipschitz} gives
\begin{equation}\label{eq:local-density}
 B_f(t)=\int_0^1f(u)k_t(u)\dd u
 \qquad(f\in\Lambda).
\end{equation}
Applying~\eqref{eq:partition-identity} to the family $(1,0)$ gives $\int_0^1k_t(u)\dd u=B_1(t)=t$.

\emph{Step 3: a consistent choice in $t$.}
If $s<t$, then $\int_Ek_s=B_{\ind_E}(s)\le B_{\ind_E}(t)=\int_Ek_t$ for every Borel set $E$, so $k_s\le k_t$ almost everywhere.
Moreover, $k_0=0$ and $k_1=1$ almost everywhere, because $\int_0^1k_0=0$, $\int_0^1k_1=1$, and $0\le k_t\le1$.
Fix versions of $k_s$ for rational $s\in[0,1]$, and choose a null set $N$ such that, for every $u\notin N$, the map $s\mapsto k_s(u)$ is $[0,1]$-valued and nondecreasing on $\mathbb Q\cap[0,1]$, with $k_0(u)=0$ and $k_1(u)=1$.
For $u\notin N$, define
\[
 K(t,u):=\inf\bigl\{k_s(u):s\in\mathbb Q\cap(t,1]\bigr\}\quad(t<1),
 \qquad
 K(1,u):=1,
\]
and set $K(t,u):=t$ for $u\in N$.
Then $K$ is jointly measurable, being a countable infimum of jointly measurable functions, and $K(\cdot,u)$ is nondecreasing and right-continuous for every $u$.

Each $B_f$ is continuous: applying~\eqref{eq:partition-identity} to $(f,1-f)$ gives, for $s<t$,
\[
 0\le B_f(t)-B_f(s)=(t-s)-\bigl(B_{1-f}(t)-B_{1-f}(s)\bigr)\le t-s.
\]
For $t<1$, let $s_m\downarrow t$ be rational.
For $u\notin N$, $k_{s_m}(u)\downarrow K(t,u)$, so~\eqref{eq:local-density} and monotone convergence give
\[
 B_f(t)=\lim_{m\to\infty}B_f(s_m)=\lim_{m\to\infty}\int_0^1f(u)k_{s_m}(u)\dd u=\int_0^1f(u)K(t,u)\dd u.
\]
For $t=1$, both sides equal $\int_0^1f$.
This is the representation~\eqref{eq:local-representation}.
Taking $f\equiv1$ gives $\int_0^1K(t,u)\dd u=t$ for every $t$, which is~\eqref{eq:kernel-marginal}.
In particular, $\int_0^1K(0,u)\dd u=0$, so $K(0,u)=0$ outside a null set $N'$.
Redefining $K(t,u):=t$ for $u\in N'$ changes no integral, and the resulting function is an allocation kernel.

\emph{Step 4: uniqueness.}
Suppose that $K$ and $\widetilde K$ both represent $\mathcal B$.
For rational $t$ and every Borel set $E$, $\int_EK(t,u)\dd u=B_{\ind_E}(t)=\int_E\widetilde K(t,u)\dd u$, so $K(t,\cdot)=\widetilde K(t,\cdot)$ almost everywhere.
Let $N''$ be the union of these null sets over rational $t$, together with the null sets off which $K(\cdot,u)$ and $\widetilde K(\cdot,u)$ are right-continuous.
For $u\notin N''$, the two right-continuous functions $K(\cdot,u)$ and $\widetilde K(\cdot,u)$ agree on the rationals and hence everywhere.
\end{proof}

\section{Constant Profiles and Load Monotonicity}\label{app:kernel-analysis}

\subsection{Constant Profiles and the I.I.D.\ Benchmark}
\label{app:iid-benchmark}

\begin{proof}[Proof of Lemma~\ref{lem:iid-kernel-blind}]
For $\lambda\equiv c$, the profile--area formulas in
Lemma~\ref{lem:profile-bijection} give
\[
 h_c(u)=-c\log(1-u),\qquad C_c(u)=cu,\qquad
 x_c(u)=1-(1-u)^{1-c}.
\]
Substituting $u=1-(1-x)^{1/(1-c)}$ into
$A_c(x_c(u))=C_c(u)+x_c(u)-u$ proves
\eqref{eq:constant-profile-area}.
Since $A_c(0)=0$, $A_c(1)=c$, and
$A_c'(x)=1-(1-x)^{c/(1-c)}>0$ for $x\in(0,1)$,
$\bar A_c=A_c/c$ is a strictly increasing bijection of $[0,1]$.
The allocation $B_{K,c}(t)=ct$ therefore gives
$y_{K,c}=\bar A_c^{-1}$.
Also,~\eqref{eq:scalar-paths} yields
$S_{K,c}(t)=(1-\Gamma t)^c$.
Substituting $t=\bar A_c(x)$ into~\eqref{eq:scalar-load} gives
\[
 \ell_\Gamma(K,c)
 =\Gamma\int_{(0,1]}(1-\Gamma t)^{c-1}\dd y_{K,c}(t)
 =\Gamma\int_0^1
   \frac{\dd x}{(1-\Gamma\bar A_c(x))^{1-c}}.
\]

For fixed $x<1$, the formula for $A_c(x)$ extends differentiably to
$c=0$, with $A_0(x)=0$ and
\[
 \left.\frac{\partial}{\partial c}A_c(x)\right|_{c=0}
 =x+(1-x)\log(1-x)=\phi(x).
\]
Thus $\bar A_c(x)\to\phi(x)$ as $c\downarrow0$.
At $x=1$, both functions equal one.
Since the functions $\bar A_c$ are nondecreasing and the limit $\phi$
is continuous, convergence is uniform on $[0,1]$.
Strict monotonicity of $\phi$ then gives
$y_{K,c}(t)\to q(t)$ for every $t\in[0,1]$.
The integrand $(1-\Gamma\bar A_c(x))^{-(1-c)}$ is bounded by
$(1-\Gamma)^{-1}$, so dominated convergence gives
$\ell_\Gamma(K,c)\to J(\Gamma)$.

To compare each finite profile with the limit, we show that
\begin{equation}\label{eq:const-profile-comparison}
 (1-\bar A_c(x))^{1-c}\ge1-\phi(x),
 \qquad x\in[0,1].
\end{equation}
Both sides vanish at $x=1$.
For $x<1$, write $1-x=e^{-r}$ with $r\ge0$.
Then $1-\phi(x)=(1+r)e^{-r}$ and
\[
 1-\bar A_c(x)
 =\frac{e^{-r}-(1-c)e^{-r/(1-c)}}{c}.
\]
Consequently,~\eqref{eq:const-profile-comparison} is equivalent to
\[
 e^{cr/(1-c)}\ge(1-c)+c(1+r)^{1/(1-c)}.
\]
Both sides equal one at $r=0$.
The derivative inequality follows from $e^r\ge1+r$:
\[
 \frac{c}{1-c}e^{cr/(1-c)}
 \ge\frac{c}{1-c}(1+r)^{c/(1-c)}.
\]
This proves~\eqref{eq:const-profile-comparison}.
Since $v\mapsto v^{1-c}$ is concave and $0\le\bar A_c\le1$,
\[
\begin{aligned}
 (1-\Gamma\bar A_c(x))^{1-c}
 &=\bigl((1-\Gamma)+\Gamma(1-\bar A_c(x))\bigr)^{1-c}\\
 &\ge1-\Gamma+\Gamma(1-\bar A_c(x))^{1-c}\\
 &\ge1-\Gamma\phi(x).
\end{aligned}
\]
Substituting this into the integral formula for the load proves
$\ell_\Gamma(K,c)\le J(\Gamma)$.
\end{proof}

For completeness, $J$ is continuous and strictly increasing because
its integrand $\Gamma/(1-\Gamma\phi(x))$ has these properties and is
locally uniformly bounded for $\Gamma\in(0,1)$.
Also, $J(\Gamma)\to0$ as $\Gamma\downarrow0$, whereas monotone
convergence and the substitution $y=1-x$ give
\[
 \lim_{\Gamma\uparrow1}J(\Gamma)
 =\int_0^1\frac{\dd x}{1-\phi(x)}
 =\int_0^1\frac{\dd y}{y(1-\log y)}
 =\infty.
\]
Thus~\eqref{eq:iid-constant} has a unique solution.
Writing $J(\Gamma)=\int_0^1(1/\Gamma-\phi(x))^{-1}\dd x$ and
again substituting $y=1-x$ yields the classical characterization
\[
 \int_0^1
 \frac{\dd y}{y(1-\log y)+1/\Gamma_\star-1}=1
\]
of the optimal i.i.d.\ ratio~\citep{Kertz86,CFHOV21}.

\subsection{Monotonicity in the Target Ratio}\label{app:gamma-monotone}

\begin{proof}[Proof of Lemma~\ref{lem:gamma-monotone}]
For fixed $K$ and $\lambda$, the allocation $B=B_{K,\lambda}$ and the threshold path $y=y_{K,\lambda}$ do not depend on $\Gamma$.
Since $-\log(1-\Gamma t)=\int_0^t\Gamma(1-\Gamma s)^{-1}\dd s$, definition~\eqref{eq:scalar-paths} gives
\[
 \frac{S_{K,\lambda}(t)}{1-\Gamma t}
 =\exp\!\left(\int_0^t\frac{\Gamma\,(1-B'(s))}{1-\Gamma s}\dd s\right).
\]
Because $B$ is $1$-Lipschitz, $1-B'\ge0$ almost everywhere, and $\Gamma\mapsto\Gamma/(1-\Gamma s)$ is increasing.
Hence the ratio above is nondecreasing in $\Gamma$ for every $t$, and so is $\ell_\Gamma(K,\lambda)=\Gamma\int_{(0,1]}S_{K,\lambda}(t)(1-\Gamma t)^{-1}\dd y(t)$.
The second claim follows from Proposition~\ref{prop:scalar-feasibility}.
\end{proof}

\section{Deferred Proofs from Section~\ref{sec:kzero-analysis}}\label{app:kzero-proofs}

Throughout this appendix, $\Gamma\in(0,1)$ is fixed unless stated otherwise, and we abbreviate $D=D_\Gamma$, $H=H_\Gamma$, and $p=p_\Gamma$.
We use the formula~\eqref{eq:kzero-hamiltonian} to define $H(r,v)$ for every $v\in\R$.
With this convention, $H$ is continuous, and continuously differentiable in $v$ with
\begin{equation}\label{eq:hamiltonian-derivative}
 \partial_vH(r,v)=-\frac{\pos{r-1-D(r)v}}{2D(r)}\le0.
\end{equation}
Moreover, $r/D(r)^2\le H(r,v)\le M(r)/D(r)^2$ for $v\ge0$, by~\eqref{eq:kzero-hamiltonian} and the bound $p\le M/D^2$.

\begin{proof}[Proof of Lemma~\ref{lem:kzero-bellman}]
\emph{Finite horizon.}
For $\bar r>0$, consider the terminal value problem
\begin{equation}\label{eq:finite-horizon-bellman}
 -V'(r)=H\bigl(r,V(r)\bigr)\quad(0\le r\le\bar r),
 \qquad V(\bar r)=0.
\end{equation}
By~\eqref{eq:hamiltonian-derivative}, $H$ is locally Lipschitz in $v$, so~\eqref{eq:finite-horizon-bellman} has a unique maximal solution, integrated backward from $\bar r$.
Since $H\ge0$ everywhere, this solution is nonincreasing and hence nonnegative; then $0\le-V'\le M/D^2$ prevents blowup.
Hence~\eqref{eq:finite-horizon-bellman} has a unique solution $V^{\bar r}$ on $[0,\bar r]$, and $V^{\bar r}\ge0$.

Let $a\in[0,\bar r]$ and let $k$ be an admissible control.
The function $r\mapsto\sigma_{a,k}(r)V^{\bar r}(r)$ is absolutely continuous on $[a,\bar r]$, and for almost every $r$,
\[
 \frac{\dd}{\dd r}\bigl(\sigma_{a,k}(r)V^{\bar r}(r)\bigr)
 =-\sigma_{a,k}(r)\left(\frac{k(r)}{D(r)}V^{\bar r}(r)+H\bigl(r,V^{\bar r}(r)\bigr)\right)
 \le-\sigma_{a,k}(r)\,p\bigl(r,k(r)\bigr),
\]
where the inequality is the definition of $H$ as a maximum.
Integrating over $[a,\bar r]$ and using $V^{\bar r}(\bar r)=0$ and $\sigma_{a,k}(a)=1$ gives
\[
 \int_a^{\bar r}\sigma_{a,k}(r)\,p\bigl(r,k(r)\bigr)\dd r\le V^{\bar r}(a).
\]
Equality holds for the continuous feedback control $k(r)=k_*(r,V^{\bar r}(r))$.
Hence $V^{\bar r}(a)$ is the value of the control problem truncated at horizon $\bar r$.

\emph{Infinite horizon.}
Since $p\ge0$ and $\sigma_{a,k}\le1$, the truncated values are nondecreasing in $\bar r$, and for $a\le\bar r<\bar r'$,
\[
 0\le V^{\bar r'}(a)-V^{\bar r}(a)\le\int_{\bar r}^\infty\frac{M(r)}{D(r)^2}\dd r.
\]
The right side tends to zero as $\bar r\to\infty$, uniformly in $a$, so $V^{\bar r}$ converges locally uniformly to a limit $V_\infty$.
Exchanging the two suprema gives $V_\infty(a)=\sup_{\bar r}\sup_k\int_a^{\bar r}\sigma_{a,k}p=\sup_k\mathcal P_a(k)$, which is~\eqref{eq:kzero-control-value}.
Passing to the limit in $V^{\bar r}(a)=\int_a^{\bar r}H(r,V^{\bar r}(r))\dd r$, by dominated convergence with the integrable majorant $M/D^2$, gives $V_\infty(a)=\int_a^\infty H(r,V_\infty(r))\dd r$.
The bounds~\eqref{eq:kzero-tail-majorant} for $V_\infty$ follow from $0\le H\le M/D^2$ and $D(r)\ge c_\Gamma e^r$; in particular, $V_\infty(r)\to0$ as $r\to\infty$.
Hence $V_\infty$ is a nonnegative $C^1$ solution of~\eqref{eq:kzero-bellman}.

\emph{Uniqueness.}
Let $V$ and $\widetilde V$ be nonnegative $C^1$ solutions of~\eqref{eq:kzero-bellman}, and put $\Delta:=V-\widetilde V$.
By the mean value theorem, $\Delta'=\beta\Delta$ with
\[
 \beta(r):=-\int_0^1\partial_vH\bigl(r,\widetilde V(r)+\theta\Delta(r)\bigr)\dd\theta\ge0,
\]
which is continuous by~\eqref{eq:hamiltonian-derivative}.
Hence, for $r<r'$, $|\Delta(r)|=|\Delta(r')|\exp(-\int_r^{r'}\beta)\le|\Delta(r')|$.
Letting $r'\to\infty$ gives $\Delta(r)=0$.
\end{proof}

The proof also justifies the numerical procedure of Section~\ref{sec:kzero-analysis}: the truncated value $V^{\bar r}(0)$ is computed by solving~\eqref{eq:finite-horizon-bellman}, and $0\le V_\Gamma(0)-V^{\bar r}(0)\le\int_{\bar r}^\infty M/D^2$, which yields~\eqref{eq:kzero-truncation}.

\begin{proof}[Proof of Lemma~\ref{lem:kzero-root}]
\emph{Monotonicity and continuity.}
Fix an admissible control $k$, which does not depend on $\Gamma$, and abbreviate $\sigma=\sigma_{0,k}$.
The payoff density is $\sigma(r)(r-k)(1+k)/D_\Gamma(r)^2$, and only the factor $\sigma(r)/D_\Gamma(r)^2$ depends on $\Gamma$.
Since $\partial_\Gamma c_\Gamma=-\Gamma^{-2}$, we have $\partial_\Gamma D_\Gamma(r)=-e^r/\Gamma^2$ and hence
\begin{equation}\label{eq:kzero-parameter-derivative}
 \partial_\Gamma\log\frac{\sigma(r)}{D_\Gamma(r)^2}
 =\frac1{\Gamma^2}
 \left[\frac{2e^r}{D_\Gamma(r)}-\int_0^r\frac{k(s)e^s}{D_\Gamma(s)^2}\dd s\right].
\end{equation}
The identities $\bigl(e^r/D_\Gamma(r)\bigr)'=re^r/D_\Gamma(r)^2$ and $1/D_\Gamma(0)=\Gamma$, together with $k(s)\le s$, show that the bracket is at least
\[
 \frac{2e^r}{D_\Gamma(r)}-\int_0^r\frac{se^s}{D_\Gamma(s)^2}\dd s
 =\frac{e^r}{D_\Gamma(r)}+\Gamma>0.
\]
The bracket is at most $2e^r/D_\Gamma(r)\le2/c_\Gamma$.
Hence
\[
 0<\partial_\Gamma\log\frac{\sigma(r)}{D_\Gamma(r)^2}\le\frac{2}{\Gamma(1-\Gamma)},
\]
so the payoff of every fixed control is nondecreasing in $\Gamma$.
On a compact interval $[\underline\Gamma,\overline\Gamma]\subset(0,1)$, all payoff densities are bounded by $c_{\overline\Gamma}^{-2}M(r)e^{-2r}$, and their derivatives in $\Gamma$ by $2\underline\Gamma^{-1}(1-\overline\Gamma)^{-1}c_{\overline\Gamma}^{-2}M(r)e^{-2r}$.
Thus all payoffs are Lipschitz on $[\underline\Gamma,\overline\Gamma]$ with a common constant.
By~\eqref{eq:kzero-control-value}, $V_\Gamma(0)$ is the supremum of these payoffs, so it is nondecreasing and locally Lipschitz in $\Gamma$.
Consequently, $\Psi(\Gamma)=\Gamma+V_\Gamma(0)$ is continuous and strictly increasing.

\emph{Limits.}
As $\Gamma\downarrow0$, \eqref{eq:kzero-tail-majorant} gives $V_\Gamma(0)\le c_\Gamma^{-2}\int_0^\infty M(r)e^{-2r}\dd r=O(\Gamma^2)$, so $\Psi(\Gamma)\to0$.
As $\Gamma\uparrow1$, the zero control gives
\[
 V_\Gamma(0)\ge\int_0^\infty\frac{r}{D_\Gamma(r)^2}\dd r
 \longrightarrow\int_0^\infty\frac{r}{(1+r)^2}\dd r=\infty
\]
by monotone convergence, since $D_\Gamma(r)\downarrow1+r$.
\end{proof}

\section{Proof of the Guarantee for \texorpdfstring{$K_1$}{K1}}
\label{app:cert-proof}

We prove Theorem~\ref{thm:certified-kernel} for the fixed pair
$(K_1,\Gamma_1)$ from Section~\ref{sec:final-kernel}.
It is enough to verify that $K_1$ is an allocation kernel and that
\begin{equation}\label{eq:cert-acceptance-bound}
 \ell_{\Gamma_1}(K_1,\lambda)\le U_{\mathrm{cert}}
 :=\frac{999999999886}{10^{12}}<1
 \qquad(\lambda\in\Lambda).
\end{equation}
Proposition~\ref{prop:scalar-feasibility} then gives the algorithmic guarantee.

After checking the kernel in Appendix~\ref{app:cert-validity}, we write the load in the
exponential coordinates used for $K_0$ (Appendix~\ref{app:cert-load}). The perturbation
changes the allocation of a profile only through four block masses, while the profile itself
determines a nondecreasing, $1$-Lipschitz path $h$. For each region of block masses, we bound the threshold
complement in terms of the current state $h(r)$ (Appendices~\ref{sec:cert-finite}
and~\ref{sec:cert-comparisons}) and then bound the load along every admissible path
(Appendix~\ref{app:cert-regional-comparison}). A finite cover of the parameter space turns
these regional bounds into~\eqref{eq:cert-acceptance-bound} (Appendix~\ref{sec:cert-results}).
The finite computation does not restrict the profiles to a finite family, nor does it require
them to be constant on the blocks. Figure~\ref{fig:certificate} shows the structure of the
argument, and Table~\ref{tab:cert-notation} collects the notation. The detailed inverse
estimates, quadrature bounds, and regional comparisons are given in
Appendix~\ref{app:cert-details}.

\begin{figure}[htb]
\centering
\begin{tikzpicture}[font=\footnotesize]
  \node[align=flush center] at (34mm,5mm) {\small\textbf{Cover of the mass vectors $\mathbf x\in[0,1]^4$}};
  \node[align=flush center] at (121mm,5mm) {\small\textbf{Bound for one region $\mathcal R$}};
  \foreach \i/\y in {0/-5, 1/-16, 2/-27, 3/-38} { \coordinate (C\i) at (34mm,\y mm); }
  \node[flowbox, text width=64mm, minimum height=8mm] (c0) at (C0)
    {$\mathbf x=\mathbf 0$: analytic bound~\eqref{eq:cert-zero-scale}};
  \node[flowbox, text width=64mm, minimum height=8mm] (c1) at (C1)
    {$0<x_{\max}\le2^{-14}$: boundary certificate ($32{,}529$ regions)};
  \node[flowbox, text width=64mm, minimum height=8mm] (c2) at (C2)
    {$2^{-14}\le x_{\max}\le2^{-10}$: $16$ scale charts ($377{,}458$ regions)};
  \node[flowbox, text width=64mm, minimum height=8mm] (c3) at (C3)
    {$x_{\max}\ge2^{-10}$: direct boxes $Q\in\mathcal Q$ ($222{,}958$ regions)};
  \node[flowbox, text width=70mm, minimum height=8mm] (t1) at (121mm,-5mm)
    {$r\ge10$: the perturbation vanishes; analytic tail bound~\eqref{eq:cert-tail}};
  \node[flowbox, text width=70mm, minimum height=8mm] (t2) at (121mm,-16mm)
    {$b_4\le r\le10$: backward comparison gives $\overline V$};
  \node[flowbox, text width=70mm, minimum height=8mm] (t3) at (121mm,-27mm)
    {$0\le r\le b_4$: threshold bounds, correction term, comparison};
  \node[flowbox, text width=70mm, minimum height=8mm] (t4) at (121mm,-38mm)
    {regional bound $U(\mathcal R)<U_{\mathrm{cert}}$};
  \foreach \i/\j in {1/2,2/3,3/4} { \draw[flowarrow] (t\i) -- (t\j); }
  \begin{scope}[on background layer]
    \node[draw=black!45, dashed, rounded corners=3pt, fit=(c0) (c3), inner sep=1.5mm] (cover) {};
    \node[draw=black!45, dashed, rounded corners=3pt, fit=(t1) (t4), inner sep=1.5mm] (region) {};
  \end{scope}
  \foreach \i in {1,2,3} { \draw[flowarrow, dashed] (c\i.east) -- (c\i.east -| region.west); }
  \node[flowbox, line width=0.9pt, fill=black!6, text width=150mm, minimum height=8mm] (fin) at (80mm,-54mm)
    {$\sup_{\lambda}\ell_{\Gamma_1}(K_1,\lambda)\le U_{\max}<U_{\mathrm{cert}}<1$
     by~\eqref{eq:cert-zero-scale} and Lemma~\ref{lem:cert-finite-data}; hence
     Theorem~\ref{thm:certified-kernel}};
  \draw[flowarrow] (cover.south) -- (cover.south |- fin.north);
  \draw[flowarrow] (region.south) -- (region.south |- fin.north);
\end{tikzpicture}
\caption{Structure of the computer-assisted proof of Theorem~\ref{thm:certified-kernel}.
Left: the cover of the set of mass vectors (Appendix~\ref{app:cert-cover}). Right: the
steps that produce the bound for a single region (Appendices~\ref{app:cert-regional-comparison}
and~\ref{app:cert-details}); dashed arrows indicate that every region of the three numerical
covers is treated in this way.}
\label{fig:certificate}
\end{figure}
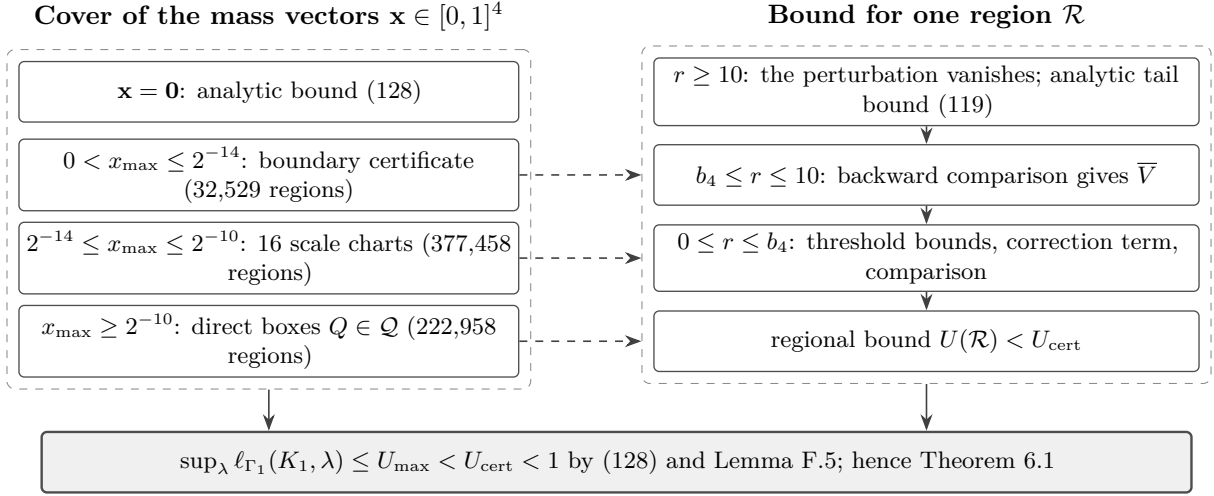

\begin{table}[htb]
\centering
\small
\begin{tabular}{@{}l p{9.4cm} l@{}}
\toprule
Symbol & Meaning & Reference\\
\midrule
$a$, $r$ & exponential rank and time: $u_a=1-e^{-a}$, $t_r=1-(1+r)e^{-r}$ & \eqref{eq:kzero-r-coordinates}\\
$h(r)$, $C(r)$ & cumulative hazard $h_\lambda(u_r)$ and winning probability $C_\lambda(u_r)$ & \eqref{eq:kzero-h}\\
$b_g$, $L_g$ & block endpoints $(0,\frac25,2,\frac{16}5,\frac{32}5)$ and lengths & \eqref{eq:cert-blocks}\\
$M_g$, $x_g$, $\mathbf x$ & block mass $h(b_{g+1})-h(b_g)$, block average $M_g/L_g$, and mass vector $(x_g)_g$ & \eqref{eq:cert-mass-map}\\
$M_{<p}$ & cumulative block mass $\sum_{g<p}M_g=h(b_p)$ & \eqref{eq:cert-mass-ranges}\\
$w_g$, $c_g$ & correction coefficients of block $g$ & \eqref{eq:cert-c-coeff}\\
$d$, $\nu$ & corrections of the allocation and of its rate & \eqref{eq:cert-dzk}, \eqref{eq:cert-zk}\\
$k$, $\sigma$ & rescaled allocation rate $h+\nu$ and survival factor & \eqref{eq:cert-kSR}, \eqref{eq:cert-abbreviations}\\
$R$ & rescaled threshold complement $e^r(1-Y(r))$ & \eqref{eq:cert-kSR}\\
$\chi$, $\Delta_\lambda$, $\Xi$, $\rho_*$ & gap $r-h(r)$, area gap, inverse correction, inverse coordinate & \S\ref{sec:cert-comparisons}\\
$\mathcal U_{\mathcal R}$, $\mathcal G$ & regional bound on $R+\Xi_+$ and correction term & \eqref{eq:cert-U-condition}, \eqref{eq:cert-credit-integrated}\\
$V_\lambda$, $\overline V$ & continuation integral and its bound at $b_4$ & \S\ref{app:cert-regional-comparison}\\
$T_0$, $W$, $P$ & i.i.d.\ continuation, scalar potential, state potential & \S\ref{app:cert-regional-comparison}\\
$x_{\max}$, $\hat x_g$, $m_g$, $\hat d$ & scale $\max_gx_g$, ratios $x_g/x_{\max}$, and normalized quantities $M_g/x_{\max}$ and $d/x_{\max}$ & \eqref{eq:cert-normalized-masses}, \S\ref{app:cert-inverse-bounds}\\
\bottomrule
\end{tabular}
\caption{Notation for the proof of Theorem~\ref{thm:certified-kernel}.}
\label{tab:cert-notation}
\end{table}

\subsection{Validity of the Kernel}
\label{app:cert-validity}\label{app:cert-coefficients}

\begin{lemma}[Kernel validity]\label{lem:cert-kernel-valid}
The density $p_1$ is nonnegative and has marginals
\begin{equation}\label{eq:cert-marginals}
 \int_0^\infty p_1(a,r)\dd r=e^{-a},\qquad
 \int_0^\infty p_1(a,r)\dd a=re^{-r}.
\end{equation}
Hence $K_1$ is an allocation kernel.  For every fixed $a$, both
$p_1(a,r)-p_0(a,r)$ and
$\int_0^r(p_1(a,v)-p_0(a,v))\dd v$ vanish for $r\ge10$.
\end{lemma}
\begin{proof}
The marginal identities follow from the zero row and column sums of
\eqref{eq:cert-delta}.  Nonnegativity is equivalent to the finite tests
\[
 \Delta\omega_j+\delta_{ij}\ge0\quad(i<j),\qquad
 \delta_{ij}\ge0\quad(i\ge j).
\]
The conditions with $i\ge j$ are rational comparisons.  For $i<j$,
the coefficient table and rational exponential bounds below verify
the displayed inequalities.
On a diagonal cell the nonnegative rectangular correction is added to
the baseline triangular density.
The first marginal makes $K_1(\cdot,u_a)$ a conditional distribution function.
Using $\dd u_a=e^{-a}\dd a$ and the second marginal gives
\[
 \int_0^1K_1(t_r,u)\dd u
 =\int_0^r ve^{-v}\dd v=t_r,
\]
which is~\eqref{eq:kernel-marginal}.
Finally, the correction is supported on
$r<10$ and its integral over every input row is zero.
\end{proof}

\begin{table}[htbp]
\centering
\small
\setlength{\tabcolsep}{4pt}
\begin{tabular}{r|r r r}
\toprule
$j$ & $N_{1,j}$ & $N_{2,j}$ & $N_{3,j}$\\
\midrule
1 & 139368429085091008 & 253843311272096320 & 32389009049152552 \\
2 & 5953487440217 & 1791759853683 & 71119674367 \\
3 & 3990741975321 & 1201052547606 & 47672943393 \\
4 & 2675074343691 & 805089599003 & 31956129611 \\
5 & -1179057110689020928 & 2698338485333 & 107104171361 \\
6 & 543811243036226816 & 1808750378235 & 71794073076 \\
7 & 187787420768742656 & 1212441636912 & 48125006371 \\
8 & 596185967112761600 & -1832980005978172416 & 51614650583 \\
9 & 423259029123334464 & -313737654730343936 & 34598334963 \\
10 & -691811398605077888 & 898120583974625280 & 23191957399 \\
11 & -184413906041231008 & 390602254327415360 & 15546034017 \\
12 & 34902452349492964 & 306734431836344000 & 10420818254 \\
13 & -37556789677386504 & 263005582702171680 & 6985283324 \\
14 & -47401672754419120 & 173208685273138400 & 11457312846336316 \\
15 & -96227740933766672 & 168402329898501888 & 22884320211898852 \\
16 & 275400707757524384 & 10295857615769126 & -114815326321236128 \\
17 & 130717989035219808 & -98670230297316768 & -10429295523610740 \\
18 & 77567271414955840 & -71065873407713880 & 172273975475231 \\
19 & -56519617666625056 & -50315851097015144 & 54605577509679488 \\
20 & -44238050328599312 & -38428024022805248 & 42002973082326256 \\
21 & -29653651932793608 & -25759074832025444 & -12472760155322594 \\
22 & -19877437328715028 & -16129817091862610 & -12759564154888524 \\
23 & -13324244705254784 & -11156299818298360 & -7869831686836923 \\
\bottomrule
\end{tabular}
\caption{Exact integer specification of the certified kernel.}
\label{tab:cert-kernel}
\end{table}

To complete the nonnegativity check in
Lemma~\ref{lem:cert-kernel-valid} we use rational bounds for the exponentials.  Define
\[
 \mathcal E_m(x):=\sum_{n=0}^m\frac{(-x)^n}{n!}.
\]
For $0\le x\le10$, the alternating-series remainder gives
$\mathcal E_{161}(x)\le e^{-x}\le \mathcal E_{160}(x)$; in particular, $\mathcal E_{161}(j\Delta)-\mathcal E_{160}((j+1)\Delta)\le\omega_j$.
Substituting the integers in Table~\ref{tab:cert-kernel} gives, for
every $0\le i<16$, $0\le j<25$ with $i<j$,
\[
 \delta_{ij}+\Delta\bigl(\mathcal E_{161}(j\Delta)
                       -\mathcal E_{160}((j+1)\Delta)\bigr)>10^{-11},
\]
and hence $\Delta\omega_j+\delta_{ij}>10^{-11}$.
All entries with $i\ge j$ satisfy $\delta_{ij}\ge0$ by exact integer
comparison.  The zero row and column sums follow directly from
\eqref{eq:cert-delta}.  These finite rational inequalities prove every
nonnegativity condition used in the kernel construction.

\subsection{Load Representation}\label{app:cert-load}

Fix a profile $\lambda\in\Lambda$.  We use the exponential coordinates
of Section~\ref{sec:kzero-exp} throughout the proof.
Recall from~\eqref{eq:kzero-h} that
\[
 h(r)=\int_0^r\lambda(u_a)\dd a,
 \qquad C(r)=\int_0^r\lambda(u_a)e^{-a}\dd a.
\]
Thus $h'(r)=\lambda(u_r)\in[0,1]$ almost everywhere, whereas
$C(r)$ is the winning probability accumulated up to maximum rank $u_r$.
We abbreviate
\begin{gather}
 D(r):=D_{\Gamma_1}(r)=c_{\Gamma_1}e^r+1+r,
 \qquad c_{\Gamma_1}=\frac{1-\Gamma_1}{\Gamma_1},
 \label{eq:cert-D}\\
 B:=B_{K_1,\lambda},\qquad
 Y(r):=y_{K_1,\lambda}(t_r),\qquad
 \sigma(r):=S_{K_1,\lambda}(t_r).
 \label{eq:cert-abbreviations}
\end{gather}
As in Section~\ref{sec:kzero-exp}, $B$ is a function of normalized
algorithmic time, while $Y$ and $\sigma$ are expressed in the exponential
coordinate.  Define the rescaled allocation rate and the rescaled threshold complement (as in Section~\ref{sec:kzero-lag}) by
\begin{equation}\label{eq:cert-kSR}
 k(r):=e^r\frac{\dd}{\dd r}B(t_r)=rB'(t_r),
 \qquad R(r):=e^r(1-Y(r)).
\end{equation}
The rate identities hold almost everywhere for $r>0$.
Since $B$ is nondecreasing and $1$-Lipschitz, $0\le k(r)\le r$.
The survival factor in~\eqref{eq:scalar-paths} therefore becomes
\[
 \sigma(r)=\exp\!\left(-\int_0^r\frac{k(v)}{D(v)}\dd v\right).
\]
For the baseline kernel, $k=h$; the perturbation changes this rate.

Define the auxiliary upper bound
\begin{equation}\label{eq:cert-augmented-load}
 \overline\ell_\lambda
 :=\ell_{\Gamma_1}(K_1,\lambda)+\Gamma_1Y(0).
\end{equation}
The added term is nonnegative, so a uniform upper bound below $1$ for
$\overline\ell_\lambda$ proves feasibility.  It also makes the boundary
term below equal to $\Gamma_1$ for every profile.  This modification is
used only in the analysis; the arrival construction still assigns no
mass to time zero.

\begin{lemma}[Load representation]\label{lem:cert-load-identity}
For every $\lambda\in\Lambda$,
\begin{equation}\label{eq:cert-load-identity}
 \ell_{\Gamma_1}(K_1,\lambda)\le\overline\ell_\lambda
 =\Gamma_1+\int_0^\infty
 \sigma(r)\frac{(r-k(r))R(r)}{D(r)^2}\dd r.
\end{equation}
\end{lemma}
\begin{proof}
Substitute $t=t_r$ into~\eqref{eq:load-parts}.
Since $\dd t=re^{-r}\dd r$,
$B'(t_r)=k(r)/r$ for almost every $r>0$, and
$1-\Gamma_1t_r=\Gamma_1e^{-r}D(r)$, the integral has the stated form.
The boundary term in~\eqref{eq:load-parts} is
$\Gamma_1(1-Y(0))$; adding $\Gamma_1Y(0)$ proves the identity.
\end{proof}

The rest of the proof bounds the integral in
\eqref{eq:cert-load-identity} uniformly over $\lambda$.
Its dependence on the allocation is captured by $k$ and $\sigma$;
the remaining difficulty is the threshold term $R$, which involves
the inverse of the profile's area curve.

\subsection{Block Parameters and Admissible States}\label{sec:cert-finite}

The correction rows in~\eqref{eq:cert-p} are constant on four blocks in $a$.
Write
\begin{equation}\label{eq:cert-blocks}
 (b_0,b_1,b_2,b_3,b_4)
 =\left(0,\frac25,2,\frac{16}5,\frac{32}5\right),
\end{equation}
and define
\begin{equation}\label{eq:cert-mass-map}
 L_g=b_{g+1}-b_g,\quad
 M_g=\int_{b_g}^{b_{g+1}}\lambda(u_a)\dd a,\quad
 x_g=M_g/L_g,\quad \mathbf x(\lambda)=(x_0,x_1,x_2,x_3)\in[0,1]^4.
\end{equation}
Here $g\in\{0,1,2,3\}$.  The block mass $M_g$ is the increment
$h(b_{g+1})-h(b_g)$ of the cumulative hazard; it is not the winning
probability $\mu_\lambda=\int_0^1\lambda(u)\dd u$.
Thus $x_g$ is the average of $\lambda(u_a)$ on that block, and we call $\mathbf x(\lambda)$ the \emph{mass vector} of $\lambda$.
Let $w_g(r)$ be the common value of the correction coefficient on block $g$;
that is,
$e^r(p_1(a,r)-p_0(a,r))=w_g(r)$ for $b_g\le a<b_{g+1}$.
Set
\begin{equation}\label{eq:cert-c-coeff}
 c_g(r)=e^r\int_0^r e^{-v}w_g(v)\dd v.
\end{equation}
On each output interval $J_j$, $w_g$ is constant and
$c_g'(r)=c_g(r)+w_g(r)$.  Direct integration therefore gives
\begin{align}
 B(t_r)&=C(r)-h(r)e^{-r}+d(r)e^{-r},&
 d(r)&=\sum_g c_g(r)M_g,\label{eq:cert-dzk}\\
 k(r)&=h(r)+\nu(r),& \nu(r)&=\sum_g w_g(r)M_g.
 \label{eq:cert-zk}
\end{align}
Thus the allocation and rate corrections depend only on the four
block masses; $k$ additionally depends on the current state $h(r)$.
The threshold term $R$ still depends on the profile within each block.
We bound it uniformly over these arbitrary measurable profiles using
the same masses and state.

For a closed parameter box
\begin{equation}\label{eq:cert-mass-box}
 \mathcal R=\prod_{g=0}^3[\underline x_g,\overline x_g],
 \qquad 0\le\underline x_g\le\overline x_g\le1,
\end{equation}
write
\[
 \Lambda(\mathcal R)=\{\lambda\in\Lambda:
                        \mathbf x(\lambda)\in\mathcal R\}.
\]
Its block-mass bounds are
\begin{equation}\label{eq:cert-mass-ranges}
 \underline M_g=L_g\underline x_g,\quad
 \overline M_g=L_g\overline x_g,\qquad
 \underline M_{<p}=\sum_{g<p}\underline M_g,\quad
 \overline M_{<p}=\sum_{g<p}\overline M_g.
\end{equation}
For a particular profile, put $M_{<p}=\sum_{g<p}M_g$, so that $M_{<p}=h(b_p)$.
On $b_p\le r\le b_{p+1}$,
\begin{equation}\label{eq:cert-path-cone-exact}
 \max\{M_{<p},M_{<p+1}+r-b_{p+1}\}
 \le h(r)\le\min\{M_{<p+1},M_{<p}+r-b_p\}.
\end{equation}
Indeed, $h(r)-M_{<p}$ lies in $[0,r-b_p]$, and
$M_{<p+1}-h(r)$ lies in $[0,b_{p+1}-r]$.
Consequently every $\lambda\in\Lambda(\mathcal R)$ satisfies
\begin{equation}\label{eq:cert-path-cone}
 h^-_{\mathcal R}(r)\le h(r)\le h^+_{\mathcal R}(r),
\end{equation}
where
\begin{align*}
 h^-_{\mathcal R}(r)&=\max\{\underline M_{<p},
                              \underline M_{<p+1}+r-b_{p+1}\},\\
 h^+_{\mathcal R}(r)&=\min\{\overline M_{<p+1},
                              \overline M_{<p}+r-b_p\},
\end{align*}
with $p$ chosen so that $b_p\le r\le b_{p+1}$.
For $r\ge b_4$, extend these functions by
\begin{equation}\label{eq:cert-tail-state}
 h^-_{\mathcal R}(r)=\underline M_{<4},\qquad
 h^+_{\mathcal R}(r)=\overline M_{<4}+r-b_4,
\end{equation}
which again enclose every actual state.

At each fixed $r$, we use the admissible set
\begin{equation}\label{eq:cert-state-domain}
 \mathcal D_{\mathcal R}(r)=
 \{(\mathbf x,h):\mathbf x\in\mathcal R,\
 h^-_{\mathcal R}(r)\le h\le h^+_{\mathcal R}(r),\
 0\le h+\nu(r;\mathbf x)\le r\}.
\end{equation}
Here $\nu(r;\mathbf x)=\sum_g w_g(r)L_gx_g$, and similarly for $d$.
For a parameter region given by a change of coordinates, the same notation
is used with its image in place of $\mathcal R$; mass bounds follow from a
containing box, and the defining coordinate constraints are retained in
$\mathcal D_{\mathcal R}(r)$.
Every profile with parameters in the region satisfies these constraints.
We may therefore obtain a uniform upper bound by maximizing over these
sets, even if they also contain states that no such profile realizes.

Near the zero mass vector, it is useful to keep the total scale
separate from the relative block masses.
Let
\begin{equation}\label{eq:cert-scale}
 x_{\max}:=\|\mathbf x(\lambda)\|_\infty
 =\max_{0\le g\le3}x_g.
\end{equation}
For $x_{\max}>0$, choose a pivot
$
 g^*\in\arg\max_{0\le g\le3}x_g$
and separate the overall scale from the relative shape by writing
\begin{equation}\label{eq:cert-normalized-masses}
 \hat x_g:=\frac{x_g}{x_{\max}},
 \qquad
 \hat x_{g^*}=1,
 \qquad
 m_g:=L_g\hat x_g,
 \qquad
 M_g=x_{\max}m_g=x_{\max}L_g\hat x_g.
\end{equation}
The ratios $\hat x_g$ lie in $[0,1]$, whereas $m_g$ are the normalized block
masses used in the inverse estimates.
These definitions do not impose a positive lower bound on
$x_{\max}$.  They apply to every positive scale, however small.
The case $x_{\max}=0$ is treated directly when the regions are assembled
in Appendix~\ref{app:cert-cover}.

\subsection{Bounding the Threshold Complement}\label{sec:cert-comparisons}

The load formula involves $R$, which is defined through the inverse
of $A_\lambda$.  We first bound this inverse and then use a comparison
inequality to control all admissible paths $h$ at once.

Put $\chi(r):=r-h(r)$, so that
$x_\lambda(u_r)=1-e^{-\chi(r)}$.
Equation~\eqref{eq:profile-area-main} becomes
\begin{equation}\label{eq:cert-source-identity}
 A_\lambda(1-e^{-\chi(a)})
 =C(a)+e^{-a}-e^{-\chi(a)}.
\end{equation}
For $0\le x<1$, define
\[
 \rho_\lambda(x):=\inf\{a\ge0:1-e^{-\chi(a)}>x\},
 \qquad \rho_\lambda(1):=\infty,
 \qquad \rho_*(r):=\rho_\lambda(Y(r)),
\]
with $\inf\varnothing=\infty$.
The coordinate $\rho_*(r)$ identifies the point on the rank
representation that corresponds to the probability threshold $Y(r)$.
The strict inequality in the definition implements the same
right-endpoint convention as $A_\lambda^\leftarrow$.
Thus $\rho_\lambda$ is a generalized inverse of the loss-score map $a\mapsto1-e^{-\chi(a)}$; in the notation of Lemma~\ref{lem:profile-recovery}, $\rho_\lambda(x)=-\log\bigl(1-\rho_{A_\lambda}(x)\bigr)$ (see the proof of that lemma).
If $Y(r)$ lies in the terminal slope-one extension, then
$\rho_*(r)=\infty$.  In the formulas below, we use
$C(\infty)=\mu_\lambda$ and $e^{-\infty}=0$.
By the upper-inverse definition, $A_\lambda(Y(r))=B(t_r)$.
If $\rho_*(r)<\infty$, apply~\eqref{eq:cert-source-identity} at
$a=\rho_*(r)$; if $\rho_*(r)=\infty$, use the terminal extension
\eqref{eq:profile-extension-main} (or $A_\lambda(1)=\mu_\lambda$ when $Y(r)=1$).
Together with~\eqref{eq:cert-dzk}, this gives
\begin{align}
 R(r)
 &=e^r\bigl(C(\rho_*(r))-C(r)+e^{-\rho_*(r)}\bigr)
   +h(r)-d(r) \notag\\
 &=1+h(r)-d(r)-\Xi(r), \label{eq:cert-inverse-identity}
\end{align}
where
\[
 \Xi(r):=
 e^r\int_r^{\rho_*(r)}
 (1-\lambda(u_a))e^{-a}\dd a.
\]
To determine on which side of $r$ the inverse coordinate
$\rho_*(r)$ lies, define
\begin{equation}\label{eq:cert-area-gap}
 \Delta_\lambda(r)
 :=e^r\{B(t_r)-A_\lambda(1-e^{-\chi(r)})\}
 =e^{h(r)}-1-h(r)+d(r).
\end{equation}
The upper-inverse convention gives $\rho_*(r)\ge r$ when
$\Delta_\lambda(r)\ge0$, and $\rho_*(r)<r$ when
$\Delta_\lambda(r)<0$.

Write $\Xi_+(r)=\max\{\Xi(r),0\}$.  On the nonnegative-gap branch,
\begin{equation}\label{eq:cert-positive-basic}
 R(r)+\Xi_+(r)=1+h(r)-d(r);
\end{equation}
on the negative-gap branch, $\Xi_+=0$ and
\begin{equation}\label{eq:cert-negative-basic}
 R(r)\le e^r.
\end{equation}
Thus the branch-independent bound
\begin{equation}\label{eq:cert-U-fallback}
 R(r)+\Xi_+(r)\le\max\{1+h(r)-d(r),e^r\}
\end{equation}
is always available.  The sharper estimates in
Appendix~\ref{app:cert-inverse-bounds} use the block masses and the
bounds on $h$.  We write their regional upper bound in the form
\begin{equation}\label{eq:cert-U-condition}
 R(r)+\Xi_+(r)\le
 \mathcal U_{\mathcal R}(r,\mathbf x(\lambda),h(r)),
 \qquad \lambda\in\Lambda(\mathcal R).
\end{equation}

The bound~\eqref{eq:cert-U-condition} gives
\[
 R(r)\le
 \mathcal U_{\mathcal R}(r,\mathbf x(\lambda),h(r))-\Xi_+(r).
\]
Using the representation of $\Xi_+$ in~\eqref{eq:cert-inverse-identity}
and changing the order of integration, Appendix~\ref{app:cert-credit}
constructs $\mathcal G\ge0$ such that
\begin{equation}\label{eq:cert-credit-integrated}
 \int_0^{b_4}\sigma(r)\frac{(r-k(r))\Xi_+(r)}{D(r)^2}\dd r
 \ge\int_0^{b_4}\sigma(r)(1-\lambda(u_r))
 \mathcal G(r,\mathbf x(\lambda),h(r))\dd r
\end{equation}
for every admissible profile.  We call $\mathcal G$ the
\emph{correction term}, since it provides a nonnegative subtraction
from the upper bound on the integral.  If no positive bound is certified,
we take $\mathcal G=0$.  Along a fixed profile path, we suppress the arguments
$\mathbf x(\lambda)$ and $h(r)$ in $\mathcal U_{\mathcal R}$ and $\mathcal G$.

\subsection{Regional Comparison and Tail Bounds}
\label{app:cert-regional-comparison}\label{app:cert-bellman}

We now convert bounds on the threshold complement into a bound on the
load of every profile in a fixed parameter region $\mathcal R$.
We first control the integral beyond $b_4$, then the integral on $[0,b_4]$.

\subsubsection{Continuation integrals and backward comparison}

To separate the finite interval from the tail, define the continuation
term for a fixed profile by
\[
 V_\lambda(a):=
 \int_a^\infty
 \frac{\sigma(r)}{\sigma(a)}
 \frac{(r-k(r))R(r)}{D(r)^2}\dd r.
\]
This is the remaining integral contribution per unit of survival at
$a$, for a fixed profile; unlike $V_\Gamma$ in
Section~\ref{sec:kzero-control}, it takes no supremum.
The load splits at the last input-block endpoint $b_4=32/5$ as
\[
 \overline\ell_\lambda
 =\Gamma_1+
 \int_0^{b_4}\sigma(r)\frac{(r-k(r))R(r)}{D(r)^2}\dd r
 +\sigma(b_4)V_\lambda(b_4).
\]
The perturbation persists until $r=10$, so we first bound the tail
beyond $10$, then propagate that bound back to $b_4$.

For $0\le a<b$, write $\sigma_a(r):=\sigma(r)/\sigma(a)$.  Hence
$\sigma_a'=-k\sigma_a/D$, $\sigma_a(a)=1$, and
\[
 V_\lambda(a)
 =\int_a^b \sigma_a(r)\frac{(r-k(r))R(r)}{D(r)^2}\dd r
 +\sigma_a(b)V_\lambda(b).
\]
If $V_\lambda(b)\le v$ and $R\le\overline R$ on $[a,b]$, then
\begin{equation}\label{eq:cert-direct-one-step}
 V_\lambda(a)
 \le v+\int_a^b \sigma_a(r)
 \left[
 \frac{(r-k(r))\overline R(r)}{D(r)^2}
 -\frac{k(r)v}{D(r)}
 \right]\dd r.
\end{equation}
This comparison leads to sufficient conditions for propagating a
continuation bound backward.  The following proposition gives the basic finite update.
The state-band refinement is given in Lemma~\ref{lem:cert-dp-chord}.

\begin{proposition}[Direct regional comparison]
\label{prop:cert-direct}
Let $0\le r_0<r_1<\cdots<r_N<\infty$, and suppose that
$v_N\ge0$ satisfies $V_\lambda(r_N)\le v_N$ for every
$\lambda\in\Lambda(\mathcal R)$.  Suppose that, for
$j=N-1,\ldots,0$, there is a uniform upper bound
$\overline R_j(r)\ge R(r)$ on $[r_j,r_{j+1}]$ and numbers
$H_j$ and $\kappa_j\ge0$ such that
\begin{equation}\label{eq:cert-direct-residual}
 \frac{(r-k(r))\overline R_j(r)}{D(r)^2}
 -\frac{k(r)v_{j+1}}{D(r)}
 \le H_j,
 \qquad
 \frac{k(r)}{D(r)}\le \kappa_j
\end{equation}
throughout the admissible mass and state ranges on that interval.
Choose recursively
\begin{equation}\label{eq:cert-direct-update}
 v_j\ge
 \max\left\{0,\,
 v_{j+1}+
 \begin{cases}
 (r_{j+1}-r_j)H_j,
     &H_j\ge0,\\[2mm]
 \displaystyle
 \frac{(r_{j+1}-r_j)H_j}
 {1+(r_{j+1}-r_j)\kappa_j},
     &H_j<0.
 \end{cases}
 \right\}.
\end{equation}
Then
\[
 V_\lambda(r_j)\le v_j
 \qquad
 (j=0,\ldots,N)
\]
for every $\lambda\in\Lambda(\mathcal R)$.  If $r_0=0$,
\begin{equation}\label{eq:cert-direct-result}
 \overline\ell_\lambda
 \le \Gamma_1+v_0.
\end{equation}
\end{proposition}

\begin{proof}
We argue backward from the assumed terminal inequality
$V_\lambda(r_N)\le v_N$.  Suppose that
$V_\lambda(r_{j+1})\le v_{j+1}$.  By~\eqref{eq:cert-direct-one-step},
\[
 V_\lambda(r_j)\le v_{j+1}
 +\int_{r_j}^{r_{j+1}}\sigma_{r_j}(r)
 \left[
 \frac{(r-k(r))\overline R_j(r)}{D(r)^2}
 -\frac{k(r)v_{j+1}}{D(r)}
 \right]\dd r.
\]
If $H_j\ge0$, using $\sigma_{r_j}(r)\le1$ gives an increment at most
$(r_{j+1}-r_j)H_j$. If $H_j<0$, the bound $k/D\le \kappa_j$ implies $\sigma_{r_j}(r)\ge e^{-\kappa_j(r-r_j)}$ using $\sigma_a'(r)=-\frac{k(r)}{D(r)}\sigma_a(r)$
and hence
\[
 \int_{r_j}^{r_{j+1}}\sigma_{r_j}(r)\dd r
 \ge
 \frac{r_{j+1}-r_j}
 {1+(r_{j+1}-r_j)\kappa_j}.
\]
Multiplying by $H_j<0$ gives the second line of
\eqref{eq:cert-direct-update}.  Thus
$V_\lambda(r_j)\le v_j$, completing the backward induction. If $r_0=0$, then $\sigma(0)=1$, so
$
 \overline\ell_\lambda
 =\Gamma_1+V_\lambda(0)
 \le\Gamma_1+v_0$.
\end{proof}

\subsubsection{The integral beyond the last block}
\label{app:cert-tail-bound}

We apply the comparison to the part of the integral beyond $b_4$.
For each region, we need a bound
\[
 \overline V\ge
 \sup_{\lambda\in\Lambda(\mathcal R)}V_\lambda(b_4).
\]
We construct this bound by combining an analytic estimate on
$[10,\infty)$ with a finite backward comparison on $[b_4,10]$.

For $r\ge10$, the perturbation has vanished, so
$d(r)=\nu(r)=0$.  Hence $k(r)=h(r)$ and $B(t_r)=B_{K_0,\lambda}(t_r)$.
Lemma~\ref{lem:kzero-inverse} gives $R(r)\le1+h(r)$.  Since
$
 (r-h(r))(1+h(r))\le\frac{(r+1)^2}{4}$
 and $
 D(r)\ge c_{\Gamma_1} e^r$,
we obtain the terminal estimate
\begin{equation}\label{eq:cert-tail}
 V_\lambda(10)
 \le
 \int_{10}^\infty
 \frac{(r+1)^2e^{-2r}}{4c_{\Gamma_1}^2}\dd r
 =
 \frac{265e^{-20}}{16c_{\Gamma_1}^2}.
\end{equation}
Starting from~\eqref{eq:cert-tail}, partition $[b_4,10]$ into finitely
many intervals of the exponential coordinate and apply
Proposition~\ref{prop:cert-direct} backward.  The finite-tail construction in
Appendix~\ref{app:cert-direct-updates} thereby yields the required
certified bound $\overline V$ at $b_4$.

\subsubsection{The prefix and a state-dependent comparison}

The bound at $b_4$ leaves the integral over $[0,b_4]$.
For every profile in the region, use
$R\le\mathcal U_{\mathcal R}-\Xi_+$ from~\eqref{eq:cert-U-condition}.
The integrated correction~\eqref{eq:cert-credit-integrated} bounds
the contribution of $\Xi_+$ from below, so subtracting it gives an
upper bound on the load.  Together with
$V_\lambda(b_4)\le\overline V$, this yields
\begin{equation}\label{eq:cert-prefix-main}
 \overline\ell_\lambda
 \le\Gamma_1+
 \int_0^{b_4}\sigma(r)
 \left[
 \frac{(r-k(r))\mathcal U_{\mathcal R}(r)}{D(r)^2}
 -(1-\lambda(u_r))\mathcal G(r)
 \right]\dd r
 +\sigma(b_4)\overline V.
\end{equation}
Taking $\mathcal G=0$ is valid; a positive correction improves the bound.

Section~\ref{sec:kzero-control} maximized over $0\le k\le r$
independently at each coordinate.  Here a comparison depending on the
state $h$ can also use $h'(r)=\lambda(u_r)\in[0,1]$.
Fix a continuous piecewise-bilinear function $P=P_{\mathcal R}$ on a
finite grid, with vertex values independent of the particular profile,
and set
\[
 T_0(r):=\int_r^{10}\frac{v}{D(v)^2}\dd v,
 \qquad
 \mathcal V(r,h):=T_0(r)+W(r)-P(r,h),
\]
where $W$ is absolutely continuous.  Once $P$ is fixed, only the scalar
function $W(r)$ remains to be found.  We choose it so that the decrease
of $\sigma\mathcal V$ bounds the integrand in~\eqref{eq:cert-prefix-main}.
Appendix~\ref{app:cert-potential-construction} specifies the fixed
potential used in the computation.

\begin{proposition}[Regional Bellman certificate]
\label{prop:cert-bellman}
Suppose that~\eqref{eq:cert-U-condition} and~\eqref{eq:cert-credit-integrated}
hold, that $V_\lambda(b_4)\le\overline V$ uniformly on the region,
and that, for almost every $r\in[0,b_4]$,
\begin{equation}\label{eq:cert-HJB}
 -W'(r)\ge
 \sup_{(\mathbf x,h)\in\mathcal D_{\mathcal R}(r)}
 \left\{
 \frac{(r-k)\mathcal U_{\mathcal R}-r}{D^2}
 -\frac{k(T_0+W-P)}{D}
 -P_r-\min\{\mathcal G,P_h\}
 \right\}.
\end{equation}
At interpolation boundaries use the corresponding one-sided derivatives.
If $\mathcal V(b_4,h)\ge\overline V$ for every admissible terminal state,
then, for every $\lambda\in\Lambda(\mathcal R)$,
\begin{equation}\label{eq:cert-bellman-result}
 \overline\ell_\lambda
 \le\Gamma_1+T_0(0)+W(0)-P(0,0).
\end{equation}
\end{proposition}

\begin{proof}
Using $\sigma'=-k\sigma/D$ and $h'=\lambda(u_r)$ gives
\begin{equation}\label{eq:cert-bellman-derivative}
 \frac{(\sigma\mathcal V(r,h(r)))'}{\sigma}
 =-\frac{k\mathcal V}{D}-\frac r{D^2}
   +W'-P_r-\lambda(u_r)P_h.
\end{equation}
The current hazard rate enters affinely, and
\begin{equation}\label{eq:cert-control-max}
 \max_{0\le\alpha\le1}
 \{-\mathcal G+\alpha(\mathcal G-P_h)\}
 =-\min\{\mathcal G,P_h\}.
\end{equation}
Consequently~\eqref{eq:cert-HJB} implies
\[
 \frac{\dd}{\dd r}[\sigma(r)\mathcal V(r,h(r))]
 \le-\sigma(r)\left[
 \frac{(r-k(r))\mathcal U_{\mathcal R}(r)}{D(r)^2}
 -(1-\lambda(u_r))\mathcal G(r)\right]
\]
almost everywhere.  The composition is absolutely continuous because
$W,h$ are absolutely continuous and $P$ is continuous and piecewise
bilinear on a finite grid; the one-sided checks cover grid boundaries.
Integrating from $0$ to $b_4$, using $h(0)=0$, $\sigma(0)=1$ and the
terminal condition, proves the claim by~\eqref{eq:cert-prefix-main}.
\end{proof}

We call the expression maximized in~\eqref{eq:cert-HJB} the
\emph{Bellman residual}.  Once $P$ and the value of $W(r)$ are fixed,
it depends only on $(r,\mathbf x,h)$; the finite verification bounds
it on whole parameter--state cells.

\subsection{Finite Verification and Completion of the Proof}
\label{sec:cert-results}\label{sec:cert-scheme-overview}

The preceding comparisons are valid for every profile in a parameter
region.  The computation chooses sufficient bounds on each region and
checks that the accepted regions cover every possible block-mass vector.

For each parameter region $\mathcal R$, the certificate specifies a
comparison method and finite grids that give a rigorous bound
\begin{equation}\label{eq:cert-regional-verifier-goal}
 \sup_{\lambda\in\Lambda(\mathcal R)}
 \overline\ell_\lambda\le U(\mathcal R).
\end{equation}
Three comparisons are used.  The state-band comparison uses
Lemma~\ref{lem:cert-dp-chord} to retain the dependence of $R$ and $k$
on the running state $h$.  The Bellman comparison additionally uses
the correction $\mathcal G$ and a state-dependent potential, through
Proposition~\ref{prop:cert-bellman} and
Lemma~\ref{lem:cert-bellman-update}.  On some small-scale boundary
regions, the forward comparison~\eqref{eq:cert-forward-bound} suffices.
All three use the same continuation bound~$\overline V$.

Appendix~\ref{sec:cert-schemes} specifies the inverse and residual
bounds available to these comparisons.  A regional record fixes which
bounds are used; verification does not repeat the search for that choice.
Whenever an estimate assumes that $w_g$ is constant, its coordinate
interval is intersected with the output intervals $J_j$.
Every $b_p$ is an output-cell boundary, so this also fixes the input
block containing $r$ on $[0,b_4]$.  Endpoint bounds use the one-sided
values from the relevant interval.  The acceptance condition is
\begin{equation}\label{eq:cert-scheme-accept}
 U(\mathcal R)<U_{\mathrm{cert}}.
\end{equation}

\subsubsection{Covering all profiles}\label{app:cert-cover}

The master certificate separates the small-scale cube and partitions
the remaining region into closed boxes with disjoint interiors,
\begin{equation}\label{eq:cert-global-partition}
 [0,1]^4=[0,2^{-10}]^4\cup\bigcup_{Q\in\mathcal Q} Q,
 \qquad Q=\prod_{g=0}^3[\underline x_g,\overline x_g].
\end{equation}
The small-scale cube uses the cover below.  Every $Q\in\mathcal Q$
satisfies $\max_g\underline x_g\ge2^{-10}$ and is verified directly.
Coverage is checked exactly by embedding the dyadic endpoints in a
common integer grid and checking that every internal node is the union
of its two complementary children.  The same check applies to each
auxiliary chart.

Fix $\lambda\in\Lambda$ and its mass vector $\mathbf x(\lambda)$.
If $x_{\max}=0$, then $\lambda(u_r)=0$ for almost every
$r\in[0,b_4]$.  The area curve $A_\lambda$ therefore vanishes on
$[0,1-e^{-b_4}]$.  Since $B(0)=0$, the generalized inverse defining
$Y(0)$ gives $Y(0)\ge1-e^{-b_4}$.
Using~\eqref{eq:scalar-load} and $S_{K_1,\lambda}\le1$ gives
\begin{equation}\label{eq:cert-zero-scale}
 \begin{aligned}
 \overline\ell_\lambda
 &\le\Gamma_1Y(0)+\frac{\Gamma_1}{1-\Gamma_1}(1-Y(0))\\
 &\le\Gamma_1+\frac{\Gamma_1^2e^{-b_4}}{1-\Gamma_1}
 <\Gamma_1+\frac{\Gamma_1^2}{600(1-\Gamma_1)}
 <0.749<U_{\mathrm{cert}}.
 \end{aligned}
\end{equation}

For positive scale, use the normalized parameters
from~\eqref{eq:cert-normalized-masses}.
A finite boundary certificate covers
$
0<x_{\max}\le2^{-14}$; its estimates are uniform in the scale, so they remain valid on pieces whose lower scale endpoint is zero (Appendix~\ref{app:cert-inverse-bounds}).
Boundary pieces also use the finite-scale forward comparison
\eqref{eq:cert-forward-bound}, with the applicable inverse estimates,
the certified survival factors, and the same continuation bound
$\overline V$.
For $2^{-14}<x_{\max}\le2^{-10}$, there are sixteen parameter charts, one for each pivot and each dyadic scale band:
\begin{equation}\label{eq:cert-scale-charts}
 g^*\in\{0,1,2,3\},\qquad
 (\hat x_g)_{g\ne g^*}\in[0,1]^3,\qquad
 x_{\max}\in[2^{-14+s},2^{-13+s}],\quad s\in\{0,1,2,3\}.
\end{equation}
Each chart is affinely rescaled to a unit
four-dimensional box before adaptive subdivision.  Ties between pivots
and common scale endpoints cause no ambiguity, since all enclosures are
closed and it is enough that one valid chart contain the point.
Equation~\eqref{eq:cert-normalized-masses} maps every chart box back to
rigorous bounds on the original masses $M_g$, with outward enclosure of
the full image.
For $x_{\max}>2^{-10}$, the original coordinates $\mathbf x$ are
handled directly by the boxes in $\mathcal Q$.
The two covers may share the boundary $x_{\max}=2^{-10}$; both use
closed enclosures.

These cases cover every mass vector, including all positive scales
arbitrarily close to zero.  Closed chart images include ties and shared
boundaries.  Thus every profile is either treated analytically or belongs
to a certified piece of the master cover.

\subsubsection{The numerical certificate}

Each accepted piece has a uniform bound below $U_{\mathrm{cert}}$;
the bound for a region covered by such pieces is their maximum.

The certificate contains $32{,}529$ boundary regions, $377{,}458$
scale-chart regions, and $222{,}958$ directly verified boxes in $\mathcal Q$.
Thus a complete replay performs $632{,}945$ regional interval calculations.

\begin{lemma}[Finite numerical certificate]\label{lem:cert-finite-data}
The supplementary regional certificate for $(K_1,\Gamma_1)$ covers
$[0,1]^4\setminus\{\mathbf0\}$ and gives
\begin{equation}\label{eq:cert-regional-upper}
 \sup_{\lambda\in\Lambda(\mathcal R)}\overline\ell_\lambda
 \le U(\mathcal R)
\end{equation}
on every region in the cover.  The largest assembled endpoint is
\begin{equation}\label{eq:cert-max}
 U_{\max}:=\frac{9007199253707871}{9007199254740992}<U_{\mathrm{cert}}<1.
\end{equation}
\end{lemma}

This lemma is the computer-assisted step; its finite inputs and
verification source are supplied in
\nolinkurl{order_selection_074500001_certificate_final.zip}.
The package contains one set of regional results.  The program
\texttt{verify.py} checks the exact kernel, complete parameter cover,
input identities, and exact rational comparison in~\eqref{eq:cert-max}.
With \texttt{--recompute}, it rebuilds the interval engines and evaluates
every numerical region from the included inputs before assembling the
bounds.  No stored endpoint is an input to a regional calculation.
The estimates account for quadrature, interpolation, derivative, tail,
and rounding errors on whole cells; Appendix~\ref{app:cert-arithmetic}
specifies the arithmetic.  Both the complete check and the full replay
are independent of earlier certificates and search logs.

\begin{proof}[Proof of Theorem~\ref{thm:certified-kernel}]
Lemma~\ref{lem:cert-kernel-valid} proves that $K_1$ is an allocation
kernel.  Fix $\lambda\in\Lambda$.
If $\mathbf x(\lambda)=\mathbf0$, then~\eqref{eq:cert-zero-scale}
gives $\overline\ell_\lambda<0.749<U_{\mathrm{cert}}$.
Otherwise, its mass vector lies in a region covered by
Lemma~\ref{lem:cert-finite-data}, so
$\overline\ell_\lambda\le U_{\max}<U_{\mathrm{cert}}$.
Equation~\eqref{eq:cert-load-identity} gives
$\ell_{\Gamma_1}(K_1,\lambda)\le\overline\ell_\lambda$ in both cases.
Proposition~\ref{prop:scalar-feasibility} now gives universal
feasibility, the exact item-wise guarantee for atomless inputs,
and the competitive guarantee for arbitrary independent
nonnegative distributions.
\end{proof}

\section{Estimates and Finite Updates for the Certificate}
\label{app:cert-details}

This appendix supplies the estimates used in Appendix~\ref{app:cert-proof}.
We retain its notation, with profiles on $[0,1]$ and hazard rate
$h'(r)=\lambda(u_r)$ in exponential coordinates.
The first two subsections bound the threshold complement $R$ and the
correction term $\mathcal G$.
The remaining subsections give the finite comparison updates, their
regional instantiations, the potential construction, and the arithmetic
used to certify every bound.

All estimates hold uniformly over their stated parameter and state
ranges.  Their target is the auxiliary quantity $\overline\ell_\lambda$
defined in~\eqref{eq:cert-augmented-load}, which is an upper bound on
the arrival load $\ell_{\Gamma_1}(K_1,\lambda)$.

\subsection{Bounds for the Inverse of \texorpdfstring{$A_\lambda$}{A lambda}}
\label{app:cert-inverse-bounds}

\subsubsection{A comparison with the baseline}
Write $B_0:=B_{K_0,\lambda}$ for the allocation induced by the
baseline kernel.  By~\eqref{eq:kzero-r-allocation},
$B_0(t_r)=C(r)-h(r)e^{-r}$.

\begin{lemma}[Inverse bound by a shift]
\label{lem:cert-shift}
For $r\ge0$, define
\[
 \upsilon_0(r):=1+\max_g(-c_g(r))_+,
 \qquad
 \overline R_0(r,h):=\upsilon_0(r)(1+h).
\]
Then
\begin{equation}\label{eq:cert-shift-bound}
 R(r)\le\overline R_0(r,h(r)).
\end{equation}
More generally, if $1\le\upsilon\le e^r$ and
\begin{equation}\label{eq:cert-prefix-budget}
 \sum_{g:\,b_{g+1}\le r}\underline M_g
 \min\{\upsilon-1,e^{r-b_{g+1}}-1\}
 \ge-\underline d(r),
\end{equation}
where $\underline d(r)\le d(r)$, then
$
 R(r)\le\upsilon(1+h(r))$.
Either comparison also gives
$
 \rho_*(r)\ge(r-\log\upsilon)_+$.
\end{lemma}
\begin{proof}
At input coordinate $a$, the cumulative baseline density is
$(e^{-a}-e^{-r})_+$.  On block $g$, the perturbed cumulative density is
this quantity plus $e^{-r}c_g(r)$; beyond $b_4$ the correction
is zero.  If $a<r-\log\upsilon_0$, it is at least
$e^{-a}-\upsilon_0e^{-r}$.
At all other input coordinates, nonnegativity of the kernel supplies the
lower bound zero.  Integrating against $\lambda(u_a)$ gives
\[
 B(t_r)\ge B_0(t_{(r-\log\upsilon_0)_+}).
\]
For a nonnegative shifted coordinate $r'$, monotonicity of the upper inverse
and Lemma~\ref{lem:kzero-inverse} imply
\[
 R(r)\le e^{r-r'}(1+h(r'))\le e^{r-r'}(1+h(r)).
\]
If the shift extends to the left of zero, $R(r)\le e^r\le\upsilon_0$.
The baseline inverse is at least $1-e^{-\chi(r')}$, which also proves the
lower bound on $\rho_*$.

For the second assertion,
$\frac{\dd}{\dd v}B_0(t_v)=h(v)e^{-v}$.  A completed block of mass
$M_g$ contributes to $h(v)$ for all $v\ge b_{g+1}$.  Hence
\[
 e^r\{B_0(t_r)-B_0(t_{r-\log\upsilon})\}
 \ge\sum_{g:\,b_{g+1}\le r}M_g
       \min\{\upsilon-1,e^{r-b_{g+1}}-1\}.
\]
Condition~\eqref{eq:cert-prefix-budget}, together with
$B(t_r)=B_0(t_r)+e^{-r}d(r)$, gives $B(t_r)\ge B_0(t_{r-\log\upsilon})$.
The same inverse comparison applies.
\end{proof}

On a normalized scale chart, $M_g=x_{\max}m_g$ and
$d=x_{\max}\hat d$.  Use lower bounds $\underline m_g\le m_g$ and
$\underline{\hat d}(r)\le\hat d(r)$.  For every $x_{\max}>0$, the common scale can be
cancelled from the completed-prefix comparison.  Thus its sufficient
budget condition becomes
\[
 \sum_{g:\,b_{g+1}\le r}\underline m_g
 \min\{\upsilon-1,e^{r-b_{g+1}}-1\}
 \ge-\underline{\hat d}(r).
\]
This form is uniform over positive scales, including those in a chart
whose lower scale endpoint is zero.

\subsubsection{Finite-scale capacity bound}
Use $M_g=x_{\max}m_g$ and $d=x_{\max}\hat d$, with
$0\le\underline m_g\le m_g\le\overline m_g$ and
$\underline{\hat d}(r)\le\hat d(r)=\sum_g c_g(r)m_g$.  Let
$0<x_{\max}\le\overline{x}_{\max}$ and
$0\le\underline{x}_{\max}\le x_{\max}$, and let
$0\le h_L\le h(r)\le h_U$, where
$r\in[b_p,b_{p+1}]$ and $p\in\{0,1,2,3\}$.
The lower enclosure endpoint $\underline{x}_{\max}$ may be zero,
but the actual scale $x_{\max}$ is positive.  In these chart coordinates, the upper
cumulative masses of~\eqref{eq:cert-mass-ranges} are replaced by the enclosures
\[
 \overline M_{<g}=\overline{x}_{\max}\sum_{j<g}\overline m_j,
\]
and define
\begin{align*}
 \bar\pi_g&=
 \begin{cases}
 e^{-\underline{x}_{\max}\sum_{j<g}\underline m_j}
 \dfrac{1-e^{-\underline{x}_{\max}\overline m_g}}{\underline{x}_{\max}},
                   &\underline{x}_{\max}>0,\\[2mm]
 \overline m_g,&\underline{x}_{\max}=0,
 \end{cases}\\
 U_g(r)&=\begin{cases}
 e^{h_U},&g\le p,\\e^{r-b_g+\overline M_{<g}},&g>p,
 \end{cases}
 \qquad
 \Theta(r)=\underline{\hat d}(r)
                   +\frac{e^{h_L}-1-h_L}{\overline{x}_{\max}}.
\end{align*}
When $\Theta(r)\ge0$,
\begin{equation}\label{eq:cert-capacity}
 R(r)\le\max\left\{e^{r-b_4+\overline M_{<4}},
 \max_{\substack{\varnothing\ne \mathcal S\subseteq\{0,1,2,3\}\\
                  \sum_{g\in \mathcal S}\bar\pi_g>0}}
 \frac{\sum_{g\in \mathcal S}\bar\pi_gU_g(r)-\Theta(r)}{\sum_{g\in \mathcal S}\bar\pi_g}
 \right\}.
\end{equation}
To prove~\eqref{eq:cert-capacity}, note first that $\Theta(r)\le\Delta_\lambda(r)/x_{\max}$, so
the nonnegative-gap branch applies and $\rho_*(r)\ge r$. If $\rho_*(r)\ge b_4$, then
$Y(r)\ge1-e^{-\chi(b_4)}$, which gives the first term:
$R(r)\le e^{r-b_4+h(b_4)}\le e^{r-b_4+\overline M_{<4}}$. Suppose
now that $\rho_*(r)<b_4$. Parametrize threshold levels by their rescaled complement
$z=e^r(1-x)$. Since $A_\lambda(Y(r))=B(t_r)$, $1-Y(r)=e^{-r}R(r)$, and
$1-x_\lambda(u_r)=e^{-r}e^{h(r)}$,
\[
 \Delta_\lambda(r)=e^r\bigl\{A_\lambda(Y(r))-A_\lambda(x_\lambda(u_r))\bigr\}
 =\int_{R(r)}^{e^{h(r)}}A'_{\lambda,+}\bigl(1-e^{-r}z\bigr)\dd z.
\]
By~\eqref{eq:rank-area-expectation}, $A'_{\lambda,+}(x)=\Pp(X_\lambda\le x)=1-e^{-h(\rho_\lambda(x))}$
whenever $0\le x<1$ and $\rho_\lambda(x)<\infty$. Put
\[
 \pi_g:=e^{-h(b_g)}-e^{-h(b_{g+1})}=e^{-x_{\max}\sum_{j<g}m_j}\bigl(1-e^{-x_{\max}m_g}\bigr),
 \qquad \pi_g/x_{\max}\le\bar\pi_g,
\]
where the inequality follows from
$(1-e^{-x_{\max} m})/x_{\max}=\int_0^m e^{-x_{\max} v}\dd v$, which is nonincreasing in
$x_{\max}$ and nondecreasing in $m$. We also write $\pi_{<g}:=\sum_{j<g}\pi_j=1-e^{-h(b_g)}$.
If $b_g<\rho_\lambda(x)\le b_{g+1}$, then $1-e^{-h(\rho_\lambda(x))}\le1-e^{-h(b_{g+1})}
=\pi_{<g+1}$; hence $A'_{\lambda,+}(x)\le\sum_g\pi_g\ind\{\rho_\lambda(x)>b_g\}$
for the levels in the integral. Moreover, $\rho_\lambda(x)>b_g$ forces
$x\ge1-e^{-\chi(b_g)}$, that is, $z\le e^{r-b_g+h(b_g)}\le U_g(r)$ when $g>p$; when $g\le p$,
every $z$ in the range of integration satisfies $z\le e^{h(r)}\le U_g(r)$. Consequently,
\[
 \Theta(r)\le\frac{\Delta_\lambda(r)}{x_{\max}}
 \le\sum_g\frac{\pi_g}{x_{\max}}\bigl(U_g(r)-R(r)\bigr)_+
 \le\sum_g\bar\pi_g\bigl(U_g(r)-R(r)\bigr)_+.
\]
The right side is the maximum, over all subsets $\mathcal S\subseteq\{0,1,2,3\}$, of the affine
expressions $\sum_{g\in\mathcal S}\bar\pi_g(U_g(r)-R(r))$. For a positive target $\Theta(r)>0$,
a nonempty subset with $\sum_{g\in\mathcal S}\bar\pi_g>0$ attains the maximum, and solving for
$R(r)$ gives~\eqref{eq:cert-capacity}. All subsets are retained, so no ordering of the
enlarged $U_g$ is needed.
For a zero target with a positive actual area gap, keep that gap when solving the
inequalities and then weaken the resulting bound to the zero-target value. If the actual gap
is zero and there is positive past mass, the root for $\mathcal S=\{0,\ldots,p\}$ equals
$e^{h_U}\ge e^{h(r)}=R(r)$. Otherwise $h(r)=0$ and $B(t_r)=0$; the upper inverse skips the
initial flat interval of $A_\lambda$ and reaches at least the first block that can carry
positive mass, which gives the same bound. Finally, replacing $e^{h_L}-1-h_L$ by
$h_L^2/2+h_L^3/6$ weakens the bound in the valid direction.

\subsubsection{Future-endpoint bounds}
Throughout this paragraph, choose state endpoints satisfying
$0\le h_L\le h(r)\le h_U$.
Fix $g\in\{1,2,3,4\}$ with $r\le b_g$ and use the cumulative
mass bound $h(b_g)\le\overline M_{<g}$.
On the nonnegative-gap branch, the two possibilities
$\rho_*\ge b_g$ and $r\le\rho_*\le b_g$ give, respectively,
\[
 R\le e^{r-b_g+\overline M_{<g}},\qquad
 \Delta_\lambda\le(1-e^{-\overline M_{<g}})(e^{h(r)}-R);
\]
the second inequality follows from the integral representation of $\Delta_\lambda$ in the
proof of~\eqref{eq:cert-capacity}, because every level in the range of integration has
$\rho_\lambda\le\rho_*\le b_g$ and hence slope at most $1-e^{-h(b_g)}$.
Thus, whenever $\underline\Delta\ge0$ is a lower bound on
$\Delta_\lambda$ and $\overline M_{<g}>0$,
\begin{equation}\label{eq:cert-future-endpoint}
 R\le\max\left\{e^{h_U}-\frac{\underline\Delta}{1-e^{-\overline M_{<g}}},
                       e^{r-b_g+\overline M_{<g}}\right\}.
\end{equation}
One may use
$\underline\Delta=e^{h_L}-1-h_L+\underline d$, provided it is nonnegative.
A scale-uniform form replaces the first term by $e^{h_U}-\Theta/\bar\pi_{<g}$,
where the upper bound $\bar\pi_{<g}\ge\pi_{<g}/x_{\max}$ is
\[
 \bar\pi_{<g}=\begin{cases}
 \dfrac{1-e^{-\underline{x}_{\max}\sum_{j<g}\overline m_j}}
       {\underline{x}_{\max}},&\underline{x}_{\max}>0,\\[2mm]
 \sum_{j<g}\overline m_j,&\underline{x}_{\max}=0.
 \end{cases}
\]
This version requires $\Theta\ge0$ and $\bar\pi_{<g}>0$.
It follows by dividing the area inequality by $x_{\max}$ and using the
same integral monotonicity as in the capacity bound: $\pi_{<g}/x_{\max}=(1-e^{-h(b_g)})/x_{\max}\le\bar\pi_{<g}$.

\subsubsection{The negative-gap branch}
When $\rho_*(r)<r$, put $h_\rho:=h(\rho_*(r))$.
Every level $x\in[Y(r),x_\lambda(u_r))$ satisfies $\rho_*(r)\le\rho_\lambda(x)\le r$
and hence $A'_{\lambda,+}(x)\ge1-e^{-h_\rho}$. Integrating as in the proof
of~\eqref{eq:cert-capacity} gives
$-\Delta_\lambda(r)\ge(1-e^{-h_\rho})\bigl(R(r)-e^{h(r)}\bigr)$, and $e^{h(r)-h_\rho}\ge1+h(r)-h_\rho$
then yields
\[
 -d(r)\ge R(r)(1-e^{-h_\rho})-h_\rho.
\]
If $0<\underline h_\rho\le h_\rho$ and $\underline d\le d(r)$, then
\begin{equation}\label{eq:cert-negative-inverse}
 R(r)\le\frac{\underline h_\rho-\underline d}{1-e^{-\underline h_\rho}}.
\end{equation}
For fixed $\underline d$, the derivative of $s\mapsto(s-\underline d)/(1-e^{-s})$ has the sign of $e^s-1-s+\underline d$.
It is negative for $0<s\le h_\rho$, because
$h_\rho\le h(r)$ and $\Delta_\lambda(r)<0$.
Replacing $h_\rho$ by its positive lower bound therefore increases the
quotient.

An explicit bound $\underline h_\rho$ follows from Lemma~\ref{lem:cert-shift}.
On a cell $r\in[a,b]$, $h\in[h_L,h_U]$, let $\upsilon$ be a uniform valid
shift factor and put
\begin{equation}\label{eq:cert-Htau-lower}
 a_0=(a-(\upsilon-1))_+,\qquad
 \underline h_\rho=\max\left\{0,h_L-(\upsilon-1),
       \max_{0\le g\le4}[\underline M_{<g}-(b_g-a_0)_+]\right\}.
\end{equation}
Indeed $\rho_*\ge(r-\log\upsilon)_+{}\ge a_0$,
$\log\upsilon\le\upsilon-1$, and $h$ is nondecreasing and $1$-Lipschitz.
If $\underline h_\rho=0$, the shift estimate is used without division.

\subsubsection{A continuous negative-gap bound}
Suppose $\rho_*(r)<r$, and let $\upsilon\ge1$ be a valid shift factor
with $\rho_*(r)\ge(r-\log\upsilon)_+$.
If $e^{h(r)}/\upsilon-1>0$, then
\begin{equation}\label{eq:cert-negative-continuous}
 R(r)\le1+h(r)-d(r)
 +\frac{-\Delta_\lambda(r)}{e^{h(r)}/\upsilon-1}.
\end{equation}
Indeed, the identity~\eqref{eq:cert-source-identity} gives
\[
 -\Delta_\lambda(r)
 =e^r\int_{\rho_*(r)}^r
 (1-\lambda(u_a))e^{-a}(e^{h(a)}-1)\dd a.
\]
Monotonicity and the $1$-Lipschitz property of $h$ imply
$h(a)\ge h(\rho_*(r))\ge h(r)-\log\upsilon$ on this interval.
Consequently,
\[
 -\Delta_\lambda(r)
 \ge(e^{h(r)}/\upsilon-1)(-\Xi(r)).
\]
Combining this inequality with
$R=1+h-d-\Xi$ proves~\eqref{eq:cert-negative-continuous}.

For a whole cell, choose fixed bounds
$0\le h_L\le h(r)\le h_U$, $\underline d\le d(r)$, and
$\underline E\le e^{h_L}-1-h_L$, together with one uniform valid shift
factor $\upsilon$.  If
\[
 \beta_{\mathrm{neg}}:=\frac{\underline E+1+h_L}{\upsilon}-1>0,
\]
then the negative-gap branch satisfies the constant upper bound
\begin{equation}\label{eq:cert-negative-continuous-cell}
 R(r)\le
 R_{\mathrm{neg}}:=1+h_U-\underline d
 +\frac{[-(\underline E+\underline d)]_+}{\beta_{\mathrm{neg}}}.
\end{equation}
Indeed, $h_L\ge0$ gives
$\Delta_\lambda(r)\ge\underline E+\underline d$, while
$e^{h(r)}/\upsilon-1\ge\beta_{\mathrm{neg}}>0$.
All bounds and the positive-denominator test are certified over the
same whole cell.  If the denominator cannot be certified positive,
this candidate is omitted.

For direct integration, bounds for the two possible branches are combined
by a maximum.  For the correction term, the required quantity is
\begin{equation}\label{eq:cert-credit-compatible}
 R+\Xi_+=\max\{R,1+h-d\}.
\end{equation}
It equals $1+h-d$ on the nonnegative-gap branch and $R$ on the negative
branch.  This gives precisely the envelope in~\eqref{eq:cert-scheme-IV-envelope}.

\subsubsection{Convexity on an output cell}
\label{app:cert-convex-bounds}

On a fixed output cell use $X=e^r$.
Each $c_g$ is affine in $X$, by~\eqref{eq:cert-c-coeff}.
For constant mass bounds,
\begin{equation}\label{eq:cert-concave-d}
 \underline{\hat d}(X)=
 \sum_g\min\{\underline m_g c_g(X),\overline m_g c_g(X)\}
\end{equation}
is concave.  Affine lower bounds on the coefficients may be used in the
same formula.
Nonnegativity at both endpoints therefore proves nonnegativity throughout
the cell.  The roots in~\eqref{eq:cert-capacity} are convex in $X$: their denominators are positive
constants, their node terms are affine in $X$, and the negative of
\eqref{eq:cert-concave-d} is convex.  Their maxima remain convex.
The unconditional shift bound of Lemma~\ref{lem:cert-shift}, with $h$
replaced by a fixed upper state endpoint, is convex in $X$ as well,
since each $c_g$ is affine in $X$ and
$\upsilon_0=1+\max_g(-c_g)_+$.
The future-endpoint and quotient bounds have the same property when their
state endpoints, denominator bounds, and branch conditions are fixed on
the cell.

The prefix budget provides another convex bound when sufficient completed
mass is available.  In~\eqref{eq:cert-prefix-budget}, its left side is
concave jointly in $(\upsilon,X)$ and nondecreasing in $\upsilon$.
Adding the concave lower bound on $d$ preserves concavity.
The superlevel set is consequently convex, and its lower boundary, the
least sufficient shift factor, is convex in $X$.
The condition that the total budget suffices is checked throughout the
cell; the root is evaluated by the finitely many linear pieces of the
budget.  If this condition fails, the global shift bound remains available.

\subsection{Lower Bounds for the Correction Term}
\label{app:cert-credit}

\subsubsection{Changing the order of integration}

For an input coordinate $a\ge0$, define
$
 \varrho_a:=
 \inf\{r\ge0:
 B(t_r)\ge A_\lambda(1-e^{-\chi(a)})\}$,
with $\inf\varnothing=\infty$.
By the upper-inverse conventions and the continuity of $B$, the condition
$a\le\rho_*(r)$ is equivalent to $r\ge \varrho_a$.
Using the representation of $\Xi_+$ in
\eqref{eq:cert-inverse-identity} and changing the order of integration,
\begin{align}
 &\int_0^\infty
 \sigma(r)\frac{(r-k(r))\Xi_+(r)}{D(r)^2}\dd r =
 \int_{\varrho_a\le a}(1-\lambda(u_a))e^{-a}
 \int_{\varrho_a}^{a}
 \sigma(r)\frac{(r-k(r))e^r}{D(r)^2}\dd r\,\dd a.
 \label{eq:cert-tonelli-main}
\end{align}
Consequently, restricting the prefix to $r\le b_4$ and retaining only
parameters $a\le b_4$ gives
\[
 \int_0^{b_4}
 \sigma(r)\frac{(r-k(r))\Xi_+(r)}{D(r)^2}\dd r
 \ge
 \int_0^{b_4}
 (1-\lambda(u_a))\,\mathcal C(a)\dd a,
\]
where
\[
 \mathcal C(a):=
 \begin{cases}
 e^{-a}\displaystyle\int_{\varrho_a}^{a}
 \sigma(r)\frac{(r-k(r))e^r}{D(r)^2}\dd r,
 &\varrho_a\le a,\\[2mm]
 0,&\varrho_a>a.
 \end{cases}
\]
To make this correction term compatible with the regional Bellman comparison,
we lower-bound $\mathcal C(a)$ by a function of the current regional
state.  Choose $\mathcal G\ge0$ so that
\begin{equation}\label{eq:cert-credit-condition}
 \sigma(a)\mathcal G(a,\mathbf x(\lambda),h(a))
 \le \mathcal C(a),
\end{equation}
for every admissible profile.  The window estimate below
constructs such a bound explicitly; if no positive bound is certified,
we take $\mathcal G=0$.

Multiplying~\eqref{eq:cert-credit-condition} by $1-\lambda(u_a)$
and integrating gives~\eqref{eq:cert-credit-integrated}.

\subsubsection{A lower bound from a past window}

Fix $a\le b_4$.  Suppose that
\[
 \underline d\le d(a),\qquad
 \zeta_+\ge \max\{0,\sup_{0\le r\le a}\nu(r)\},\qquad
 \zeta_-\ge \max\{0,\sup_{r\ge0}(-\nu(r))\}.
\]
Set
\[
 \underline\Delta
 :=e^{h(a)}-1-h(a)+\underline d,
 \qquad
 \bar k:=h(a)+\zeta_+,
 \qquad
 \chi_*:=(a-\bar k)_+.
\]
Thus $\bar k$ bounds the rate $k$ on $[0,a]$, and $\chi_*=(\chi(a)-\zeta_+)_+$.
If $\underline\Delta,\bar k,\chi_*$ are all positive, choose a window length
\begin{equation}\label{eq:cert-credit-window}
 0\le\eta\le\min\left\{\log\left(1+\frac{\underline\Delta}{\bar k}\right),\chi_*\right\};
\end{equation}
otherwise set $\eta=0$ and use a zero correction term.
For $v\le a$, one has $k(v)\le \bar k$, and hence
\[
 B(t_a)-B(t_{a-\eta})\le \bar ke^{-a}(e^{\eta}-1)
 \le e^{-a}\underline\Delta
 \le B(t_a)-A_\lambda(1-e^{-\chi(a)}).
\]
Therefore $\varrho_a\le a-\eta$.  Moreover, for $0\le v\le\eta$,
$(a-v)-k(a-v)\ge \chi_*-v\ge0$.

Let $\psi(r)=e^r/D(r)^2$.  Direct differentiation gives
\begin{equation}\label{eq:cert-phi-derivative}
 -\frac{\dd}{\dd r}\log\psi(r)=1-\frac{2r}{D(r)}.
\end{equation}
Choose a nonnegative lower bound $\beta$ for the right-hand side of
\eqref{eq:cert-phi-derivative} on the whole interval
$[a-\eta,a]$.  Let $h_L\le h(a)$, $D_U\ge D(a)$, and use
$h^-_{\mathcal R}$ as a nondecreasing lower bound on the past state.  Set
\[
 \beta_{\mathrm{surv}}=\frac1{D_U}\max\{0,h_L-\eta/2-\zeta_-,
                         h^-_{\mathcal R}(a-\eta)-\zeta_-\},
 \qquad \beta_{\mathrm{eff}}=\beta+\beta_{\mathrm{surv}}.
\]
For $0\le v\le\eta$, integrating $h(a-w)\ge h_L-w$ and the state lower
bound gives
\[
 \int_{a-v}^{a}\frac{k(w)}{D(w)}\dd w\ge v\beta_{\mathrm{surv}}.
\]
Together with~\eqref{eq:cert-phi-derivative}, this proves
\[
 \sigma(a-v)\psi(a-v)\ge \sigma(a)\psi(a)e^{\beta_{\mathrm{eff}}v}.
\]
Since $\mathcal C(a)\ge e^{-a}\int_0^{\eta}\sigma(a-v)\psi(a-v)\bigl((a-v)-k(a-v)\bigr)\dd v$,
keeping the first seven nonnegative terms of the exponential series yields
\begin{equation}\label{eq:cert-tonelli-credit}
 \mathcal G_L=\frac1{D_U^2}\sum_{j=0}^6\frac{\beta_{\mathrm{eff}}^j}{j!}
 \left[\frac{\chi_*\eta^{j+1}}{j+1}-\frac{\eta^{j+2}}{j+2}\right].
\end{equation}
Every bracket is nonnegative since $\eta\le\chi_*$.
Integration over $[a-\eta,a]$ therefore proves
\eqref{eq:cert-credit-condition} with $\mathcal G=\mathcal G_L$.  Positivity of the window guarantees $\varrho_a\le a$; otherwise
the correction term is zero.  Taking lower bounds for this expression throughout
a parameter--coordinate--state cell therefore gives a valid $\mathcal G$.
The dependence of $\beta_{\mathrm{surv}}$ on $\eta$ is retained in the interval evaluation.

\subsubsection{Cellwise instantiation}
On a regional mass box
$\underline M_g\le M_g\le\overline M_g$, we use
\begin{equation}\label{eq:cert-credit-d-lower}
\underline d(a)
:=
\sum_{g=0}^3
\min\{c_g(a)\underline M_g,
c_g(a)\overline M_g\}
\le d(a).
\end{equation}
Together with the corresponding interval bounds for $\nu$ and $-\nu$,
this supplies the quantities required in
\eqref{eq:cert-tonelli-credit}. When the window conditions above are not certified on the whole cell,
we use $\mathcal G=0$.

The implementation uses, as a uniform starting value,
\[
 \beta_0=\frac{4584664420663165}{18014398509481984},
\]
the exact binary64 value represented by the decimal $0.2545$.
Here is a rational check that it is below~\eqref{eq:cert-phi-derivative} for every $r\ge0$.
Set $\theta_0=2/(1-\beta_0)$ and $r^\sharp=\theta_0/(\theta_0-1)$.
The desired inequality is equivalent to
$f(r):=c_{\Gamma_1} e^r+1-(\theta_0-1)r\ge0$, and $f$ is convex.
The exact comparison
\[
 c_{\Gamma_1}\sum_{j=0}^{14}\frac{(r^\sharp)^j}{j!}>\theta_0-1
\]
places the possible interior minimizer
$r_{\min}=\log((\theta_0-1)/c_{\Gamma_1})$ below $r^\sharp$.
Its value is $f(r_{\min})=\theta_0-(\theta_0-1)r_{\min}>0$, while $f(0)>0$.
Local bounds over a shorter window may improve $\beta_0$.
For example, the derivative of $r/D(r)$ has the sign of
$1+c_{\Gamma_1} e^r(1-r)$, with its unique zero between $1.5$ and $1.69$.
The endpoint signs follow from rational exponential bounds.

\subsection{Quadrature and Finite Comparison Updates}
\label{app:cert-direct-updates}

\subsubsection{One-dimensional integration}
\label{app:cert-quadrature}

\begin{lemma}[Integration of a convex chord]\label{lem:cert-chord}
Let $0\le a<b$.  Suppose $f$ is convex on $[e^a,e^b]$ and
$R(r)\le f(e^r)$ for $r\in[a,b]$, uniformly over the stated profile range,
and $U_a\ge f(e^a)$, $U_b\ge f(e^b)$.
Let $\overline R(r)$ be the chord in $e^r$ with these endpoints.
Then $R(r)\le\overline R(r)$ and
\begin{equation}\label{eq:cert-chord-integral}
 \int_a^b\frac{r\overline R(r)}{D(r)^2}\dd r
 = U_a\mathcal M_0+(U_b-U_a)\frac{\mathcal M_1-e^a\mathcal M_0}{e^b-e^a},
\end{equation}
where
\begin{equation}\label{eq:cert-chord-moments}
 \mathcal M_0=\int_a^b\frac r{D(r)^2}\dd r,\qquad
 \mathcal M_1=\int_a^b\frac{re^r}{D(r)^2}\dd r
 =\frac{(a+1)/D(a)-(b+1)/D(b)}{c_{\Gamma_1}}.
\end{equation}
\end{lemma}
\begin{proof}
Convexity gives the chord bound.  Substitute
$\overline R(r)=U_a+(U_b-U_a)(e^r-e^a)/(e^b-e^a)$ and integrate.
The expression for $\mathcal M_1$ follows by differentiating $(r+1)/D(r)$, whose derivative is $-c_{\Gamma_1}re^r/D(r)^2$.
\end{proof}

Different bounds may be integrated separately and their integral upper
endpoints minimized.  When a single chord is used, its two endpoints must
bound the same convex function.  This also applies after combining possible
inverse branches by a maximum.

For $f_0(r)=r/D(r)^2$,
\begin{equation}\label{eq:cert-quadrature-error}
 |f_0''(r)|\le\frac{4+8r}{D(r)^2}
              \le\frac{4+8b}{D(a)^2}\qquad(a\le r\le b).
\end{equation}
Indeed,
$f_0''=-4D'/D^3-2rD''/D^3+6r(D')^2/D^4$ and
$0\le D'/D,D''/D\le1$.
A composite trapezoidal sum with maximal substep $\delta_r$ therefore has
error at most
\[
 \frac{(b-a)\delta_r^2(4+8b)}{12D(a)^2}.
\]
Subtracting and adding this quantity supplies lower and upper endpoints
for $\mathcal M_0$ and, by summation, for $T_0$.

\subsubsection{State-band chord bound}

\begin{lemma}[State-band chord bound]
\label{lem:cert-dp-chord}
Fix $j$ with $0\le r_j<r_{j+1}$, and let $v_{j+1}\ge0$
satisfy $V_\lambda(r_{j+1})\le v_{j+1}$ uniformly over the
admissible profiles.  Partition the admissible state range on
$[r_j,r_{j+1}]$ into finitely many bands indexed by $i$.
For each retained band, let $0\le k_{j,i}\le r_{j+1}$.
Given endpoint values $\overline R_{j,i}(r_j)$ and
$\overline R_{j,i}(r_{j+1})$, define
\[
 \overline R_{j,i}(r)
 =\frac{e^{r_{j+1}}-e^r}{e^{r_{j+1}}-e^{r_j}}
   \overline R_{j,i}(r_j)
 +\frac{e^r-e^{r_j}}{e^{r_{j+1}}-e^{r_j}}
   \overline R_{j,i}(r_{j+1}),
 \qquad r\in[r_j,r_{j+1}].
\]
Assume that
\[
 k_{j,i}\le k(r),\qquad R(r)\le\overline R_{j,i}(r)
\]
whenever an admissible profile lies in band $i$ at coordinate $r$.

Define
\begin{equation}\label{eq:cert-generic-dp-chord}
 Z_{j,i}(r)
 =
 \left(1-\frac{k_{j,i}}{r_{j+1}}\right)
 \overline R_{j,i}(r)
 -
 \frac{k_{j,i}v_{j+1}D(r_j)}{r_{j+1}}.
\end{equation}
Set
\[
 Z_j(r_j)=\max_iZ_{j,i}(r_j),
 \qquad
 Z_j(r_{j+1})=\max_iZ_{j,i}(r_{j+1}),
\]
and, for $r\in[r_j,r_{j+1}]$, define
\[
 Z_j(r)
 =
 Z_j(r_j)
 +
 \frac{Z_j(r_{j+1})-Z_j(r_j)}
      {e^{r_{j+1}}-e^{r_j}}
 \bigl(e^r-e^{r_j}\bigr).
\]
Then any $v_j$ satisfying
\begin{align}
 v_j\ge\max\Biggl\{0,\;&
 v_{j+1}
 +\int_{r_j}^{r_{j+1}}
   \frac{rZ_j(r)}{D(r)^2}\dd r \notag\\
 &+
 \min\left\{
 1,\frac{(r_{j+1}-r_j)r_{j+1}}{D(r_j)}
 \right\}
 \bigl(-\min\{Z_j(r_j),Z_j(r_{j+1})\}\bigr)_+
 \int_{r_j}^{r_{j+1}}\frac r{D(r)^2}\dd r
 \Biggr\}
\label{eq:cert-generic-dp-update}
\end{align}
satisfies
\[
 V_\lambda(r_j)\le v_j
\]
uniformly over the admissible profiles.
\end{lemma}

\begin{proof}
On state band $i$,
\[
 \frac{(r-k)R}{D^2}
 -\frac{kv_{j+1}}{D}
 \le
 \frac r{D^2}Z_{j,i}(r),
\]
because $k_{j,i}\le k(r)\le r\le r_{j+1}$,
$R(r)\le\overline R_{j,i}(r)$, and
$D(r)\ge D(r_j)$.

By~\eqref{eq:cert-generic-dp-chord}, for every
$r\in[r_j,r_{j+1}]$,
\begin{align*}
 Z_{j,i}(r)
 &=\frac{e^{r_{j+1}}-e^r}{e^{r_{j+1}}-e^{r_j}}Z_{j,i}(r_j)
 +\frac{e^r-e^{r_j}}{e^{r_{j+1}}-e^{r_j}}Z_{j,i}(r_{j+1})\\
 &\le\frac{e^{r_{j+1}}-e^r}{e^{r_{j+1}}-e^{r_j}}Z_j(r_j)
 +\frac{e^r-e^{r_j}}{e^{r_{j+1}}-e^{r_j}}Z_j(r_{j+1})\\
 &=Z_j(r).
\end{align*}
The inequality follows from the endpoint maxima and the nonnegative
interpolation weights.
Hence
\[
 \frac{(r-k)R}{D^2}
 -\frac{kv_{j+1}}{D}
 \le
 \frac r{D^2}Z_j(r).
\]
Moreover,
\[
 0\le1-\sigma_{r_j}(r)
 \le
 \min\left\{
 1,\frac{(r_{j+1}-r_j)r_{j+1}}{D(r_j)}
 \right\}.
\]
Since $Z_j$ is linear in $e^r$, its minimum on
$[r_j,r_{j+1}]$ is attained at an endpoint.  Therefore
\begin{align*}
 \sigma_{r_j}(r)Z_j(r)
 \le{}&
 Z_j(r)\\
 &+
 \min\left\{
 1,\frac{(r_{j+1}-r_j)r_{j+1}}{D(r_j)}
 \right\}
 \bigl(-\min\{Z_j(r_j),Z_j(r_{j+1})\}\bigr)_+ .
\end{align*}
Substitution into~\eqref{eq:cert-direct-one-step} proves
\eqref{eq:cert-generic-dp-update}.
\end{proof}

\subsubsection{A finite-scale forward comparison}

For some regions the full backward state-band comparison is unnecessary.
On a partition $0=r_0<r_1<\cdots<r_N=b_4$,
let $k_j\ge0$ be a lower bound for $k$ on $[r_j,r_{j+1}]$, and let
$\mathcal Z_j$ bound $\int_{r_j}^{r_{j+1}}rR(r)/D(r)^2\,\dd r$.  Define
$\overline \sigma_0=1$ and
\[
 \overline \sigma_{j+1}
 =
 \frac{\overline \sigma_j}
 {1+(r_{j+1}-r_j)k_j/D(r_{j+1})}.
\]
Since $e^{-x}\le1/(1+x)$, we have
$\sigma(r_j)\le\overline \sigma_j$ and
$(r-k(r))/r\le1-k_j/r_{j+1}$ on $[r_j,r_{j+1}]$; moreover $\sigma(r)\le\sigma(r_j)$ there.  Consequently,
\begin{equation}\label{eq:cert-forward-bound}
 \overline\ell_\lambda\le\Gamma_1+
 \sum_{j:r_{j+1}\le b_4}
 \overline \sigma_j
 \left(1-\frac{k_j}{r_{j+1}}\right)\mathcal Z_j
 +\overline \sigma_N\overline V.
\end{equation}

\subsubsection{Finite Bellman update}

\begin{lemma}[Finite Bellman update]
\label{lem:cert-bellman-update}
Let
\[
 0=r_0<r_1<\cdots<r_N=b_4,
 \qquad
 \delta_j:=r_{j+1}-r_j,
 \qquad
 \kappa_j:=\frac{r_{j+1}}{D(r_j)}.
\]
Suppose that, at a fixed right-end value $W_{j+1}$,
$H_j$ is a uniform upper bound for the right-hand side of
\eqref{eq:cert-HJB} throughout
$[r_j,r_{j+1}]$ and over all admissible parameter--state values.

Choose
\begin{equation}\label{eq:cert-potential-update}
 W_j\ge W_{j+1}+
 \begin{cases}
  \delta_j H_j,
      &H_j\ge0,\\[2mm]
  \displaystyle
  \frac{\delta_jH_j}{1+\delta_j\kappa_j},
      &H_j<0.
 \end{cases}
\end{equation}
Then the affine interpolation of $W$ between $W_j$ and $W_{j+1}$
satisfies~\eqref{eq:cert-HJB} throughout the interval.

If, in addition,
\begin{equation}\label{eq:cert-W-terminal}
 W_N\ge
 \overline V-T_0(b_4)+\sup_hP(b_4,h),
\end{equation}
where the supremum is over
$h\in[h^-_{\mathcal R}(b_4),h^+_{\mathcal R}(b_4)]$, then repeated
application of~\eqref{eq:cert-potential-update} gives
\[
 \overline\ell_\lambda
 \le
 \Gamma_1+T_0(0)+W_0-P(0,0)
\]
for every profile $\lambda\in\Lambda(\mathcal R)$.
\end{lemma}

\begin{proof}
For each fixed admissible parameter--state value, the expression
inside the supremum in~\eqref{eq:cert-HJB} is affine in $W$,
with slope $-k/D$.  Throughout the interval,
\[
 0\le\frac{k(r)}{D(r)}\le\frac{r_{j+1}}{D(r_j)}=\kappa_j,
 \qquad -W'(r)=\frac{W_j-W_{j+1}}{\delta_j}.
\]
First suppose that $W_j\ge W_{j+1}$.
Then $W(r)\ge W_{j+1}$, so the residual is at most $H_j$.
If $H_j\ge0$, the update condition gives $-W'(r)\ge H_j$.
If $H_j<0$, we have $-W'(r)\ge0>H_j$.
Thus~\eqref{eq:cert-HJB} holds in either case.

Now suppose that $W_j<W_{j+1}$.  The update condition then
forces $H_j<0$.  Since $W_j\le W(r)\le W_{j+1}$, the residual
is at most $H_j+\kappa_j(W_{j+1}-W_j)$.
Equation~\eqref{eq:cert-potential-update} implies
\[
 (1+\delta_j\kappa_j)(W_j-W_{j+1})\ge\delta_jH_j,
\]
and hence
\[
 H_j+\kappa_j(W_{j+1}-W_j)
 \le\frac{W_j-W_{j+1}}{\delta_j}=-W'(r).
\]
Thus~\eqref{eq:cert-HJB} holds throughout the interval.

Finally,~\eqref{eq:cert-W-terminal} implies
\[
 T_0(b_4)+W_N-P(b_4,h)\ge\overline V
\]
for every admissible terminal state.
Proposition~\ref{prop:cert-bellman} gives the stated load bound.
\end{proof}

\subsubsection{Finite-tail construction of \texorpdfstring{$\overline V$}{V}}

We now complete the construction of the continuation bound at $b_4$
used in Appendix~\ref{app:cert-tail-bound}.  Equation~\eqref{eq:cert-tail}
provides the terminal bound at $r=10$.  Partition $[b_4,10]$ into
finitely many intervals of the exponential coordinate and process them from right to left.

Consider one such interval $[a,b]$ and suppose that $v\ge0$ is a uniform
upper bound for the continuation at its right endpoint $b$.  Using
\eqref{eq:cert-tail-state}, restrict $h$ to its admissible state interval
and let $\underline d\le d(r)$ and $\underline \nu\le \nu(r)$ be certified
lower bounds throughout $[a,b]$.  Set
$\underline k(h):=(h+\underline \nu)_+$, so that $k(r)\ge \underline k(h)$.

On the nonnegative-gap branch, $R(r)\le1+h-\underline d$, and hence
\begin{equation}\label{eq:cert-tail-positive-residual}
 \frac{(r-k)R}{D(r)^2}-\frac{kv}{D(r)}
 \le
 \frac{(r-\underline k(h))(1+h-\underline d)}{D(r)^2}
 -\frac{\underline k(h)v}{D(r)}.
\end{equation}
On the negative-gap branch,
\[
 \frac{h^2}{2}\le e^h-1-h<-d(r)\le-\underline d,
\]
so $h\le\sqrt{-2\underline d}$.  If $\upsilon$ is a uniform upper bound
for $\upsilon_0(r)$ on $[a,b]$, then
$R(r)\le\overline R_0(r,h)\le\upsilon(1+h)$, and therefore
\begin{equation}\label{eq:cert-tail-negative-residual}
 \frac{(r-k)R}{D(r)^2}-\frac{kv}{D(r)}
 \le
 \frac{(r-\underline k(h))\upsilon(1+h)}{D(r)^2}
 -\frac{\underline k(h)v}{D(r)}.
\end{equation}

For fixed interval coefficients, the two right-hand sides are affine in
$h$ for $h\le-\underline \nu$ and concave quadratic for
$h\ge-\underline \nu$.  Split the admissible state interval at
$h=-\underline \nu$ whenever this point lies in the interval.
Their maxima are controlled by the endpoints of the resulting affine and
quadratic pieces, including that splitting point, and by the quadratic
vertex whenever it lies in the clipped quadratic piece.  Together with interval bounds for the
$r$-dependent coefficients, this gives a rigorous whole-interval
residual bound $H$.

Applying~\eqref{eq:cert-direct-update} with
$[r_j,r_{j+1}]=[a,b]$ and $\kappa_j=b/D(a)$ propagates the continuation bound from $b$ to $a$.
Starting from the analytic bound
$
 V_\lambda(10)\le\frac{265e^{-20}}{16c_{\Gamma_1}^2}$
in~\eqref{eq:cert-tail} and repeating this update over all intervals from $10$ back to $b_4$
produces the required bound
\[
 \overline V\ge
 \sup_{\lambda\in\Lambda(\mathcal R)}V_\lambda(b_4).
\]

\subsection{Regional Comparison Bounds}
\label{sec:cert-schemes}

We now specify the bounds used in the comparisons of
Appendix~\ref{sec:cert-scheme-overview}.  Fix a parameter region and a
partition $0=r_0<\cdots<r_N=b_4$.  Each comparison starts with the
continuation bound $\overline V$ from
Appendix~\ref{app:cert-direct-updates}.
The shift bound of Lemma~\ref{lem:cert-shift} is always available;
sharper bounds are used only after their conditions hold throughout
the relevant coordinate--state cell.

\subsubsection{State-band comparison}

The record specifies integers $N_r,N_b$: the coordinate step is $1/N_r$,
and each admissible state interval is divided into $N_b$ equal bands.
On a coordinate interval $[r_j,r_{j+1}]$, index these bands by $i$.
On a band
$0\le h_L\le h\le h_U$, a lower bound
$\underline\nu_j\le\nu(r)$ gives
\[
 k_{j,i}:=(h_L+\underline\nu_j)_+\le k(r).
\]
The nonnegative-gap branch is certified by checking
\[
 \underline{\hat d}(r)
 +\frac{h_L^2/2+h_L^3/6}{\overline{x}_{\max}}\ge0
\]
at both endpoints.  The left-hand side is concave in $e^r$, so the
condition holds throughout the cell.  Here $\underline{\hat d}$ bounds
$d/x_{\max}$, whereas $\underline d$ bounds $d$.

The basic set of inverse bounds consists of the unconditional shift bound, the completed-prefix
bound whenever~\eqref{eq:cert-prefix-budget} holds throughout the cell,
and, on certified nonnegative-gap cells, the bounds
$1+h_U-\underline d$ and~\eqref{eq:cert-capacity}.
The refined set also includes the absolute and normalized future-endpoint
bounds~\eqref{eq:cert-future-endpoint} on those cells,
provided their respective nonnegative area-gap lower bounds and
positive denominators are certified throughout the cell.
The refined set further includes a candidate valid on both branches.  From a uniform
shift factor, obtain $\underline h_\rho\le h(\rho_*(r))$ by~\eqref{eq:cert-Htau-lower}.
If $\underline h_\rho>0$, put
\[
 s_3(\underline h_\rho)=\underline h_\rho+\frac{\underline h_\rho^2}{2}+\frac{\underline h_\rho^3}{6},\qquad
 Q_{j,i}^{\mathrm{inv}}(r)
 =(\underline h_\rho-\underline d(r))\left(1+\frac1{s_3(\underline h_\rho)}\right).
\]
On the negative-gap branch, $\underline h_\rho-\underline d>0$ and
$e^{\underline h_\rho}-1\ge s_3(\underline h_\rho)>0$, so~\eqref{eq:cert-negative-inverse}
bounds $R$ by $Q_{j,i}^{\mathrm{inv}}$.
Thus $\max\{1+h_U-\underline d,Q_{j,i}^{\mathrm{inv}}\}$
covers both branches; omit it unless the denominator is certified positive.

Each candidate is convex in $e^r$ for fixed state endpoints and
denominator bounds, by Appendix~\ref{app:cert-convex-bounds}.
For each applicable candidate separately, choose certified endpoint
bounds and form
\begin{equation}\label{eq:cert-scheme-II-R-chord}
 \overline R_{j,i}(r)
 :=\overline R_{j,i}(r_j)
 +\frac{\overline R_{j,i}(r_{j+1})-\overline R_{j,i}(r_j)}
 {e^{r_{j+1}}-e^{r_j}}(e^r-e^{r_j}).
\end{equation}
Select the candidate with the smallest certified chord integral and
use its two endpoints in Lemma~\ref{lem:cert-dp-chord}; endpoints from
different candidates are never combined.
Starting from $v_N=\overline V$ at $r_N=b_4$ and applying that lemma backward gives
\begin{equation}\label{eq:cert-state-band-upper}
 U(\mathcal R)=\Gamma_1+v_0.
\end{equation}
The record specifies the basic or refined set and the grid resolutions.
Scale charts use the same comparison in the coordinates
of~\eqref{eq:cert-normalized-masses}.

\subsubsection{Bellman comparison}

For each region and resolution, fix the potential from
Appendix~\ref{app:cert-potential-construction}.
Apply Lemma~\ref{lem:cert-bellman-update} with terminal condition
\eqref{eq:cert-W-terminal}.  We give a basic whole-cell residual bound,
then two optional refinements that preserve dependence on the block
masses and on the coordinate $r$.

Consider a state cell $i$ over the coordinate interval $[r_j,r_{j+1}]$.
Choose fixed bounds $0\le h_L\le h\le h_U$, $\underline d\le d$,
and $\underline E\le e^{h_L}-1-h_L$ that are valid throughout this cell.
Let $\upsilon\ge1$ be a uniform valid shift factor from
Lemma~\ref{lem:cert-shift}; the bound $\upsilon=e^{r_{j+1}}$ is also
available from $R(r)\le e^r$ and $\rho_*(r)\ge0$.
On the negative-gap branch, form whole-cell upper bounds from
$\upsilon(1+h_U)$, the quotient in~\eqref{eq:cert-negative-inverse}
with a certified positive lower bound $\underline h_\rho\le h(\rho_*(r))$,
and the continuous bound~\eqref{eq:cert-negative-continuous-cell}.
Denote the last two upper bounds by $Q_{j,i}^{\mathrm{neg}}$ and $R_{j,i}^{\mathrm{cont}}$, respectively,
and set
\begin{equation}\label{eq:cert-scheme-IV-negative}
 \overline R_{j,i,-}
 :=\min\{\upsilon(1+h_U),Q_{j,i}^{\mathrm{neg}},R_{j,i}^{\mathrm{cont}}\}.
\end{equation}
The quotient candidate is omitted if a strictly positive denominator
cannot be certified, and $R_{j,i}^{\mathrm{cont}}$ is omitted when its denominator
condition is not certified.  These are whole-cell bounds formed
directly from the inverse inequalities, independently of the chords selected by the state-band comparison.

Using
\[
 R+\Xi_+=
 \begin{cases}
  1+h-d,&\Delta_\lambda\ge0,\\
  R,&\Delta_\lambda<0,
 \end{cases}
\]
we take
\begin{equation}\label{eq:cert-scheme-IV-envelope}
 \mathcal U=
 \begin{cases}
  1+h-d,
     &\text{if $\underline E+\underline d\ge0$ on the cell},\\[1mm]
  \max\{1+h-d,\overline R_{j,i,-}\},
     &\text{otherwise}.
 \end{cases}
\end{equation}
The second case covers every possible negative-gap state and also
retains the nonnegative-gap bound when the branch is not determined.
These bounds are used directly in whole-cell residual enclosures;
no convexity claim or endpoint-chord interpolation is made for the
minimum in~\eqref{eq:cert-scheme-IV-negative}.

Let $\mathcal G$ be the certified correction term obtained from
Appendix~\ref{app:cert-credit}; if its window conditions cannot be
certified on the cell, set $\mathcal G=0$.

For a fixed right-end value $w$, let $H_{j,i}^{\mathrm{base}}(w)$
be the resulting whole-cell upper bound for the residual
in~\eqref{eq:cert-HJB}.
A second bound retains the common dependence
of $d$ and $\nu$ on the masses.  Bounding them separately can combine
extreme values that no single mass vector attains.

Suppose
\[
 [r_j,r_{j+1}]\subseteq[b_p,b_{p+1}]
\]
and the current state band is
$h\in[\underline h_{j,i},\overline h_{j,i}]$.
The path constraints imply
\begin{equation}\label{eq:cert-scheme-V-mass}
 \underline h_{j,i}-(r_{j+1}-b_p)
 \le\sum_{g<p}M_g
 \le\overline h_{j,i},
 \qquad
 \underline h_{j,i}
 \le\sum_{g\le p}M_g
 \le\overline h_{j,i}+b_{p+1}-r_j.
\end{equation}
These constraints couple the ranges of $d$, $\nu$, and $-\nu$.

When the nonnegative-gap branch is certified throughout the cell,
the common-mass expansion in
Appendix~\ref{app:cert-local-tightening} gives a second upper bound
$\widetilde H_{j,i}(w)$ for the same Bellman residual.
This candidate is omitted when the nonnegative-gap condition is not
certified throughout the cell.
The other refinement bounds variation in $r$: for the center term $F_0$
of that expansion,~\eqref{eq:cert-time-derivative} propagates a certified
value at an interior point over the entire coordinate cell.
Combining this enclosure with the other residual terms gives
$\widehat H_{j,i}(w)$, as detailed below.  The package calls this refinement
\emph{time tightening}, with time referring to $r$.

At each step, evaluate all selected candidates at $w=W_{j+1}$.
The record specifies the basic bound alone, the basic and common-mass
bounds, or all three bounds.  Let $\mathcal H_{j,i}(w)$ be that set,
omitting any refinement whose conditions have not been certified.
It always contains $H_{j,i}^{\mathrm{base}}(w)$.  Set
\begin{equation}\label{eq:cert-scheme-residual}
 H_j=\max_i\min\mathcal H_{j,i}(W_{j+1}).
\end{equation}
Every candidate bounds the same residual on the whole cell.  The minimum
therefore selects a valid bound, and the maximum covers every admissible
state.

Any certified upper bound on the residual over the whole state range
may be used in~\eqref{eq:cert-potential-update}.
During adaptive subdivision, let $H$ be the largest retained upper
bound and let $\tau$ be the stored binary64 value of $2\cdot10^{-7}$.
A state interval need not be subdivided further if its certified bound
is at most the upward-rounded value of $H+\tau$.
Since $H$ never decreases, the final upward-rounded value of $H+\tau$
bounds every retained or unrefined interval.  Thus the tolerance is
included in the residual bound used to obtain $W_j$.
The resulting regional bound is
\begin{equation}\label{eq:cert-bellman-scheme-upper}
 U(\mathcal R)=\Gamma_1+T_0(0)+W_0-P(0,0).
\end{equation}

\subsubsection{Joint residual bounds}
\label{app:cert-local-tightening}

On
$[r_j,r_{j+1}]\subseteq[b_p,b_{p+1}]$ and
$h\in[\underline h_{j,i},\overline h_{j,i}]$,
rearranging~\eqref{eq:cert-path-cone-exact} gives
\eqref{eq:cert-scheme-V-mass}.
To bound a linear mass expression, intersect its mass box with either
prefix-sum constraint.  After writing $M_g=\underline M_g+v_g$, the
remaining minimization has individual capacities and an interval constraint
on $\sum v_g$.  Ordering the coefficients and filling lower-cost coordinates
first solves this fractional knapsack problem.  Negative coefficients use
beneficial mass up to the upper total limit; nonnegative coefficients use
only the mass needed by the lower total limit.
This follows by exchanging mass from a higher-cost to a lower-cost
coordinate.  Enlarging the feasible mass ranges preserves a lower bound.
Applying the calculation separately to $d,\nu,-\nu$ strengthens their bounds.

\subsubsection{Common-mass quadratic expansion}
On a cell with $\Delta_\lambda\ge0$, equation~\eqref{eq:cert-positive-basic}
gives $R+\Xi_+=1+h-d$.
For a fixed lower bound $V$ on the continuation term $T_0+W-P$, define
\[
 F(h,d,\nu)=\frac{r(h-d)-(h+\nu)(1+h-d)}{D^2}-\frac{(h+\nu)V}{D}.
\]
The replacement by $V$ is valid because $k=h+\nu\ge0$.
Choose centers $h_0,M_g^0$, and put
\[
 d_0=\sum_gc_gM_g^0,\quad \nu_0=\sum_gw_gM_g^0,\quad
 k_0=h_0+\nu_0,\quad U_0=1+h_0-d_0,
\]
\[
 \delta_d=\sum_gc_g(M_g-M_g^0),\quad
 \delta_\nu=\sum_gw_g(M_g-M_g^0).
\]
The exact expansion is
\begin{align}\label{eq:cert-quadratic}
 F(h,d,\nu)={}&F(h_0,d_0,\nu_0)+\gamma_h(h-h_0)
                         +\sum_g\gamma_g(M_g-M_g^0)\\
 &+\frac{-(h-h_0)^2+(h-h_0)(\delta_d-\delta_\nu)
                       +\delta_d\delta_\nu}{D^2},\notag
\end{align}
where
\begin{equation}\label{eq:cert-quadratic-coefficients}
 \gamma_h=\frac{r-1-2h_0+d_0-\nu_0}{D^2}-\frac VD,\qquad
 \gamma_g=-\frac{(r-k_0)c_g+U_0w_g}{D^2}-\frac{w_gV}{D}.
\end{equation}
Substitution and expansion of the numerator prove the identity.
Both $\delta_d$ and $\delta_\nu$ use the same mass increments.
The cross terms are bounded on the whole cell, and only the nonpositive
term $-(h-h_0)^2/D^2$ is discarded.
For a fixed right-end continuation value $w$, combining this expansion
with the corresponding lower bound on $T_0+W-P$ and with the terms
$-P_r-\min\{\mathcal G,P_h\}$ gives a whole-cell residual bound
$\widetilde H_{j,i}(w)$.
This bound is used only when $\Delta_\lambda\ge0$ is certified throughout
the cell; otherwise it is omitted.

\subsubsection{Variation of the center term in the exponential coordinate}
Inside one output cell, $\nu_0$ is constant, $d_0'=d_0+\nu_0$, and $D'=D-r$.
With $h_0,M^0,V$ fixed, write
\[
 N_0=r(h_0-d_0)-k_0(1+h_0-d_0),\qquad F_0=N_0/D^2-k_0V/D.
\]
Differentiating gives
\begin{equation}\label{eq:cert-time-derivative}
 F_0'=\frac{h_0-d_0+(k_0-r)(d_0+\nu_0)
                   -2N_0(D-r)/D+k_0V(D-r)}{D^2}.
\end{equation}
If an interval $\mathcal J$ encloses this derivative throughout $[a,b]$, then for any
$r_c\in[a,b]$,
\begin{equation}\label{eq:cert-midpoint}
 F_0(r)\in F_0(r_c)+\mathcal J\cdot[a-r_c,b-r_c]\qquad(a\le r\le b).
\end{equation}
The mean value theorem proves the inclusion.
For a fixed right-end continuation value $w$, combine the enclosure
\eqref{eq:cert-midpoint} with the retained terms of
\eqref{eq:cert-quadratic}, the potential derivatives, and the correction term.  This gives the whole-cell residual bound
$\widehat H_{j,i}(w)$ used in the Bellman comparison.
The quadratic expansion and derivative in $r$ can each be checked by clearing
powers of $D$ and comparing integer polynomial coefficients.

\subsection{Construction of the State Potential}
\label{app:cert-potential-construction}

The residual in~\eqref{eq:cert-HJB} contains
$-\min\{\mathcal G,P_h\}$.  This motivates choosing $P_h$ close to the
favorable correction $\mathcal G$.  For each region and resolution,
the construction below integrates a candidate derivative on a grid and
interpolates the resulting vertices.  These vertices define the actual
potential checked in the residual; no accuracy assumption about the
candidate is needed.

Fix a parameter region $\mathcal R$ and the mass and state bounds
from~\eqref{eq:cert-mass-ranges}--\eqref{eq:cert-path-cone}.
Choose positive integers $N_r,N_h$, with $N_r$ divisible by $5$, and put
\[
 \delta_r=1/N_r,\qquad \delta_h=1/N_h,\qquad
 r_j=j\delta_r,\qquad h_i=i\delta_h.
\]
The grid contains all output-cell boundaries, $b_4=32/5$, and the
terminal coordinate $10$; the potential is needed only on $[0,b_4]$.

\subsubsection{Choosing the state derivative}

To choose one potential for the whole region, use the midpoint masses
\[
 M_g^{\mathrm{mid}}:=\frac{\underline M_g+\overline M_g}{2},
 \qquad M_{<p}^{\mathrm{mid}}:=\sum_{g<p}M_g^{\mathrm{mid}}.
\]
For $b_p\le r\le b_{p+1}$, define the reference path
\begin{equation}\label{eq:cert-potential-anchor}
 h^{\mathrm{mid}}(r)
 :=\max\{M_{<p}^{\mathrm{mid}},M_{<p+1}^{\mathrm{mid}}+r-b_{p+1}\}.
\end{equation}
This is the lower feasible path for the midpoint masses, by
\eqref{eq:cert-path-cone-exact}, and $h^{\mathrm{mid}}(b_4)=M_{<4}^{\mathrm{mid}}$.
Superscripts $\mathrm{mid}$ denote midpoint quantities, not certified
upper or lower bounds.

Let $w_{gm}$ be the value of $w_g$ on $J_m$, and set
\[
 \nu_m^{\mathrm{mid}}:=\sum_{g=0}^3w_{gm}M_g^{\mathrm{mid}},\qquad
 \zeta_m^{\mathrm{mid}}:=\max\{0,\nu_0^{\mathrm{mid}},\ldots,
                                     \nu_m^{\mathrm{mid}}\}.
\]
Set $\zeta_{25}^{\mathrm{mid}}=\zeta_{24}^{\mathrm{mid}}$.
For $r=m\Delta+\theta\Delta$, $0\le\theta\le1$, use
\begin{equation}\label{eq:cert-potential-zeta}
 \zeta^{\mathrm{mid}}(r)
 :=(1-\theta)\zeta_m^{\mathrm{mid}}+\theta\zeta_{m+1}^{\mathrm{mid}}.
\end{equation}
Here $\Delta=2/5$ as in Section~\ref{sec:final-kernel}.
The other inputs to the candidate derivative are
\begin{align}
 d^{\mathrm{mid}}(r)&:=\sum_{g=0}^3c_g(r)M_g^{\mathrm{mid}},
 \label{eq:cert-potential-midpoint-d}\\
 \zeta_{\mathcal R}^-&:=\max_{0\le m<25}
 \left[-\min_{M_g\in[\underline M_g,\overline M_g]}
                     \sum_{g=0}^3w_{gm}M_g\right]_+.
 \label{eq:cert-potential-negative-z}
\end{align}
The minimum is a linear expression over a box and is evaluated by
choosing the appropriate endpoint of each mass interval.

The candidate follows the past-window construction in
Appendix~\ref{app:cert-credit}, with the midpoint quantities in place of certified bounds.  To evaluate $e^h-1-h$ at small $h$, use
\begin{equation}\label{eq:cert-potential-Ehat}
 \widehat E(h):=
 \begin{cases}
 \displaystyle\sum_{j=2}^{6}\frac{h^j}{j!},&0\le h<10^{-2},\\[2mm]
 e^h-1-h,&h\ge10^{-2}.
 \end{cases}
\end{equation}
For $(r,h)$, put
\[
 \widehat\Delta:=\widehat E(h)+d^{\mathrm{mid}}(r),\qquad
 \bar k:=h+\zeta^{\mathrm{mid}}(r),\qquad \chi_*:=(r-\bar k)_+.
\]
If any of $\widehat\Delta,\bar k,\chi_*$ is nonpositive, set
$\widehat{\mathcal G}_{\mathcal R}(r,h)=0$.
Otherwise choose the window length
\begin{equation}\label{eq:cert-potential-length}
 \eta:=\min\left\{\log\left(1+\frac{\widehat\Delta}{\bar k}\right),\chi_*\right\},
\end{equation}
and define
\begin{equation}\label{eq:cert-potential-beta0}
 \beta_0(r,h):=
 \begin{cases}
 \displaystyle\max\left\{\frac{509}{2000},
                1-\frac{2(r-\eta)}{D(r-\eta)}\right\},&r-\eta\ge1.7,\\[3mm]
 \displaystyle\max\left\{\frac{509}{2000},
                1-\frac{2r}{D(r)}\right\},&r-\eta<1.7,\ r\le1.5,\\[3mm]
 \displaystyle\frac{509}{2000},&\text{otherwise}.
 \end{cases}
\end{equation}
Here $509/2000=0.2545$ lies slightly below the global lower bound $\beta_0$ of Appendix~\ref{app:cert-credit}, so it is also a lower bound for $1-2r/D(r)$; the other cases use that $r/D(r)$ increases on $[0,1.5]$ and decreases on $[1.7,\infty)$.
The derivative also incorporates the survival factor through
\begin{equation}\label{eq:cert-potential-beta}
 \beta:=\beta_0(r,h)
 +\frac{\left[\max\{h-\eta/2,h^-_{\mathcal R}(r-\eta)\}
                       -\zeta_{\mathcal R}^-\right]_+}{D(r)}.
\end{equation}
The candidate derivative is then
\begin{equation}\label{eq:cert-potential-Ghat}
 \widehat{\mathcal G}_{\mathcal R}(r,h)
 :=\frac1{D(r)^2}\sum_{j=0}^{6}\frac{\beta^j}{j!}
 \left[\frac{\chi_*\eta^{j+1}}{j+1}-\frac{\eta^{j+2}}{j+2}\right].
\end{equation}
This midpoint construction chooses the potential.  The correction
$\mathcal G$ is bounded independently on the entire cell by
Appendix~\ref{app:cert-credit}.

\subsubsection{From the derivative to the potential}

At each node $r_j$, retain a state-grid interval containing the
admissible ranges of the adjacent coordinate cells and the reference
value $h^{\mathrm{mid}}(r_j)$.  The generator uses two additional grid
points on each side, with indices
\begin{align*}
 i_j^-&=\max\left\{0,
 \left\lfloor N_h\min\{h^-_{\mathcal R}((r_j-\delta_r)_+),
                          h^{\mathrm{mid}}(r_j)\}\right\rfloor-2\right\},\\
 i_j^+&=\left\lceil N_h\max\{h^+_{\mathcal R}(r_j+\delta_r),
                              h^{\mathrm{mid}}(r_j)\}\right\rceil+2.
\end{align*}
For $r_j+\delta_r>b_4$, use the continuation
$h^+_{\mathcal R}(r)=\overline M_{<4}+r-b_4$
from~\eqref{eq:cert-tail-state}.  This padding ensures that all vertices
needed for interpolation are present.

Let $g_{j,i}:=\widehat{\mathcal G}_{\mathcal R}(r_j,h_i)$.
Integrating the candidate derivative in $h$ by a trapezoidal sum gives
\begin{equation}\label{eq:cert-potential-trapezoid}
 Q_{j,i_j^-}:=0,\qquad
 Q_{j,i}:=Q_{j,i-1}
       +\frac{\delta_h}{2}(g_{j,i-1}+g_{j,i})
 \quad(i_j^-<i\le i_j^+).
\end{equation}
Next choose the additive constant so that the interpolated primitive
vanishes at the reference path.  Write
\[
 h^{\mathrm{mid}}(r_j)
 =(1-\alpha_j)h_{\iota_j}+\alpha_jh_{\iota_j+1},
 \qquad 0\le\alpha_j<1,
\]
and set
\begin{equation}\label{eq:cert-potential-offset}
 Q_j^{\mathrm{mid}}:=(1-\alpha_j)Q_{j,\iota_j}
                    +\alpha_jQ_{j,\iota_j+1},\qquad
 \widetilde P_{j,i}:=Q_{j,i}-Q_j^{\mathrm{mid}}.
\end{equation}
Finally enforce monotonicity in the state index by
\begin{equation}\label{eq:cert-potential-vertices}
 P_{j,i}:=\max_{i_j^-\le\iota\le i}\widetilde P_{j,\iota}.
\end{equation}
Since the reference path starts at zero, $P_{0,0}=0$.

On $[r_j,r_{j+1}]\times[h_i,h_{i+1}]$, let
$\theta_r=(r-r_j)/\delta_r$ and $\theta_h=(h-h_i)/\delta_h$.
Define $P$ by bilinear interpolation,
\begin{equation}\label{eq:cert-potential-bilinear}
\begin{aligned}
 P(r,h)={}&(1-\theta_r)(1-\theta_h)P_{j,i}
 +(1-\theta_r)\theta_h P_{j,i+1}\\
 &+\theta_r(1-\theta_h)P_{j+1,i}
 +\theta_r\theta_h P_{j+1,i+1}.
\end{aligned}
\end{equation}
In particular,
\begin{equation}\label{eq:cert-potential-origin}
 P(0,0)=0.
\end{equation}
Inside each cell the derivatives are
\begin{align}
 P_h(r,h)&=(1-\theta_r)\frac{P_{j,i+1}-P_{j,i}}{\delta_h}
          +\theta_r\frac{P_{j+1,i+1}-P_{j+1,i}}{\delta_h},
 \label{eq:cert-potential-Ph}\\
 P_r(r,h)&=(1-\theta_h)\frac{P_{j+1,i}-P_{j,i}}{\delta_r}
          +\theta_h\frac{P_{j+1,i+1}-P_{j,i+1}}{\delta_r}.
 \label{eq:cert-potential-Pr}
\end{align}
Thus the vertex values explicitly determine $P,P_h,P_r$ wherever
the Bellman comparison needs them, with one-sided derivatives at
cell boundaries.

The numerical generator evaluates this construction in binary64 arithmetic,
using outward-rounded preliminary mass and state bounds.  The resulting
finite vertex values are treated as exact dyadic rationals during
verification.  Outward interval bounds for
\eqref{eq:cert-potential-bilinear}--\eqref{eq:cert-potential-Pr}, together
with the independently certified $\mathcal G$, verify the Bellman
inequality on every admissible cell.  This checks the actual generated
potential, including the effects of its construction and rounding.

The package supplies the deterministic vertex generator with every
regional input and resolution.  It regenerates the vertices during
verification, so separate potential tables are unnecessary.  Floating-point
evaluation can change a vertex, but acceptance always depends on rigorous
checks of the actual interpolant and its derivatives.

\subsection{Certified Arithmetic}
\label{app:cert-arithmetic}

The regional inequalities above must hold for real parameters throughout
each cell.  The programs therefore evaluate their bounds by interval
arithmetic, including the quadrature errors in
Appendix~\ref{app:cert-quadrature} and the finite-tail estimate in
\eqref{eq:cert-tail}.  This subsection specifies the arithmetic used to
turn those inequalities into finite numerical certificates.
All numerical coordinate grids have step $1/N_r$, where $N_r$ is a
positive integer divisible by five, so they include every output-cell
and source-block boundary.

\subsubsection{Outward interval evaluation}

All interval endpoints are finite IEEE-754 binary64 numbers, interpreted
as exact dyadic rationals.  The programs use round-to-nearest arithmetic
with gradual underflow, without excess precision, reassociation, or
fused multiply-add contraction.  After each basic arithmetic operation,
the lower endpoint is moved to the adjacent representable number below
the computed value and the upper endpoint to the adjacent number above.
Operations known to be exact need no enlargement.  This encloses the exact
result even when the rounded result is subnormal.  Products use all four
endpoint products; division is allowed only after the denominator has
been certified strictly positive.  Every interval is checked to have
finite, ordered endpoints.  Failure of any check rejects the computation.

Rational constants and grid coordinates are enclosed by the same
operations.  The loaded kernel intervals are checked against the exact
integer specification in Table~\ref{tab:cert-kernel}; every interval
used for the target ratio encloses the exact rational
$\Gamma_1$.  The execution wrapper checks the rounding environment
before and after every regional call, and verifies that the numerical
engines loaded the certified coefficient endpoints without alteration.

\subsubsection{Elementary functions}

For $0\le x<98$, the positive Taylor series gives
\begin{equation}\label{eq:cert-exp-arithmetic}
 \sum_{j=0}^{96}\frac{x^j}{j!}
 \le e^x
 \le \sum_{j=0}^{96}\frac{x^j}{j!}
       +\frac{x^{97}}{97!}\frac1{1-x/98}.
\end{equation}
Indeed, after the first omitted term, the ratio of consecutive terms
is at most $x/98$.  The programs evaluate these expressions with
outward rounding and check the stated domain.  Monotonicity gives the
bounds for an interval of arguments.  Negative arguments are handled by
$e^{-x}=1/e^x$, and an interval crossing zero is bounded at its endpoints.
Thus the certified exponential bounds do not rely on a library
approximation to $e^x$.

The past-window estimate requires only a lower bound on $\log(1+x)$
for $x\ge0$.  First round $1+x$ downward to $v$.  If $v\le1$, use zero.
Otherwise write $v=2^m y$, where $m\ge0$ is an integer and $1\le y<2$,
and put $z=(y-1)/(y+1)$.  The inequality
\begin{equation}\label{eq:cert-log-arithmetic}
 2\sum_{j=0}^{19}\frac{z^{2j+1}}{2j+1}
 \le\log\frac{1+z}{1-z},\qquad 0\le z\le\frac13,
\end{equation}
follows by dropping the remaining positive terms.
Applied to $z$ and to $1/3$, it gives lower bounds on $\log y$ and
$\log2$.  Their downward-rounded combination bounds
$\log v=m\log2+\log y\le\log(1+x)$ from below.

When an upper bound on $\sqrt{x}$ is needed, a library evaluation is
used only to choose a starting value.  The program increases that value
until the lower interval endpoint of its square is at least $x$.
The resulting nonnegative value is therefore a certified upper bound;
zero is handled exactly.  Invalid inputs or failure of this check
reject the computation.  Ordinary floating-point evaluations used to
generate the candidate potential do not need such enclosures, since
the resulting vertices are subsequently treated as exact values and
their inequalities are checked independently, as described in
Appendix~\ref{app:cert-potential-construction}.

\section{Extension to Distributions with Atoms}\label{app:cert-atoms}

We prove the extension stated in
Proposition~\ref{prop:scalar-feasibility}.
Fix any target ratio $\Gamma\in(0,1)$ and allocation kernel $K$
satisfying~\eqref{eq:robust}.
We show that the corresponding competitive guarantee extends to
arbitrary independent nonnegative distributions, including those
with atoms, via randomized tie-breaking.

Let $F_{\mathrm{mix}}(v):=n^{-1}\sum_iF_i(v)$ and let $Z_i$ be
independent uniform random variables, independent of the values.
Define scores
\begin{equation}\label{eq:cert-distributional-transform}
 X_i:=F_{\mathrm{mix}}(V_i^-)
 +Z_i\bigl(F_{\mathrm{mix}}(V_i)-F_{\mathrm{mix}}(V_i^-)\bigr).
\end{equation}
The scores are independent and atomless.  Indeed, if $\widetilde I$ is a uniformly
chosen independent index, then $V_{\widetilde I}$ has distribution function $F_{\mathrm{mix}}$, and the randomized probability integral transform
gives $X_{\widetilde I}\sim\mathrm{Unif}[0,1]$.  Hence, for every Borel set $E$,
\[
 \frac1n\sum_i\Pp(X_i\in E)=|E|,
\]
where $|E|$ denotes its Lebesgue measure in $[0,1]$.
Each score law is therefore absolutely continuous, with density at most $n$.

The transformation is nondecreasing in $V_i$ and orders tied values by
$Z_i$.  Because independent atomless scores are distinct almost surely,
their ordering agrees almost surely with this tie-broken ordering of
the original values.  Moreover,
\[
 V_i=F_{\mathrm{mix}}^{-1}(X_i)\qquad\text{almost surely},
\]
where $F_{\mathrm{mix}}^{-1}$ is the usual lower quantile.
Apply the arrival construction to the scores, using their known
atomless distributions and the kernel $K$.
This is a valid order-selection algorithm for the original problem:
its inspection order is chosen independently of the values, and the
score of an inspected item can be computed from its value and its
independent tie-breaking variable.

Let $I^*$ denote the index of the largest score, so that $V_{I^*}=\OPT$
almost surely.  The exact item-wise guarantee for the scores implies,
for every Borel set $E$,
\[
 \Pp(I=i,\,X_i\in E)
 =\Gamma\Pp(I^*=i,\,X_i\in E).
\]
Taking $E=\{x:F_{\mathrm{mix}}^{-1}(x)>v\}$ and summing over $i$ gives
\[
 \Pp(\ALG>v)=\Gamma \Pp(\OPT>v)\qquad(v\ge0).
\]
Integration over $v$ proves
$\E[\ALG]=\Gamma\E[\OPT]$ whenever the latter expectation is finite,
and in particular the claimed competitive guarantee.

\end{document}